\documentclass[10pt,reqno]{amsart}
\usepackage{graphicx, amsmath, amsthm, amsfonts, enumerate, amssymb, mathrsfs, color}
\usepackage{fullpage}
\usepackage{bbm}

\def\dtau{\partial^{(\tau)}}

\def\widehatdntau{\widehat{\partial}^{(\tau)}_n}

\def\soft{\text{\rm soft}}
\def\stiff{\text{\rm stiff}}
\def\rr{\text{\rm right}}
\def\ll{\text{\rm left}}
\def\eff{\text{\rm eff}}
\def\hom{\text{\rm hom}}
\def\xit{\xi^{(\tau)}}
\def\bxit{\overline{{\xi}^{(\tau)}}}
\def\bxitone{\overline{{\xi}^{(\tau)}_1}}
\def\bxittwo{\overline{{\xi}^{(\tau)}_2}}
\def\psit{\psi^{(\tau)}}

\def\Port{\mathcal P_{\bot}^{(\tau)}}
\def\P{\mathcal P^{(\tau)}}

\def\e{\varepsilon}

\DeclareMathOperator*{\dom}{\mathrm{dom}}
\DeclareMathOperator*{\clos}{\mathrm{clos}}
\DeclareMathOperator*{\card}{\mathrm{card}}
\DeclareMathOperator*{\diag}{\mathrm{diag}}
\DeclareMathOperator*{\ran}{\mathrm{ran}}

\newtheorem{theorem}{Theorem}[section]
\newtheorem{proposition}[theorem]{Proposition}

\newtheorem{corollary}[theorem]{Corollary}

\newtheorem{definition}[theorem]{Definition}
\newtheorem{lemma}[theorem]{Lemma}
\newtheorem{prop}[theorem]{Proposition}

\newtheorem{remark}[theorem]{Remark}

\usepackage[pdftex]{hyperref}

\begin{document}
\title{Unified approach to critical-contrast homogenisation with explicit links to time-dispersive media}
\author{Kirill D. Cherednichenko}
\address{Department of Mathematical Sciences, University of Bath, Claverton Down, Bath, BA2 7AY, United Kingdom}
\email{cherednichenkokd@gmail.com}
\author{Yulia Yu. Ershova}
\address{Department of Mathematical Sciences, University of Bath, Claverton Down, Bath, BA2 7AY, United Kingdom  {\sc and} Department of Mathematics, St.\,Petersburg State University of Architecture and Civil Engineering, 2-ya Krasnoarmeiskaya St. 4, 190005 St.\,Petersburg, Russia}
\email{julija.ershova@gmail.com}
\author{Alexander V. Kiselev}
\address{Departamento de F\'{i}sica Matem\'{a}tica, Instituto de Investigaciones en Matem\'aticas Aplicadas y en Sistemas, Universidad Nacional Aut\'onoma de M\'exico, C.P. 04510, M\'exico D.F. {\sc and} International research laboratory ``Multiscale Model Reduction'', Ammosov North-Eastern Federal University, Yakutsk, Russia}
\email{alexander.v.kiselev@gmail.com}
\author{Sergey N. Naboko}
\address{Department of Higher Mathematics and Mathematical Physics, St.\,Petersburg State University, Ulianovskaya~1, St.Peterhoff, St.\,Petersburg, 198504, Russia {\sc and} Department of Mathematics, Stockholm University, Sweden}
\email{sergey.naboko@gmail.com}


\begin{abstract}
A novel approach to critical-contrast homogenisation is proposed.
Norm-resolvent asymptotics are explicitly constructed. An essential feature of our approach is that it relates homogenisation limits to a class of time-dispersive media.
\end{abstract}

\maketitle

\par{\raggedleft\slshape To Professor Andrei Shkalikov with deepest respect\par}

\section{Introduction}

The research aimed at modelling and engineering metamaterials has been recently brought to the forefront of materials science (see, e.g., \cite{Phys_book} and references therein). It is widely acknowledged that these novel materials acquire non-classical properties as a result of a carefully designed microstructure of a composite medium, which can be assumed periodic with a small enough periodicity cell. The mathematical apparatus involved in their modelling must therefore include as its backbone the theory of homogenisation, see
{\it e.g.} \cite{Lions,Jikov_book}, which aims at characterising limiting, or ``effective'', properties of small-period composites. A typical problem here is to study the asymptotic behaviour of solutions to equations of the type
\begin{equation}
-{\rm div}\bigl(A^\varepsilon(\cdot/\varepsilon)\nabla u_\e\bigr)-\omega^2u_\e=f,\ \ \ \ f\in L^2({\mathbb R}^d),\qquad \omega^2\notin{\mathbb R}_+,
\label{eq:generic_hom}
\end{equation}
where for all $\varepsilon>0$ the matrix $A^\varepsilon$ is $Q$-periodic, $Q:=[0,1)^d,$ non-negative and may additionally be required to satisfy the condition of uniform ellipticity.

On the other hand, the result sought ({\it i.e.}, the metamaterial behaviour in the limit of vanishing $\e$) belongs to the domain of the so-called time-dispersive media (see, e.g., \cite{Tip_1998,Figotin_Schenker_2005,Tip_2006,Figotin_Schenker_2007b}). For such media, in the frequency domain one faces a setup of the type
\begin{equation*}
-{\rm div}\bigl(A\nabla u\bigr)+\mathfrak B(\omega)u=f,\ \ \ \ f\in L^2({\mathbb R}^d),\
\end{equation*}
where $A$ is a constant matrix and $\mathfrak B(\omega)$ is a frequency-dependent operator in $L^2({\mathbb R}^d)$ taking the place of $-\omega^2$ in (\ref{eq:generic_hom}) if, for the sake of argument, in the time domain we started with an equation of second order in time as in (\ref{eq:generic_hom}). If in addition $\mathfrak B$ is a scalar matrix function, {\it i.e.}, $\mathfrak B(\omega)=\beta(\omega) I$ with a scalar function $\beta(\omega)$, then a problem of the type
\begin{equation}
-{\rm div}\bigl(A(\omega)\nabla u\bigr)=\omega^2 u
\label{eq:generic_spectral_td}
\end{equation}
appears in place of the spectral problem after a formal division by $-\beta(\omega)/\omega^2$, with frequency-dependent (but independent of the spatial variable) matrix $A(\omega)$.

Thus, the matrix elements of $A(\omega)$, interpreted as material parameters of the medium, acquire non-trivial dependence on the frequency, which may lead to their taking negative values in certain  frequency intervals. The latter property is in turn characteristic of metamaterials \cite{Veselago}. It is therefore of a paramount interest to understand how inhomogeneity in the spatial variable (see \eqref{eq:generic_hom}) can lead in the limit of vanishing $\e$ to the inhomogeneity in the frequency variable, and in particular to uncover the conditions on $A^\e$ sufficient for this. A result, which from the above perspective can be seen as negative, is provided by the homogenisation theory in the uniformly strongly elliptic setting ({\it i.e.}, both $A^\e$ and $(A^\e)^{-1}$ are uniformly bounded). Here one proves (see \cite{Zhikov_1989,BirmanSuslina} and references therein) the existence of a constant $A^\hom$ matrix such that solutions $u_\e$ to \eqref{eq:generic_hom} converge to those of
\begin{equation*}
-{\rm div}\bigl(A^\hom \nabla u_\hom\bigr)-\omega^2u_\hom=f,
\end{equation*}
which leaves no room for time dispersion. This result also carries over to vector models, including the Maxwell system.

If the uniform ellipticity assumption is dropped, the analysis of (\ref{eq:generic_hom}) becomes much more complicated, which is to be expected (see, e.g., \cite{Shkalikov}). By employing the technique of two-scale convergence, first Zhikov \cite{Zhikov2000,Zhikov2005} and then Bouchitte and Felbacq \cite{BouchitteFelbacq} obtained an effective problem of the type \eqref{eq:generic_spectral_td}. The former works treat the critical-contrast model of the type \eqref{eq:generic_hom}, while the latter is devoted to an associated scattering problem. Here, under critical contrast one means that the components of the medium must have material properties in a proper contrast to each other, governed by the size of periodicity cell (see Section 2 for further details). In a closely related development, Hempel and Lienau \cite{HempelLienau_2000} proved the band-gap structure of the limiting spectrum (see also \cite{Friedlander} for a result concerning the asymptotics of the integrated density of states). Although well received, these results fall short of establishing a rigorous one-to-one correspondence between homogenisation limits for
critical-contrast media and time dispersion in the effective medium. This is due to the following: (i) the additional assumptions imposed only permit to treat a limited set of models (curiously, excluding even the one-dimensional version of the problem, let alone physically relevant cases like, {\it e.g.}, ``split-ring resonator'' type inclusions); (ii) the error control is lacking, due to the very weak convergence of solutions claimed. A more general theory, akin to that of Birman and Suslina \cite{BirmanSuslina,BirmanSuslina_corr} in the
moderate-contrast case, is therefore required. The present paper attempts to suggest precisely this.

The benefits of the novel unified approach as developed henceforth are these, in a nutshell:
\begin{enumerate}
  \item Being free from additional assumptions on the geometry and PDE type, it can be successfully applied in a consistent way to diverse problems motivated by applications;
  \item It can be viewed as a natural (albeit non-trivial) generalisation of the Birman and Suslina approach in the uniformly elliptic case;
  \item The analysis is shown to be reducible by purely analytical means to an auxiliary uniformly elliptic problem; the latter, unlike the original problem, is within the reach of robust numerical techniques;
  \item The error bounds are controlled uniformly via norm-resolvent estimates (yielding the spectral convergence as a by-product);
  \item Not only the relation of the composite with the corresponding effective time-dispersive medium is made transparent (showing the introduction of second (``fast'') variable via the two-scale asymptotics to be unnecessary), but the approach can be also seen to offer a recipe for the construction of such media with prescribed dispersive properties from periodic composites whose individual components are non-dispersive.
\end{enumerate}

The analytical toolbox we propose also allows\footnote{This argument will appear in a separate publication.}: (i) to explicitly construct spectral representations and functional models for both homogenisation limits of critical-contrast composites and the related time-dispersive models; (ii) on this basis, to solve direct and inverse scattering problems in both setups. It therefore paves the way to treating the inverse problem of constructing a metamaterial ``on demand'', based on its desired properties.

In the present paper we consider a model of a high-contrast graph periodic along one axis (see Section 2 for the setup and notation). It is instructive from the point of view of applications to think of this graph as being embedded into $\mathbb R^d$ for some $d\ge 1$. Indeed, by \cite{KuchmentZeng,Exner,KuchmentZeng2004}, see also \cite{Zhikov_singular_structures}, it can be viewed as an idealised model of a thin periodic critical-contrast network. The named setup (see also our paper \cite{Physics}, where its asymptotic behaviour is studied in terms of eigenfunctions; the approach of the named paper can be viewed as an alternative manifestation of the toolbox developed here) allows to keep technicalities to a bare minimum, at the same time making the substance of the argument highly transparent. Having said that, we remark, that the main ingredients of the theory remain more or less the same when one passes over to the PDE setup, see our paper \cite{PDE_paper}.

The analytical backbone of our approach is the so-called generalised resolvent, or in other words the resolvent of the original operator family sandwiched by orthogonal projections to one of components of the medium (``soft'' one, see Section 2 for details). In its analysis, we rely upon the celebrated general theory due to Neumark \cite{Naimark1940,Naimark1943} and the follow-up work by Strauss \cite{Strauss}. An explicit analysis of Dirichlet-to-Neumann maps (separately for the components comprising the medium) becomes necessary to facilitate the use of the well-known Kre\u\i n resolvent formula. The corresponding analysis is based on a version of Birman-Kre\u\i n-Vi\v sik theory \cite{Birman,Krein,Vishik}, treating self-adjoint extensions of symmetric operators.

The use of generalised resolvents is a rather common place in operator and spectral theory, notable examples ranging from scattering theory \cite{Yafaev} to the celebrated Birman-Schwinger principle \cite{Simon}. An essential part of the Birman-Suslina theory on moderate contrast homogenisation can also be viewed as
an application (albeit, degenerate) of the same. The present paper can therefore be seen as yet another example of how surprisingly far one can reach by a consistent application of the existing vast toolbox of abstract spectral theory.

\section{Periodic graph setup}

\label{setup_section}

Let $B_\e$ be a strip in $\mathbb R^d$, $d\ge 2,$
$
B_\e=\{y\in \mathbb R^d : y_1\in [ 0, c_0 \e) \},
$
where $y_1$ is the first component of the $d$-dimensional vector $y$ and $c_0$ is a real number, $c_0>0$. Introducing the translation vector $\vec \ell=(c_0 \e,0,\dots,0)^\top\in \mathbb R^d$, one represents the whole space as a periodic structure with respect to a lattice of hyperplanes, i.e., $\mathbb R^d=\cup_{j\in \mathbb Z} (B_\e+j\vec \ell)\equiv \cup_{j\in \mathbb Z} B_\e^{(j)}$, where the mentioned lattice is given by $\mathbb T=\cup_{j\in\mathbb Z} T_j:=\cup_{j\in\mathbb Z} (\{y\in \mathbb R^d : y_1=0 \}+ j \vec\ell)$.

Let $\mathbb G_{\text{per}}$ be a periodic metric graph embedded into $\mathbb R^d$
so that: (i) $\mathbb G_{\text{per}}$ is invariant under $\vec\ell$-translations, $\mathbb G_{\text{per}}=\mathbb G_{\text{per}}+\vec\ell$; (ii) $\mathbb G_\e:=\clos(\mathbb G_{\text{per}}\cap  B_\e)$ is a compact finite metric graph, see {\it e.g.} \cite{Kuchment2}. Such graphs arise in applications as limits of thin periodic networks, as the thickness of the network vanishes.

By introducing the standard parametrisation of each edge of the graph $\mathbb G_{\text{per}}$ (and hence of $\mathbb G_\e$) by a one-dimensional parameter, we treat the graph $\mathbb G_\e$ as a collection of intervals of the real line $e^{(p)}=[x^{(p), \ll},x^{(p), \rr}]$  ($p =1,...,n\equiv n(\mathbb G_\e)$). We denote the set of edges of $\mathbb G_\e$ by $\mathcal E=\mathcal E(\mathbb G_\e)$. We further assume that the total length ({\it i.e.} the sum of the edge lengths $l^{(p)}:=x^{(p), \rr}-x^{(p), \ll}$) of the graph $\mathbb G_\e$ is equal to $\e$. It is convenient to identify the vertices $V_m$, $m=1,\dots, N\equiv N(\mathbb G_\e)$ of the graph $\mathbb G_\e$ with equivalence classes of edge endpoints,
$$
V_m=\{x^{(p_1), \ll},\ldots,x^{(p_{s(m)}), \ll},x^{(p_{s(m)+1}), \rr},\ldots,x^{(p_{\gamma_m}), \rr}\},
$$
where $s(m)$ is the number of left endpoints of the edges comprising the vertex $V_m$ and
$\gamma_m$ is the degree, or valence, often denoted by $\deg V_m\equiv \gamma_m$ of the vertex $V_m$ is defined as the number of elements of $V_m.$
The set of vertices of a graph $\mathbb G$ will henceforth be denoted by $\mathcal V(\mathbb G)$.


For the graph $\mathbb G_\e$ introduced above we single out two natural classes of its vertices,
\[
\mathcal V^\ll:=\{V\in \mathcal V(\mathbb G_\e) : V\in T_0\},\qquad
\mathcal V^\rr:=\{V\in \mathcal V(\mathbb G_\e) : V\in T_1\}.
\]
These represent the vertices of the graph $\mathbb G_\e$ which are located on the left $T_0$ and right $T_1$ ``boundaries'' of the strip $B_\e$, respectively.
We assume throughout that the translation of the set $\mathcal V^\ll$ by the vector $\vec \ell$ intersected with the set $\mathcal V^\rr$ is non-empty (otherwise, the graph $\mathbb G_{\text{per}}$ is clearly disconnected). This intersection periodically extended will be denoted $\mathcal V^{\text{per}}$,
$$
\mathcal V^{\text{per}}=\cup_{j\in \mathbb Z} ((\mathcal V^\ll+\vec \ell)\cap \mathcal V^\rr+j\vec \ell).
$$


With the graph $\mathbb{G}_{\text{per}}$ we associate the Hilbert space $L^2(\mathbb{G}_{\text{per}})$ which is the direct sum of   $L^2$-spaces pertaining to the edges of the graph:
$$
L^2(\mathbb{G}_{\text{per}})=\bigoplus\limits_{j\in \mathbb Z}\bigoplus\limits_{e^{(p)}\in \mathcal{E}(\mathbb G_\e+j\vec \ell)} L^2\bigl(e^{(p)}\bigr),
$$
where $\mathcal{E}(\mathbb G_\e+j\vec \ell)$ is the set of edges of the translated graph $\mathbb G_\e+j\vec \ell$.

Next, we define the second-order high-contrast operator $A_\e$ on the graph $\mathbb G_{\text{per}}$. On each edge $e$ of the graph $\mathbb{G}_{\text{per}}$ the action of the operator $A_\e$ is set by the differential expression
\[
-a_e^2(\e)\frac {d^2} {dx^2},
\]
where $a_e(\e)$ is a positive weight, constant on each edge of the graph. The weight $a_e(\e)$ is further assumed to be periodic with the same periodicity cell as $\mathbb{G}_{\text{per}}$, {\it i.e.}, $a_e(\e)=a_{e+j\vec\ell}(\e)$ for any $j\in \mathbb Z$. We consider the weights $a_e(\e)$ of two classes only. The first class, corresponding to the \emph{stiff} component of the medium modelled by $\mathbb G_{\text{per}}$, corresponds to $a_e(\e)\equiv a_e$, independent of $\e$. The second class considered, corresponding to the \emph{soft} component, is of the form $a_e(\e)=a_e \e$. Correspondingly, the graph $\mathbb G_\e$ admits a decomposition into its soft and stiff components, $\mathbb G_\e=\mathbb G_{\stiff,\e} \cup \mathbb G_{\soft,\e}$, where both $\mathbb G_{\stiff,\e}$ and  $\mathbb G_{\soft,\e}$ are graphs (either connected or otherwise) formed of the edges of the first and second class, respectively.

The domain $\dom A_\e$ of the operator $A_\e$ is described as follows: it consists of all functions from the Sobolev space $W^{2}_2(\mathbb{G}_{\text{per}})$ subject to the so-called Kirchhoff matching conditions at the graph vertices:
\begin{equation}\label{Kirchhoff_periodic}
\dom (A_{\varepsilon})=\biggl\{\left.u\in W^2_2(\mathbb{G}_{\text{per}})\right| \text{ at any vertex } V\ u \text{ is continuous and } \sum\limits_{e\thicksim V} \partial_n u_e(V) =0
 \biggr\},
\end{equation}
where the notation $e \thicksim V$ abbreviates the condition ``the edge $e$ is adjacent to the vertex $V$'' ({\it i.e.} one of its endpoints belongs to $V$) and $u_e(x)= \left. u(x) \right|_{e}$ for $x\in e$.
Here the operator $\partial_n$ computes the inward normal co-derivative at the vertex $V$:
$$
\partial_n u_e(V):=(a_e(\e))^2\begin{cases}
u_e'(V), & V \text{ is the left endpoint of  } e,\\[0.3em]
-u_e'(V), & V \text{ is the right endpoint of  } e,
\end{cases}
$$
with an obvious adjustment of notation when $e$ is a loop. It is easily seen that the operator $A_\e$ thus defined is a bounded below self-adjoint operator in $L^2(\mathbb{G}_{\text{per}})$. Clearly, the Kirchhoff matching condition at a hanging vertex of the graph, i.e., at a vertex $V$ such that $\deg V=1$, reduces to the Neumann boundary condition.

\section{Gelfand transform and auxiliary rescaling}

\label{Gelfand_section}

\subsection{Gelfand transform for a graph periodic along one axis}

It is customary to apply either Floquet or Gelfand transform to a periodic differential operator, in view of obtaining a fibre representation in the form of a direct von Neumann integral over the dual cell of quasimomentum. This allows to reduce the analysis of the original operator $A_\e$ to the one of the operator family $A_\e^{(t)}$ such that at each value of quasimomentum $t$ the operator $A_\e^{(t)}$ has compact resolvent and thus discrete spectrum accumulating to plus infinity.

Since the original graph $\mathbb G_{\text{per}}$ is periodic in precisely one direction, it would be natural to apply the one-dimensional Gelfand transform to it. Regretfully, the graph as defined above is embedded in $\mathbb R^d$ rather than $\mathbb R^1$, necessitating an auxiliary procedure which we will refer to as \emph{flattening}.

The possibility to re-embed the periodic graph $\mathbb G_{\text{per}}$ into $\mathbb R$ arises due to the fact that the Hilbert space associated with it is simply the orthogonal sum of $L^2$ spaces over segments of the real line. The ``geometry'' of the graph is encoded in the matching conditions \emph{only}. In essence, the graph geometry is only based on the \emph{locality} of matching conditions of Kirchhoff type at graph vertices. If one forgets for a moment the customary practice of drawing graphs with locality of matching conditions in mind, one is then free to consider the same graph as a collection of segments of the real line subject to a set of (non-local) conditions intertwining the values of functions and their derivatives at edge endpoints. In view of applicability of the Gelfand transform, it will be convenient for us to arrange the edges of the graph $\mathbb G_\e$ as consecutive segments of the real line\footnote{It is of course clear that the result of flattening thus understood, and therefore ultimately the image of the Gelfand transform \eqref{Gelfand}, will depend on the order in which the graph edges are counted. Due to the unitarity of Gelfand transform this is nevertheless irrelevant as all resulting fibre decompositions are unitary equivalent. For the purposes of the present paper it suffices to fix some particular numbering of edges, which we henceforth assume done.}, starting with zero. The periodicity condition then yields an $\e-$periodic infinite chain graph spanning the space $\mathbb R^1$. By a slight abuse of notation, we will keep the same notation $\mathbb G_\e$, $\mathbb G_{\text{per}}$ for the periodicity cell and the periodic graph, respectively, after the flattening in hope that it will not lead to misunderstanding.


The price paid for flattening the original graph $\mathbb G_{\text{per}}$ is that the Gelfand transform, once applied, yields additional unimodular weights in the  non-local matching conditions.
The Gelfand transform we apply next is defined as
\begin{equation}\label{Gelfand}
\widehat{u} (y, t)\equiv (G u)(y,t)= \sqrt{\frac {\varepsilon }{2\pi}} \sum_{n\in \mathbb{Z}} u(y+\varepsilon  n)e^{-i t (y+\varepsilon  n)},
\end{equation}
which is shown to be a unitary operator from $L^2(\mathbb R)$ to $L^2\bigl((0,\e)\times (-\pi/\e, \pi/\e)\bigr).$ Applied to the original operator family $A_\e$, it yields the fibre decomposition of the latter into the direct von Neumann integral
\begin{equation}\label{vonNeumann}
A_\e\cong\oplus \int_{-\pi/\e}^{\pi/\e}A_\e^{(t)} dt,
\end{equation}
where at each value of $t$ and for each edge $e$ of $\mathbb G_\e$
the action of  $A_\e^{(t)}$ is set by the differential expression
\[
-a^2_e(\e)\biggl(\frac{d}{dx}+it\biggr)^2.
\]
The operator $A_\e^{(t)}$ is defined on the compact graph $\mathbb G_\e$
subject to a set of non-local conditions at the vertices
of its flattened realisation. These non-local conditions are naturally split into two sets, the former originating from the non-local matching conditions arising due to \eqref{Kirchhoff_periodic} via flattening. These under the Gelfand transform are converted into matching conditions admitting the same form, although the values of functions and inward normal co-derivatives acquire certain unimodular ``weights''. The second set appears due to the periodicity of the operator $A_\e$ on the graph $\mathbb G_{\text{per}}$. These only involve the endpoints of the edges which belong to ``periodic'' vertices $\mathcal V^{\text{per}}$ in the original (unflattened) graph.

These observations allow to ``invert'' the flattening introduced above once the Gelfand transform has been applied. This inversion clearly yields the original graph $\mathbb G_\e$, in which every vertex $V$ belonging to $\mathcal V^{\text{per}}\cap T_0$ has been identified with its translation $W=V+\vec \ell$ in $T_1$. This identification means that the combined vertex $VW:=V\cup W$ is the equivalence class of edge endpoints which is the union of equivalence classes defining $V$ and $W$. We denote the resulting graph $\widehat{\mathbb G}_\e$.   Clearly, the number of edges in it is the same as in ${\mathbb G}_\e$, whereas the number of vertices is $\widehat{N}=N-\card \mathcal V^{\text{per}}\cap T_0$. This construction depends heavily on the geometry of the original graph, but it will be clarified in Section \ref{section:examples}, where we consider three examples in full detail.

Next we define the operator family $A_\e^{(t)}$ on the compact finite graph $\widehat{\mathbb G}_\e$.
On each edge $e$ of the graph $\widehat{\mathbb G}_\e$ the operator $A_\e^{(t)}$ is defined by the differential expression
$$
-a_e^2(\e)\biggl(\frac d{dx}+{\it i}t\biggr)^2.
$$
The domain of $A_\e^{(t)}$ is described by \emph{weighted Kirchhoff}, or \emph{Datta--Das Sarma} (see \cite{Datta}) conditions at each graph vertex. Namely, at any vertex $V$ of $\widehat{\mathbb G}_\e$ one sets for $u\in W^{2,2}(\widehat{\mathbb G}_\e)\equiv W^{2,2}({\mathbb G}_\e)$ (we set $u_e:=u|_{e}$ as above, and recall that $t\in[-\pi/\varepsilon, \pi/\varepsilon)$):
\begin{equation}
\label{weightedK}
\begin{gathered}
(i)\ \text{For any } e,e' \text{ such that } e,e'\thicksim V\\[0.4em]
\text{there exists a common value, denoted  } u(V), \text{ such that }\\[0.4em]
w_V(e)u_e(V)=w_V(e') u_{e'}(V)=:u(V)\\[0.4em]
(ii)\ \sum_{e\thicksim V} \widehat{\partial}_n^{(t)}u_e(V)=0, \text{ where }\\[0.4em]
\widehat{\partial}_n^{(t)} u_e(V):=\begin{cases}
w_V(e) a_e^2(\e) \biggl(\dfrac d{dx}+it\biggr) u_e(V), & V \text{ is the left endpoint of  } e,\\[1.0em]
-w_V(e) a_e^2(\e) \biggl(\dfrac d{dx}+it\biggr)u_e(V), & V \text{ is the right endpoint of  } e.
\end{cases}
\end{gathered}
\end{equation}
Here we suitably adjust the notation when $e$ is a loop, and $\{w_V(e)\}_{e\thicksim V}$ is defined at each vertex $V$ as a list of unimodular complex numbers. This list depends on a concrete choice of the graph flattening, see the related discussion above, leading to unitarily equivalent formulations for all such choices. We therefore refrain from providing any explicit expressions in the general case. However, in Section \ref{section:examples} we give such expressions for each of the three examples discussed there.

It is easily seen that $A_\e^{(t)}$ thus defined is a self-adjoint operator. The standard compactness argument is used to ascertain that its spectrum is discrete and accumulates to $+\infty$. The following theorem follows from the argument presented above.
\begin{theorem}\label{Gelfand-graph}
For each vertex $V\in \widehat{\mathbb G}_\e$ there exists a unimodular list $\{w_V(e)\}_{e\thicksim V}$ such that
the image of the operator family $A_\e$ under the Gelfand transform \eqref{Gelfand} is the fibre representation \eqref{vonNeumann}, where the operator family $A_\e^{(t)}$ is defined on $\widehat{\mathbb G}_\e$ by \eqref{weightedK}.
\end{theorem}

The \emph{proof} is obtained by a straightforward computation based on those presented in \cite{Yorzh3,Physics}. It has to be noted that the values of $w_V(e)$ do depend on the numbering of graph vertices in $\mathbb G_\e$ via the flattening applied. This, however, happens in a unitary equivalent fashion.

\subsection{Rescaling to the graph of total length $1$}
\label{rescaling_section}

Guided by the result presented in \cite{Physics}, we introduce a unitary rescaling for the operator family $A_\e^{(t)}$. Set $\tau:=\e t\in[-\pi, \pi)$ and consider the rescaled operator family $A_\e^{(\tau)}$ defined by
$A_\e^{(\tau)}=\Phi_\e A_\e^{(t)}\Phi_\e^*$, where the unitary $\Phi_\e$ acts on $u\equiv\{u_e\}_{e\in \widehat{\mathbb G}_\e}\in L^2(\widehat{\mathbb G}_\e)$ by the following rule:
$$
\Phi_\e u_e= \sqrt{\e} u_e(\e x),\quad \forall\ e\in \widehat{\mathbb G}_\e .
$$
To simplify notation, we have elected to keep the same symbol $A_\e^{(\tau)}$ for the unitary image of $A_\e^{(t)},$ where $t=\tau/\varepsilon\in[-\pi/\varepsilon, \pi/\varepsilon).$ We hope that this does not lead to any misunderstanding.

Under the transformation $\Phi_\e$ the graph $\widehat{\mathbb G}_\e$ becomes $\mathbb G$ which is the same graph with every edge length multiplied by $1/\e$; the total length of $\mathbb G$ is $1$. For brevity we keep the same notation $e^{(p)}$ for the edges of the rescaled graph $\mathbb G.$ The operator $A_\e$ is then unitary equivalent to the direct von Neumann integral,
$$
A_\e\cong\oplus \int_{-\pi}^\pi A_\e^{(\tau)} d\tau,
$$
and the operator family $A_\e^{(\tau)}$ admits the following explicit description. On each edge $e$ of the graph $\mathbb G$ the operator $A_\e^{(\tau)}$ is defined by the differential expression
\[
-\frac{a_e^2(\e)}{\e^2}\biggl(\frac d{dx}+i\tau\biggr)^2.
\]
The domain of $A_\e^{(\tau)}$ is described by the \emph{Datta--Das Sarma}  conditions, {\it i.e} $u\in\dom A_\e^{(\tau)}$ if at each graph vertex $V$ of $\mathbb G$ the function $u\in W^{2,2}(\mathbb G)$ and
\begin{equation}\label{weightedK1}
\begin{gathered}
(i)\ \text{For any } e,e' \text{ such that } e,e'\thicksim V\\[0.4em]
\text{there exists a common value, denoted  } u(V), \text{ such that }\\[0.4em]
w_V(e)u_e(V)=w_V(e') u_{e'}(V)=:u(V)\\[0.4em]
(ii)\ \sum_{e\thicksim V} \widehat{\partial}_n^{(\tau)}u_e(V)=0, \text{ where }\\[0.4em]
\widehat{\partial}_n^{(\tau)} u_e(V):=\begin{cases}
w_V(e) \dfrac {a_e^2(\e)}{\e^2}\biggl(\dfrac d{dx}+i\tau\biggr) u_e(V), & V \text{ is the left endpoint of  } e\\[1.0em]
-w_V(e) \dfrac {a_e^2(\e)}{\e^2}\biggl(\dfrac d{dx}+i\tau\biggr)u_e(V), & V \text{ is the right endpoint of  } e,
\end{cases}
\end{gathered}
\end{equation}
where, as before, we suitably adjust the notation when $e$ is a loop. The stiff-soft decomposition $\mathbb G_\e=\mathbb G_{\stiff,\e} \cup \mathbb G_{\soft,\e}$ is replicated in $\mathbb G,$ so that $\mathbb G=\mathbb G_\stiff \cup \mathbb G_\soft$.

We remark that under the choice of $a_e(\e)$ made above, the operator family $A_\e^{(\tau)}$ on the soft component $\mathbb G_\soft$ is defined by the
$\e$-independent differential expression
\[
-a_e^2\biggl(\frac d{dx}+i\tau\biggr)^2,\quad e\in \mathbb G_\soft.
\]
On the stiff component $\mathbb G_\stiff$ the symbol of the operator is $1/\e^2$-large. It is for this reason that we find convenient to introduce the factors $1/\e^2$ in the definition of the operator $\widehatdntau$ above; indeed, on the edges comprising the soft component this leads to the $\e$-independent expression
\[
\pm w_V(e) a_e^2 \biggl(\frac d{dx}+i\tau\biggr) u_e(V),
\]
in line with that for the action of the operator itself.

We also note that the same family $A_\e^{(\tau)}$ would have appeared if we first applied the unitary rescaling $\Phi_\e$ to the operator family $A_\e$, followed by the application of the Gelfand transform \eqref{Gelfand} with $\e$ set to $1$. This remark also proves that the Datta-Das Sarma weights $w_V(e)$ depend on $\varepsilon$ via $\e t=\tau$.

Finally, we note that each edge $e$ of the graph $\mathbb G$ can be, by a shift of variable, identified with the segment $[0,l^{(e)}]$, where as above $l^{(e)}$ is the length of the edge $e$. We will consistently make use of this identification below.

\section{Preliminaries: boundary triples and the Weyl $M$-function}

Our approach is
based
on the theory of boundary triples \cite{Gor,Ko1,Koch,DM}, applied
to the class of operators introduced above. We next recall two fundamental
concepts of this theory, namely the boundary triple and the generalised Weyl-Titchmarsh
matrix function. Assume that $A_{\min}$ is a symmetric
densely defined operator with equal deficiency indices in a Hilbert space $H,$ and set $A_{\max}:=A_{\min}^*$.

\begin{definition}[\cite{Gor,Ko1,DM}]\label{Def_BoundTrip}
\label{definition1_1}
Let $\Gamma_0,$ $\Gamma_1$ be linear mappings of ${\rm dom}(A_{\max})$
to an auxiliary separable Hilbert space $\mathcal{H}.$ The triple
$(\mathcal{H}, \Gamma_0,\Gamma_1)$ is called \emph{a boundary
triple} for the operator $A_{\max}$ if:
\begin{enumerate}
\item For all $u,v\in {\rm dom}(A_{\max})$ one has
\begin{equation}
\langle A_{\max} u,v \rangle_H -\langle u, A_{\max} v \rangle_H = \langle \Gamma_1 u, \Gamma_0
v\rangle_{\mathcal{H}}-\langle\Gamma_0 u, \Gamma_1 v\rangle_{\mathcal{H}}.
\label{Green_identity}
\end{equation}
\item The mapping
$u\longmapsto (\Gamma_0 u;
\Gamma_1 u),$ $f\in {\rm dom}(A_{\max})$ is onto ${\mathcal H}\times{\mathcal H}.$
\end{enumerate}

A non-trivial extension ${A}_B$ of the operator $A_{\min}$ such
that $A_{\min}\subset  A_B\subset A_{\max}$  is called
\emph{almost solvable} if there exists a boundary triple
$(\mathcal{H}, \Gamma_0,\Gamma_1)$ for $A_{\max}$ and a bounded
linear operator $B$ defined on $\mathcal{H}$ such that for every
$u\in {\rm dom}(A_{\max})$
$$
u\in {\rm dom}({A_B})\ \ \ \text{\ if and only if }\ \ \ \Gamma_1 u=B\Gamma_0 u.
$$

The (correctly defined) operator-valued function $M=M(z)$ given by
\begin{equation*}\label{Eq_Func_Weyl}
M(z)\Gamma_0 u_{z}=\Gamma_1 u_{z}, \ \
u_{z}\in \ker (A_{\max}-z),\  \ z\in
\mathbb{C}_+\cup{\mathbb C}_-,
\end{equation*}
is called the Weyl-Titchmarsh function, or $M$-function function, of the operator
$A_{\max}$ with respect to the corresponding boundary triple.
\end{definition}


One of the cornerstones of our analysis is the celebrated Kre\u\i n formula, which allows to relate the resolvent of $A_B$ to the resolvent of a self-adjoint operator $A_\infty$ defined as the
restriction of the maximal operator $A_{\text{\rm max}}$ to the set
$$
{\rm dom}(A_\infty)=\bigl\{u\in \dom A_{\text{\rm max}}|\, \Gamma_0 u=0\bigr\}.
$$
(We follow Birman-Krein-Vishik \cite{Birman, Krein, Vishik}, see also \cite{Ryzhov}, in using the notation $A_\infty,$ justified by the fact that in the language of triples this extension formally corresponds to $A_B$ with $B=\infty.$)

In particular, we will find it necessary to consider not only proper operator extensions $A_B$, but also those for which the parameterising operator $B$ depends on the spectral parameter $z$. This amounts to considering spectral boundary-value problems where the spectral parameter is present not only in the differential equation but also in the boundary conditions:
\begin{equation}
A_{\text{\rm max}}u-z u =f,\ \ \ \
u\in {\rm dom} (A_{\text{\rm max}}),\ \ \
\Gamma_1 u = B(z) \Gamma_0 u.
\label{gen_prob}
\end{equation}
The solution operator $R(z)$ for a boundary-value problem of this type is known \cite{Strauss} to be a generalised resolvent in the case when $-B(z)$ is an operator-valued $R$-function:
if $B(z)$ is analytic in $\mathbb C_+\cup \mathbb C_-$ with $\Im z \Im B(z)\leq 0,$ then
\begin{equation}\label{eq:out-of-space}
R(z)=P_H(A_{\mathfrak H}-z)^{-1}\bigr|_H,
\end{equation}
where $\mathfrak H$ is a Hilbert space such that $H\subset \mathfrak H,$ the operator $P_H$ is the orthogonal projection of $\mathfrak H$ onto $H,$ and $A_{\mathfrak H}$ is a self-adjoint in $\mathfrak H$ out-of-space extension of the operator $A_{\text{min}}$.

On the other hand, for any fixed $z$ the operator $R(z)$  coincides with the resolvent (evaluated at the point $z$) of a closed linear operator in $H$
that is an anti-dissipative for $z\in{\mathbb C}_+$ (dissipative for $z\in{\mathbb C}_-$)  extension of $A_{\text{min}}$ with the $z$-dependent domain given in (\ref{gen_prob}).
It is for this reason that in what follows we preserve the notation $(A_B-z)^{-1}$ for the generalised resolvent of $A_B$ when
$B=B(z).$

The Kre\u\i n formula suitable for treatment of such problems was obtained in \cite{DM}.

\begin{prop}[Version of the Kre\u\i n formula of \cite{DM}]\label{prop:Krein}
Assume that $\{\mathcal{H},\Gamma_0,\Gamma_1\}$ is a boundary triple for the operator $A_{\text{\rm max}}$. Then for the (generalised) resolvent  $(A_B-z)^{-1}$, where $B=B(z)$ is a bounded operator in $\mathcal{H}$ for $z\in\mathbb C_+\cup \mathbb C_-$, one has, for all
$z\in \rho(A_B)\cap \rho(A_\infty):\,$\footnote{It is checked that the operator function $B-M$ is invertible under the conditions of the proposition.}
\begin{multline}
\label{eq:resolvent}
(A_B-z)^{-1}=(A_\infty-z)^{-1}+ \gamma(z)\bigl(B(z)-M(z)\bigr)^{-1}\gamma^*(\bar z)
\\
=(A_\infty-z)^{-1}+ \gamma(z)\bigl(B(z)-M(z)\bigr)^{-1}\Gamma_1 (A_\infty-z)^{-1},
\end{multline}
where $M(z)$ is the  M-function of $A_{\text{\rm max}}$ with respect to the boundary triple $\{\mathcal{H},\Gamma_0,\Gamma_1\}$ and $\gamma(z)$ is the solution operator
$$
\gamma(z)=\bigl(\Gamma_0|_{\text{\rm ker\,}(A_{\text{\rm max}}-z)}\bigr)^{-1}.
$$
\end{prop}

\subsection{The triple}
\label{triple_section}

In order to apply the theory of boundary triples to the operator family $A_\e^{(\tau)}$ we first choose a convenient boundary triple. It was shown in \cite{Yorzh3}, see also references therein, how a ``natural'' boundary triple is selected in the setting of quantum graphs, and we follow the mentioned approach here.

First, we define a ``maximal'' operator
$A_{\max}$ in the space $L^2(\mathbb G)$ by the same differential expression\footnote{For brevity, we henceforth omit the subscript $\e$ and the superscript $(\tau)$ in the notation pertaining to maximal operators.} as the operator $A_\e^{(\tau)}$.
Its domain $\dom A_{\max}$ is defined as those $u\in W^{2,2}(\mathbb G)$ that admit the \emph{weighted continuity condition} at all graph vertices $V$:
\begin{equation}\label{domAmax}
\begin{gathered}
\text{For any } e,e' \text{ such that } e,e'\thicksim V\\
\text{there exists a common value, denoted  } u(V), \text{ such that }\\
w_V(e)u_e(V)=w_V(e') u_{e'}(V)=:u(V)
\end{gathered}
\end{equation}
(as above, $u_e:=u|_{e}$).

We set the adjoint to $A_{\max}$ to be the ``minimal'' densely defined symmetric operator $A_{\min}.$  We choose the boundary triple as follows: the boundary space is $\mathcal{H}=\mathbb{C}^N$, where $N=N(\mathbb G)$ is the number of vertices in $\mathbb G$, and the boundary operators are chosen as follows.
\begin{equation}\label{eq:triple}
     (\Gamma_0u)_V:=  u(V),\ \ \ \ \
     (\Gamma_1u)_V:=
        \sum_{{e}\thicksim V} \widehatdntau u_e(V), \ \ \ \ \ \forall\ V\in\mathbb G,
\end{equation}
where $\widehatdntau$ is defined in (\ref{weightedK1}).
The Green identity (\ref{Green_identity})
holds by integration by parts, see \cite{CherKis,Yorzh3} for details. It follows that the operator $A_\e^{(\tau)}$ is an almost solvable extension of $A_{\min}$ associated with the matrix $B=0$.

In what follows, we will further require two more boundary triples, constructed separately for the operators pertaining to the stiff and soft components of the graph $\mathbb G$, respectively. The maximal operators of these triples $A_{\max}^\stiff$ and $A_{\max}^\soft$ are defined by the same differential expression as that defining $A_\e^{(\tau)}$, but on $L^2(\mathbb G_\stiff)$ and $L^2(\mathbb G_\soft),$ respectively.
The domains of $A_{\max}^\stiff$ and $A_{\max}^\soft$ are set by \eqref{domAmax} restricted to $e,e'\in \mathbb G_{\stiff\ (\soft)}$. Finally, the boundary operators $\Gamma_0^{\text{stiff (soft)}}$ and $\Gamma_1^{\text{stiff (soft)}}$ are defined by \eqref{eq:triple}, but the sum in the second expression is taken over $e\in \mathbb G_{\stiff\ (\soft)}$ only. We remark that although formally all three operators $\Gamma_0,$ $\Gamma_0^{\text{stiff (soft)}}$ are defined by the same rule, their domains are clearly different since the weighted continuity condition \eqref{domAmax} depends on the underlying graph.

By this construction, the boundary spaces for $A_{\max}^\stiff$ and $A_{\max}^\soft$ are chosen as $\mathbb C^{N(\mathbb G_\stiff)}$ and $\mathbb C^{N(\mathbb G_\soft)}$, respectively. In general, this does not quite suite us as (see \cite{Physics}) we need to ensure that $N(\mathbb G)=N(\mathbb G_\soft)=N(\mathbb G_\stiff)$. In the examples considered below, the latter identity holds automatically. The general case admits a reduction to the one considered below via an application of \cite[Appendix A]{Physics}.

\subsection{$M$-matrix}
The derivation of the $M$-matrix $M_\e^{(\tau)}(z)$ with respect to the triple $(\mathbb C^N, \Gamma_0,\Gamma_1)$ defined by \eqref{eq:triple} is based on the same argument as in \cite{Yorzh3} (cf. \cite{CherKis,Physics}) which permits us to omit it here. The result is formulated in the following
\begin{theorem}\label{thm:M-matrix}
Assume that $\mathbb G$ contains no loops.\footnote{This assumption is without loss of generality. Indeed, one can always add auxiliary vertices of degree 2 to the graph $\mathbb G$ to satisfy it.}
The Weyl-Titchmarsh $M$-matrix $M_\e^{(\tau)}(z)$ has
matrix elements given by the following formula:
\begin{equation}\label{Eq_Weyl_Func_Delta}
M_{jm}=\begin{cases}
-k \sum\limits_{e\thicksim V_m} \dfrac {a_e(\e)}{\e} \cot \dfrac{k \varepsilon l^{(e)}}{a_{e}(\e)},& m=j,\\[0.8em]
 \sum\limits_{e\thicksim V_m, e\thicksim V_j} \overline{w_{V_m}(e)}w_{V_j}(e)e^{i\sigma_m(e)l^{(e)}\tau} k \dfrac{a_{e}(\e)}{\e}\csc \dfrac{k \varepsilon l^{(e)}}{a_{e}(\e)},\ \ \ \ \ &m\neq j;\ \exists\ e\thicksim V_m \text{ and }
                     e\thicksim V_j,\\[1.2em]
0& \text{ otherwise.}
\end{cases}
\end{equation}
Here $k=\sqrt{z}$ (the branch such that $\Im k\geq 0$), $l^{(e)}$ is the length of the edge $e$, and
$$
\sigma_m(e)=\begin{cases}
-1,& e \text{ is an outgoing edge for } V_m,\\[0.3em]
+1,& e \text{ is an incoming edge for } V_m.
\end{cases}
$$
\end{theorem}
Note that \eqref{Eq_Weyl_Func_Delta} defines a Hermitian matrix for real values of $k$ away from a discrete set of $k$.

The next statement, which is commonly used in both ODE and PDE contexts, proves to be valuable for our analysis. Its proof can be obtained, {\it e.g.}, by minor modifications, due to the presence of Datta-Das Sarma weights, of the related proof in \cite{CherKisSilva}.

\begin{proposition}\label{thm:additivity}
Let\footnote{This assumption is without loss of generality due to the reduction of \cite[Appendix A]{Physics}.} $N(\mathbb G)=N(\mathbb G_\soft)=N(\mathbb G_\stiff)$. Let the operators $A_{\max},$ $A_{\max}^\stiff$ and $A_{\max}^\soft$ along with their boundary triples be chosen as in Section \ref{triple_section}. Then
\begin{equation}\label{eq:additivity}
M_\e^{(\tau)}(z)=M_\e^{(\tau),\stiff}(z)+M_\e^{(\tau),\soft}(z),
\end{equation}
where $M_\e^{(\tau),\stiff}(z)$ and $M_\e^{(\tau),\soft}(z)$ are the $M-$matrices of the operators $A_{\max}^\stiff$ and $A_{\max}^\soft$, respectively.
\end{proposition}

\section{Three examples}
\label{section:examples}

In the present section, we introduce three examples of graphs periodic along one axis, which we consider in full detail below. These correspond to the following mutually exclusive setups: (0) disconnected stiff and soft components; (1) disconnected soft component, connected stiff component; (2) connected soft component, disconnected stiff component. The fourth possibility, {\it i.e.}, the one where both components of the medium are connected, proves to yield no new effects compared to (1) and (2), therefore we omit it. As stated above, the general case can be reduced to one of these examples (although the reduction proves to be non-trivial). We mention that the reduction via \cite[Appendix A]{Physics} is not the most effective and elegant of those available; it proves to be possible to construct a deletion-contraction type reduction. This latter subject falls beyond the scope of the present paper and will be treated elsewhere.

\subsection*{(0) A medium with both components disconnected}
Consider the high-contrast one-dimensional periodic medium (cf. \cite{CherKis}), as shown in Fig. \ref{fig:0}. In this case, the graph $\mathbb G_{\text{per}}$ is nothing but an infinite periodic chain-graph.
\begin{figure}[h!]
\begin{center}
\includegraphics[scale=0.7]{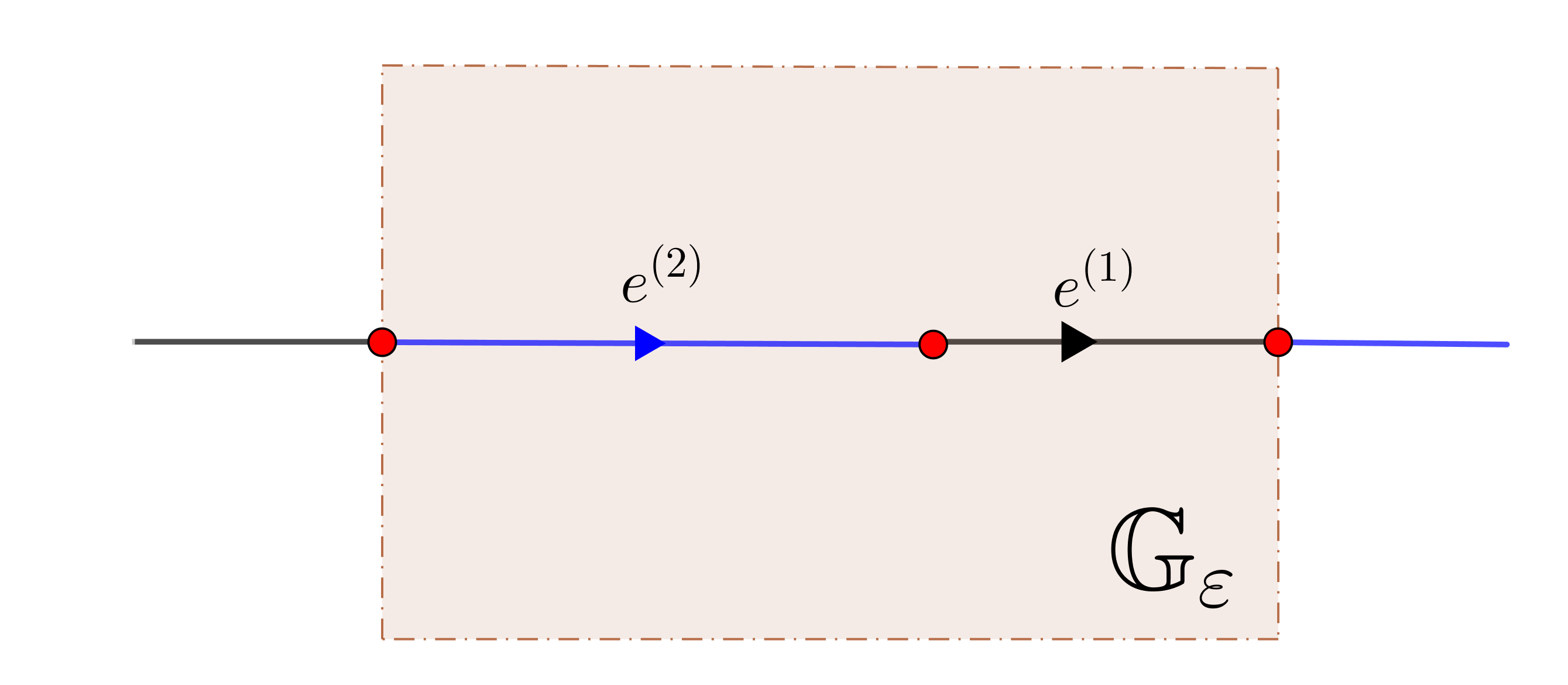}
\includegraphics[scale=0.7]{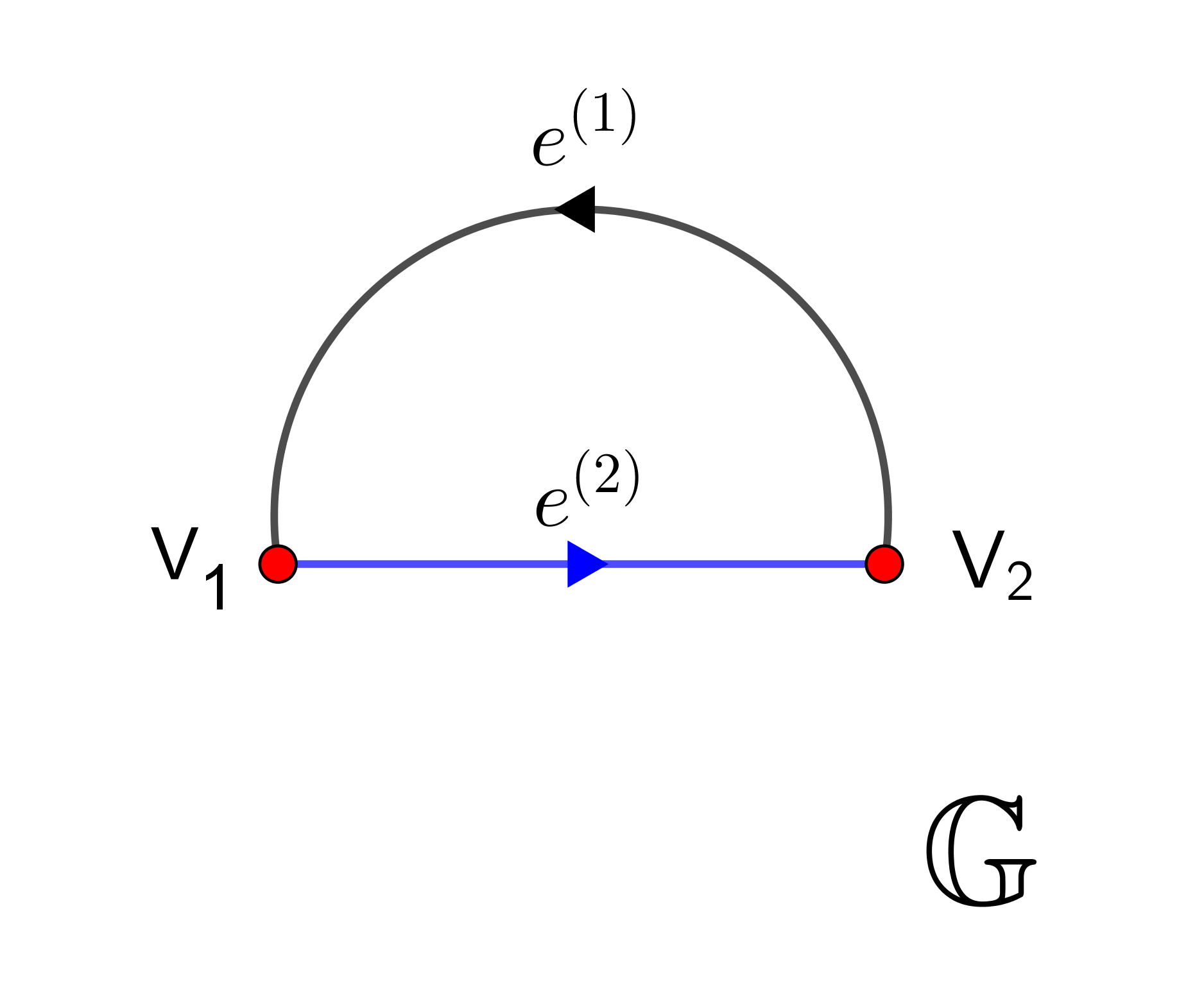}
\end{center}
\caption{{\scshape Example (0).} {\small $\mathbb G_{\text{per}}$ with $\mathbb G_\e$ outlined on the left; the graph $\mathbb G$ after Gelfand transform on the right. The soft component is drawn in blue.}\label{fig:0}}
\end{figure}

The boundary space $\mathcal H$ pertaining to the graph $\mathbb G$ is chosen as $\mathcal H=\mathbb C^2$. The unimodular list functions $w_{V_1}$ and $w_{V_2}$ are chosen as follows:
\begin{equation}\label{eq:0-weights}
\begin{gathered}
\{w_{V_1}(e^{(j)})\}_{j=1}^2=\{1,1\},\qquad\{w_{V_2}(e^{(j)})\}_{j=1}^2=\{1,1\}.
\end{gathered}
\end{equation}
We note, that Datta-Das Sarma weights in this example can be chosen to be trivial due to the fact that no flattening was applied to $\mathbb G_\e$.

Theorem \ref{thm:M-matrix} and Proposition \ref{thm:additivity} now yield the following expressions for the corresponding Dirichlet-to-Neumann maps.

\begin{lemma}\label{lemma:M_0}
Let the maximal operators $A_{\max}$, $A_{\max}^{\stiff (\soft)}$ and boundary operators $\Gamma_j$, $\Gamma^{\stiff (\soft)}_j$ ($j=1,2$) be chosen as in Section \ref{triple_section}. Then
\begin{equation}\label{D2N_0}
\begin{gathered}
M_\e^{(\tau),\stiff}=\dfrac k\e\begin{pmatrix}
-a_1  \cot \dfrac{k \e l^{(1)}}{a_1}& a_1 e^{-i l^{(1)} \tau} \csc \dfrac{k \e l^{(1)}}{a_1}\\[1.1em]
a_1 e^{i l^{(1)} \tau} \csc \dfrac{k \e l^{(1)}}{a_1}&-a_1  \cot \dfrac{k \e l^{(1)}}{a_1}
\end{pmatrix},\\[0.9em]
M_\e^{(\tau),\soft}=k\begin{pmatrix}
-a_2 \cot \dfrac{k l^{(2)}}{a_2}& a_2 e^{i l^{(2)} \tau} \csc \dfrac{k l^{(2)}}{a_2}\\[1.1em]
a_2 e^{-i l^{(2)} \tau} \csc \dfrac{k l^{(2)}}{a_2}&-a_2  \cot \dfrac{k l^{(2)}}{a_2}
\end{pmatrix},
\end{gathered}
\end{equation}
and \eqref{eq:additivity} holds.
\end{lemma}

For simplicity, we will henceforce assume without loss of generality that $a_2=1$.

\subsection*{(1) A case of connected stiff component}
The periodic graph considered, its periodicity cell and the result of Gelfand transform is shown in Fig. \ref{fig:1}.
\begin{figure}[h!]
\begin{center}
\includegraphics[scale=0.7]{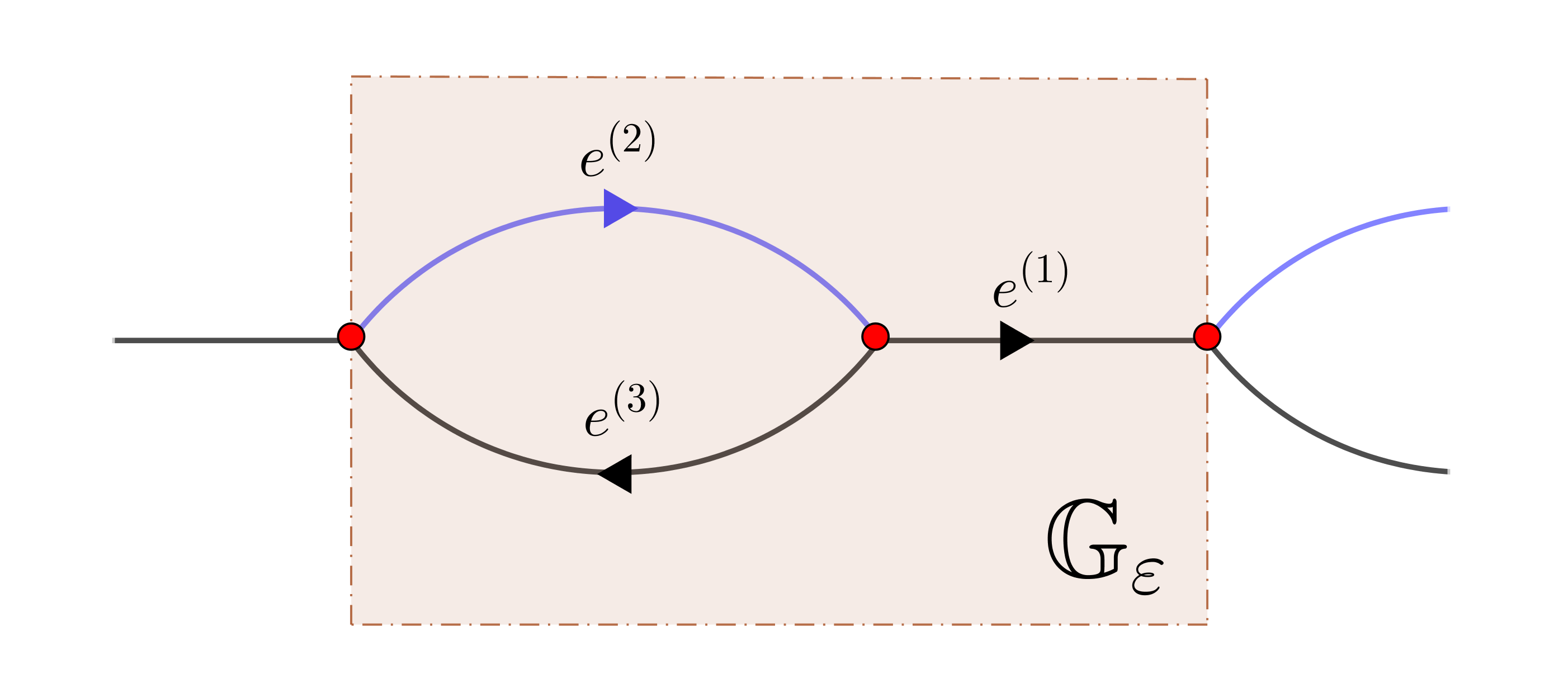}
\includegraphics[scale=0.7]{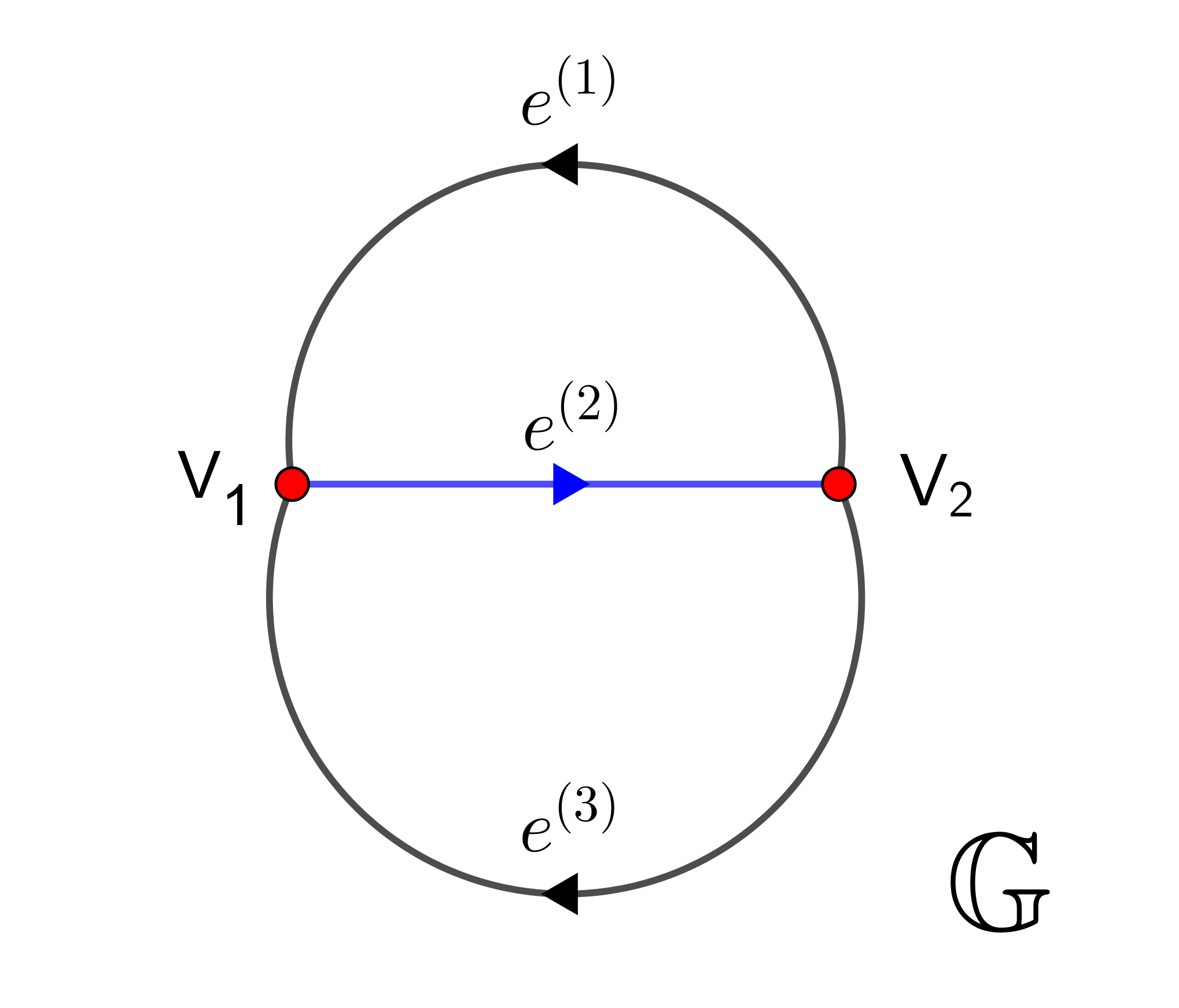}
\end{center}
\caption{{\scshape Example (1).} {\small $\mathbb G_{\text{per}}$ with $\mathbb G_\e$ outlined on the left; the graph $\mathbb G$ after Gelfand transform on the right. The soft component is drawn in blue.}\label{fig:1}}
\end{figure}
The boundary space $\mathcal H$ pertaining to the graph $\mathbb G$ is chosen as $\mathcal H=\mathbb C^2$. The unimodular list functions $w_{V_1}$ and $w_{V_2}$ are chosen as follows:
\begin{equation}\label{eq:1-weights}
\begin{gathered}
\{w_{V_1}(e^{(j)})\}_{j=1}^3=\{1,1,e^{i\tau(l^{(2)}+l^{(3)})}\},\quad \{w_{V_2}(e^{(j)})\}_{j=1}^3=\{e^{i\tau l^{(3)}},1,1\}
\end{gathered}
\end{equation}

Theorem \ref{thm:M-matrix} and Proposition \ref{thm:additivity} yield the following statement.

\begin{lemma} Let the maximal operators $A_{\max}$, $A_{\max}^{\stiff (\soft)}$ and boundary operators $\Gamma_j$, $\Gamma_j^{\stiff (\soft)}$  ($j=1,2$) be chosen as in Section \ref{triple_section}. Then
\begin{equation}\label{D2N_1}\begin{gathered}
M_\e^{(\tau),\stiff}=\dfrac k\e\begin{pmatrix}
-a_1  \cot \dfrac{k \e l^{(1)}}{a_1}-a_3 \cot \dfrac{k \e l^{(3)}}{a_3}
& a_1 e^{-i (l^{(1)}+l^{(3)}) \tau} \csc \dfrac{k \e l^{(1)}}{a_1}+a_3 e^{i l^{(2)} \tau} \csc \dfrac{k \e l^{(3)}}{a_3}\\[1.1em]
a_1 e^{i (l^{(1)}+l^{(3)}) \tau} \csc \dfrac{k \e l^{(1)}}{a_1}+a_3 e^{-i l^{(2)} \tau} \csc \dfrac{k \e l^{(3)}}{a_3}&-a_1  \cot \dfrac{k \e l^{(1)}}{a_1}-a_3  \cot \dfrac{k \e l^{(3)}}{a_3}
\end{pmatrix},\\[0.9em]
M_\e^{(\tau),\soft}=k\begin{pmatrix}
-a_2  \cot \dfrac{k l^{(2)}}{a_2}& a_2 e^{i l^{(2)} \tau} \csc \dfrac{k l^{(2)}}{a_2}\\[1.1em]
a_2 e^{-i l^{(2)} \tau} \csc \dfrac{k l^{(2)}}{a_2}&-a_2 \cot \dfrac{k l^{(2)}}{a_2}
\end{pmatrix}
\end{gathered}
\end{equation}
and \eqref{eq:additivity} holds.
\end{lemma}

For simplicity, we will henceforce assume when treating this example, that $a_2=1$.

\subsection*{(2) A case of connected soft component}
The periodic graph considered, its periodicity cell and the result of Gelfand transform is shown in Fig. \ref{fig:2}. It represents a ``dual'' situation to the one of Example (1), exhibiting a globally connected soft component.
\begin{figure}[h!]
\begin{center}
\includegraphics[scale=0.7]{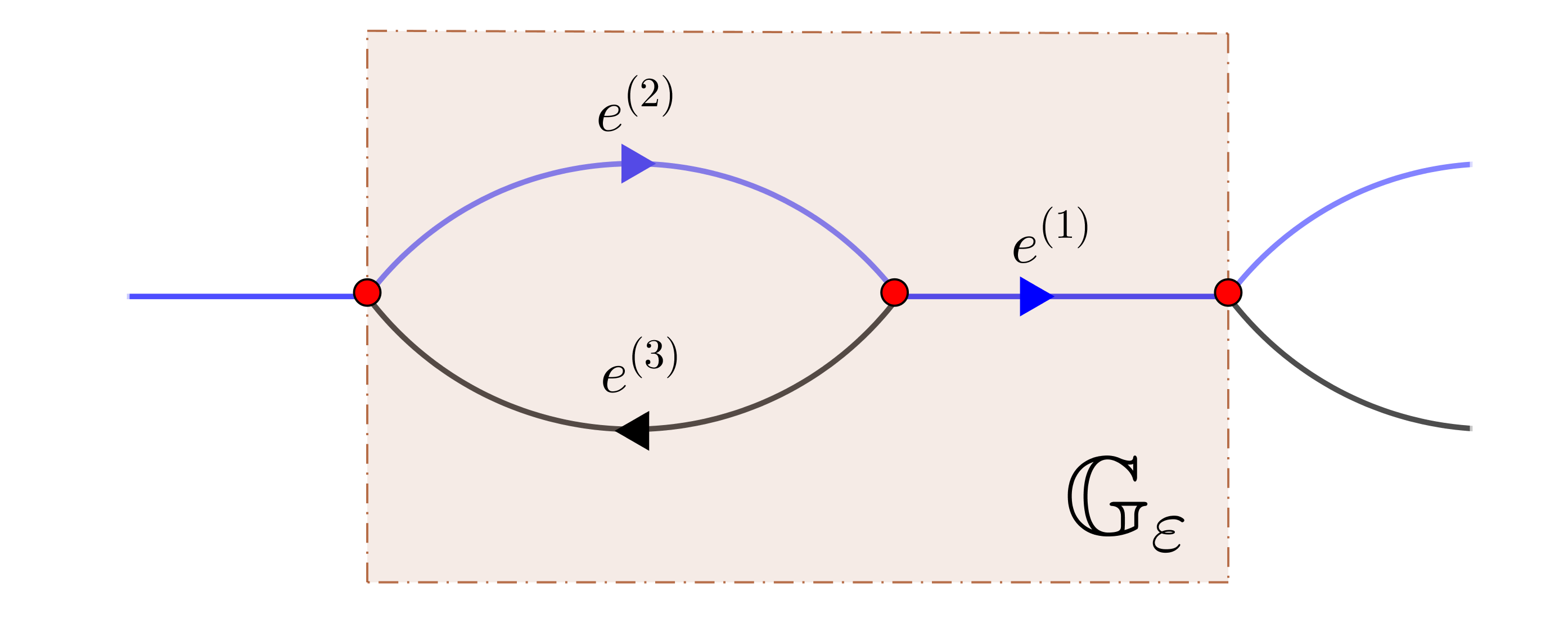}
\includegraphics[scale=0.7]{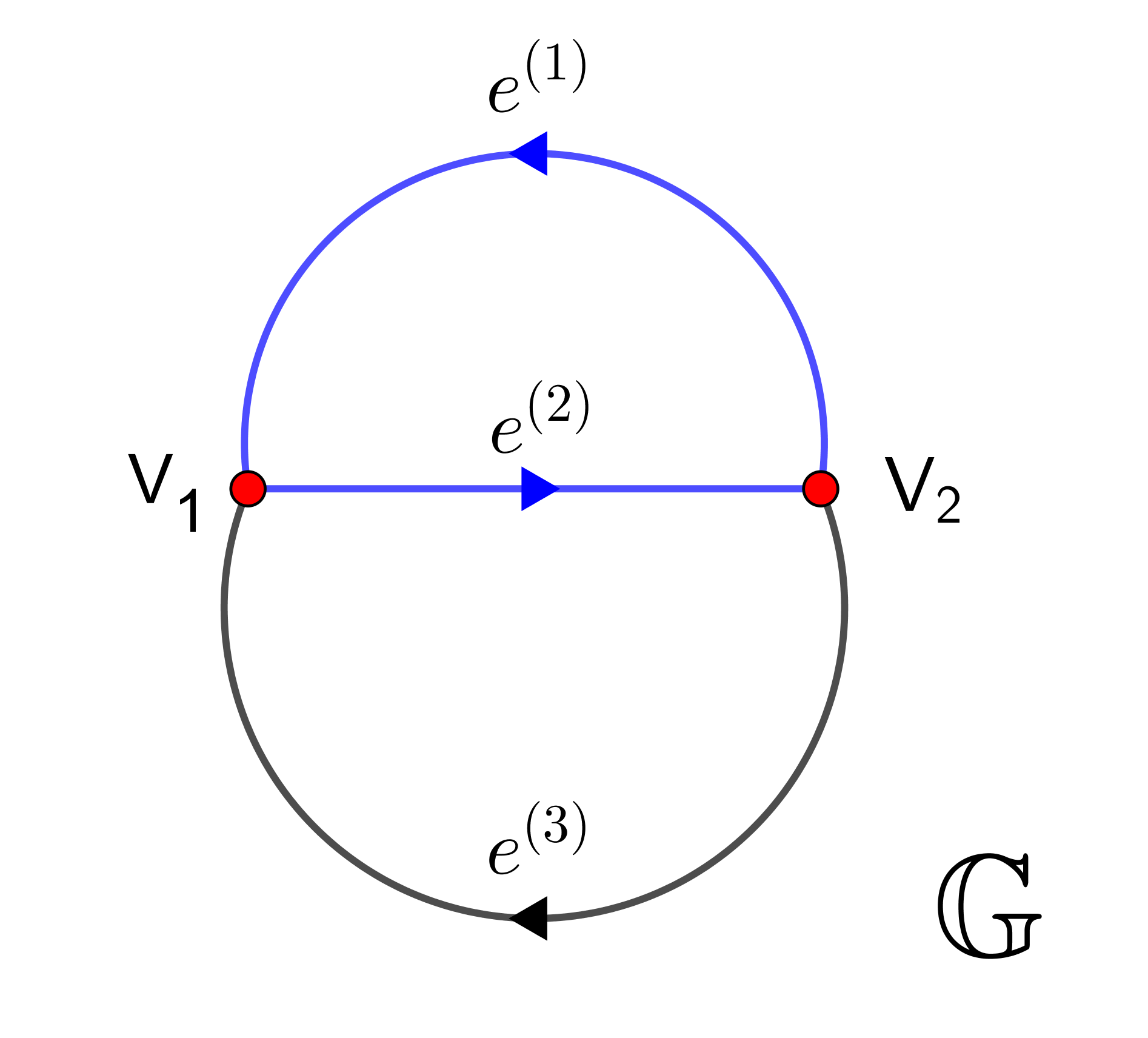}
\end{center}
\caption{{\scshape Example (2).} {\small $\mathbb G_{\text{per}}$ with $\mathbb G_\e$ outlined on the left; the graph $\mathbb G$ after Gelfand transform on the right. The soft component is drawn in blue.}\label{fig:2}}
\end{figure}
The boundary space $\mathcal H$ pertaining to the graph $\mathbb G$ is chosen as $\mathcal H=\mathbb C^2$. The unimodular list functions $w_{V_1}$ and $w_{V_2}$ are chosen as in \eqref{eq:1-weights}.

Theorem \ref{thm:M-matrix} and Proposition \ref{thm:additivity} yield the following statement.

\begin{lemma} Let the maximal operators $A_{\max}$, $A_{\max}^{\stiff (\soft)}$ and boundary operators $\Gamma_j$, $\Gamma^{\stiff (\soft)}_j$ ($j=1,2$) be chosen as in Section \ref{triple_section}.
Then
\begin{equation}\label{D2N_2}\begin{gathered}
M_\e^{(\tau),\stiff}=\dfrac k\e\begin{pmatrix}
-a_3  \cot \dfrac{k \e l^{(3)}}{a_3}
& a_3 e^{i l^{(2)} \tau} \csc \dfrac{k \e l^{(3)}}{a_3}\\[1.1em]
a_3 e^{-i l^{(2)} \tau} \csc \dfrac{k \e l^{(3)}}{a_3}&
-a_3  \cot \dfrac{k \e l^{(3)}}{a_3}
\end{pmatrix},\\[0.9em]
M_\e^{(\tau),\soft}=k\begin{pmatrix}
-a_1 \cot \dfrac{k l^{(1)}}{a_1}-a_2 \cot \dfrac{k l^{(2)}}{a_2}& a_1 e^{-i (l^{(1)}+l^{(3)}) \tau} \csc \dfrac{k  l^{(1)}}{a_1}+a_2 e^{i l^{(2)} \tau} \csc \dfrac{k l^{(2)}}{a_2}\\[1.1em]
a_1 e^{i (l^{(1)}+l^{(3)}) \tau} \csc \dfrac{k l^{(1)}}{a_1}+a_2 e^{-i l^{(2)} \tau} \csc \dfrac{k l^{(2)}}{a_2}&-a_1  \cot \dfrac{k l^{(1)}}{a_1}-a_2  \cot \dfrac{k l^{(2)}}{a_2}
\end{pmatrix},
\end{gathered}
\end{equation}
and \eqref{eq:additivity} holds.
\end{lemma}

\section{Asymptotic analysis of a sandwiched resolvent}

\label{analysis_sec}

In the present section, we proceed with the analysis in the general setting of graphs periodic along one axis. This section can be seen as containing the crucial bit of analysis from the point of view of attaining the main results of our study. For this reason, we start with a discussion which should motivate what follows.

In our setup, $A_\e^{(\tau)}$  acts in the Hilbert space $H=H_{\soft}\oplus H_{\stiff}$, where $H_{\soft}=L^2(\mathbb G_\soft)$ and $H_{\stiff}=L^2(\mathbb G_\stiff)$. Denote $P_{\soft}$ to be the orthogonal projection from $H$ onto $H_{\soft}$; $P_{\stiff}$ is projecting $H$ onto $H_{\stiff}$. Thus, $I=P_{\soft}+P_{\stiff}$.

Instead of the resolvent $(A_\e^{(\tau)}-z)^{-1}$, consider the sandwich $R_\e^{(\tau)}(z):=P_{\soft}(A_\e^{(\tau)}-z)^{-1} P_{\soft}$. Assume for the sake of argument that $R_\e^{(\tau)}(z)$ has a limit, as $\varepsilon\to0,$ in the uniform operator topology for $z$ in an open domain of $\mathbb C$. Further assume that for a reason yet unknown the resolvent $(A_\e^{(\tau)}-z)^{-1}$ also admits such limit. Then clearly
\begin{equation}
P_\soft (A^{(\tau)}_{\text{eff}}-z)^{-1}\bigr|_{H_\soft}=R_0^{(\tau)}(z),\quad z\in D\subset \mathbb C,
\label{R_def}
\end{equation}
where $R^{(\tau)}_0(z)$ and $A^{(\tau)}_{\text{eff}}$ are the limits introduced above. The powerful idea of simplifying the required analysis by passing over the resolvent ``sandwiched'' by orthogonal projections onto a carefully chosen subspace stems from the pioneering work of Lax and Phillips \cite{LP}, where the resulting sandwiched operator is shown to be the resolvent of a dissipative operator. This idea was later successfully extended to the case of generalised resolvents in \cite{DM}, as well as in \cite{AdamyanPavlov} with the scattering theory in mind.

The function  $R^{(\tau)}_0(z)$ defined by (\ref{R_def}) is a generalised resolvent, whereas $A^{(\tau)}_{\text{eff}}$ is its out-of-space self-adjoint extension (or \emph{Strauss dilation} \cite{Strauss}, as we will refer to it below). By a theorem of Neumark \cite{Naimark1940} (cf. \cite{Naimark1943}) this dilation is defined uniquely up to a unitary transformation of a special form, which leaves the subspace $H_\soft$ intact, provided that a minimality condition holds. This minimality condition is written as
$$
H=\bigvee_{\Im z\not =0}(A^{(\tau)}_{\text{eff}}-z)^{-1} H_\soft.
$$
This can be read as follows: one has minimality, provided that there are no eigenmodes in the effective media modelled by the operator $A^{(\tau)}_\varepsilon,$ and therefore in the medium modelled by the operator $A^{(\tau)}_{\text{eff}}$ as well, such that they never ``enter'' the soft component of the medium. A quick glance at our setup helps one immediately convince oneself that this must be true. But then it would follow that the effective medium is completely determined, up to a unitary transformation, by $R_0^{(\tau)}(z).$ We must admit that the Neumark-Strauss general theory is not directly applicable in our setting. Part of the reason for this is that $R_\e^{(\tau)}(z)$ in general does not converge (it will be shown below to admit an asymptotic expansion instead). Even in Examples (0) and (2), where one can obtain a limit proper, one still needs to prove that the resolvents $(A_\e^{(\tau)}-z)^{-1}$ converge as well. Therefore, in what follows we only use the general theory presented above as a guide. We manage to compute the required asymptotics of the resolvents, thus eliminating the non-uniqueness due to the unitary transformation mentioned above.

Inspired by the above general theory, we base our analysis on establishing the asymptotics of $R_\e^{(\tau)}(z).$ If one writes down the boundary-value problem defining $R_\e^{(\tau)}(z)$ (as we do below), one realises that it is an ODE with piecewise constant symbol independent of $\e$, so that its study does not involve homogenisation in the usual sense. Instead, one faces $\e$- and $z$-dependent boundary conditions, effectively reducing the original problem to a much simpler task of the asymptotic analysis of these boundary conditions as $\varepsilon\to0.$

In this section, we will start the asymptotic analysis of the sandwich $R_\e^{(\tau)}(z)$. This is based on the Kre\u\i n formula applied to generalised resolvents of the class considered. We start by deriving a convenient representation for $R_\e^{(\tau)}(z)$.

We assume throughout that $z\in \mathbb C$ is separated from the spectrum of the original operator family, more precisely, we assume that $z\in K_\sigma$, where
$$
K_\sigma:=\bigl\{z\in \mathbb C |\ z\in K \text{ a compact set in } \mathbb C,\ \text{dist}(z, \mathbb R)\geq \sigma>0\bigr\}.
$$
After we have established the operator-norm asymptotics of $(A_\e^{(\tau)}-z)^{-1}$ for $z\in K_\sigma,$ the result is extended to a compact set the distance of which to the spectrum of the leading order of the asymptotics is bounded below by $\sigma.$

From \eqref{eq:resolvent} and the material of Section \ref{triple_section} it follows that for all
$z\in K_\sigma$,
\begin{equation}
\label{eq:resolvent1}
(A_\e^{(\tau)}-z)^{-1}=(A_\infty^{(\tau)}-z)^{-1}-  \gamma(z)\bigl(M^{(\tau)}(z)\bigr)^{-1}\Gamma_1 (A_\infty^{(\tau)}-z)^{-1},
\end{equation}
where we have abbreviated $M_\e^{(\tau)}$ as $M^{(\tau)}$ and the decoupled operator $A_\infty^{(\tau)}$ is the restriction of $A_{\max}^{(\tau)}$ to the domain $\dom(A_\infty^{(\tau)})=\dom(A_{\max}^{(\tau)})\cap \ker \Gamma_0$. From the definition of $\Gamma_0$ it is clear that this operator is nothing but the Dirichlet decoupling associated with $A_{\max}^{(\tau)}$, {\it i.e.}  $A_\infty^{(\tau)}=\oplus_{j=1}^n A_{\infty}^{(\tau),(j)}$, where for each $j=1,\dots, n=n(\mathbb G)$ the self-adjoint operator $A_{\infty}^{(\tau),(j)}$ is defined on the edge $e^{(j)}$ of the graph $\mathbb G$, i.e., on the Hilbert space $L^2(0, l^{(j)})$, by the differential expression defining $A_\e^{(\tau)}$ on this same edge. The domain of $A_{\infty}^{(\tau),(j)}$ is set by the following:
$$
\dom(A_{\infty}^{(\tau),(j)})=\bigl\{ u\in W^{2,2}(0,l^{(j)}) \text{ such that } u(0)=u(l^{(j)})=0\bigr\}.
$$
Therefore, for the self-adjoint restrictions $A_\infty^{(\tau),\stiff(\soft)}$ of $A_{\max}^{\stiff(\soft)}$ one has the following orthogonal decomposition relative to the split $H=H_\soft \oplus H_\stiff$:
$$
A_\infty^{(\tau)}=A_\infty^{(\tau),\stiff}\oplus A_\infty^{(\tau),\soft}, \text{ where }
\dom A_\infty^{(\tau),\stiff(\soft)}:=\dom A_{\max}^{\stiff (\soft)}\cap \ker \Gamma_0^{\stiff(\soft)}.
$$

One arrives at the following lemma.
\begin{lemma}\label{lemma:gen_resolvent}
One has
\begin{equation}\label{eq:quasi-Krein1}
P_\soft(A_\e^{(\tau)}-z)^{-1}P_\soft=(A_{\infty}^{(\tau),\soft}-z)^{-1}-
 \gamma^\soft(z)\bigl(M^{(\tau)}_\soft(z)-B^{(\tau)}(z)\bigr)^{-1}\Gamma_1^\soft (A_{\infty}^{(\tau),\soft}-z)^{-1},\ \ z\in K_\sigma,
\end{equation}
where an abbreviation $M^{(\tau)}_\soft:=M^{(\tau),\soft}_\e$ is used,
$$
\begin{gathered}
B^{(\tau)}(z):=-M^{(\tau),\stiff}_\e,\qquad \gamma^\soft(z)=\bigl(\Gamma_0^\soft|_{\text{\rm ker\,}(A_{\max}^{\soft}-z)}\bigr)^{-1},
\end{gathered}
$$
and the right-hand side of (\ref{eq:quasi-Krein1}) is correctly defined for all $z\in K_\sigma,$ due to the analytic properties of the matrix functions $M^{(\tau)}_\soft$ and $B^{(\tau)}.$
\end{lemma}

\begin{proof}
It suffices to prove the claim in relation to the second summand on the right hand side of \eqref{eq:resolvent1}. This follows immediately from the identity $P_\soft \gamma(z)=\gamma^\soft(z)$, which in its turn is a consequence of the fact that $\gamma^\soft(z)$ and $\gamma(z)$ solve the boundary value problems $A_{\max}^\soft u=zu$ and $A_{\max}u=zu$ subject to boundary conditions $\Gamma_0^\soft u$ and $\Gamma_0 u$, respectively. Since $(\mathcal H, \Gamma_0^\soft, \Gamma_1^\soft)$ is a boundary triple for the maximal operator $A_{\max}^\soft$ (see Section \ref{triple_section}), one can apply the identity
$$
(\gamma^\soft)^*(\bar z)=\Gamma_1^\soft (A_{\infty}^{(\tau),\soft}-z)^{-1}
$$
(see \cite{DM}), which completes the proof.
\end{proof}

Lemma \ref{lemma:gen_resolvent} shows that the influence of the stiff component on the generalised resolvent $R_\e^{(\tau)}(z)$ is encoded in the $z$-dependent boundary conditions. The same lemma allows us to view $R_\e^{(\tau)}(z)$ at each point $z$ as the resolvent (computed at the same value $z$) of an anti-dissipative almost solvable extension of $A_{\max}^\soft$ defined by the parameterising matrix $B^{(\tau)}(z)$, cf. \cite{Strauss}. Precisely, having restricted all the operators appearing on both sides of \eqref{eq:quasi-Krein1} to the space $H_\soft$, one ascertains that the right-hand side of this formula represents the generalised resolvent $(A_{B^{(\tau)}(z)}-z)^{-1}$, {\it i.e.} the solution operator of the boundary-value problem
$$
A_{\max}^{\soft}u-zu=f, \quad f\in H_\soft
$$
with the following $z$-dependent boundary conditions described by the matrix $B^{(\tau)}(z)$ (note, that this matrix only depends on the stiff component of the medium):
$$
u\in \dom A_{\max}^{\soft}; \quad \Gamma_1^\soft u = B^{(\tau)}(z)\Gamma_0^\soft u.
$$

\begin{corollary}\label{cor:Strauss} One has the following representation for $R_\e^{(\tau)}$:
$$
R_\e^{(\tau)}(z)\equiv P_\soft(A_\e^{(\tau)}-z)^{-1}P_\soft=(A_{B^{(\tau)}(z)}-z)^{-1},
$$
where the operator on the left hand side is treated as an operator in $H_\soft$.
\end{corollary}
The problem of finding an operator asymptotics for $P_\soft(A_\e^{(\tau)}-z)^{-1}P_\soft$, as $\e\to0,$ is thus effectively reduced to the problem of finding the same for the matrix-function $B^{(\tau)}(z)$.

The following theorem proves crucial in the analysis to follow.

\begin{theorem}\label{thm:crucial}
Let $(\mathcal H, \Gamma_0^\soft, \Gamma_1^\soft)$ be a boundary triple\footnote{We do not assume here that this is necessarily {\it the same} triple as introduced for the named operator in Section \ref{triple_section}.} for $A_{\max}^{\soft}$.
Assume that for two bounded operators $B^{(\tau)}(z)$ and $B^{(\tau)}_\eff(z)$ in $\mathcal H$ which can be both assumed  $\e$- and $z$-dependent, the estimate
$$
B^{(\tau)}(z)- B^{(\tau)}_\eff (z)= O(\e^r)
$$
holds in the sense of the uniform operator norm in $\mathcal H$ for some positive $r$ and uniformly with respect to $\tau\in[-\pi,\pi)$ and $z\in K_\sigma$.
Assume further that $B^{(\tau)}$ and $B^{(\tau)}_\eff$ are double-sided operator-valued $R$-functions, so that in particular$(B^{(\tau)}_\eff (z))^*=B^{(\tau)}_\eff (\bar z)$ and the same holds for $B^{(\tau)}.$

Let $A_{B^{(\tau)}}$ and $A_{B^{(\tau)}_\eff}$ be the anti-dissipative for $z\in{\mathbb C}_+$ (dissipative for $z\in{\mathbb C}_-$) restrictions of $A_{\max}^\soft$ corresponding to the matrices $B^{(\tau)}$ and $B^{(\tau)}_\eff$, respectively.
Then the corresponding generalised resolvents admit the estimate
$$
\bigl\| (A_{B^{(\tau)}}-z)^{-1}-(A_{B^{(\tau)}_\eff}-z)^{-1}\bigr\|_{H_\soft\to H_\soft}= O(\e^r)
$$
uniformly in $\tau$ and $z\in K_\sigma$.
\end{theorem}

\begin{proof}
We use Corollary \ref{cor:Strauss} and then representation \eqref{eq:quasi-Krein1} of Lemma \ref{lemma:gen_resolvent} for both generalised resolvents.
The first summands on the right hand side of Kre\u\i n formula \eqref{eq:quasi-Krein1} cancel out; it remains to be seen that the difference of the second ones admits the estimate claimed. The second Hilbert identity yields:
\begin{multline}
(A_{B^{(\tau)}(z)}-z)^{-1}-(A_{B^{(\tau)}_\eff(z)}-z)^{-1}\\[0.3em]
=\gamma^\soft(z)\bigl(M^{(\tau)}_\soft(z)-B^{(\tau)}_\eff(z)\bigr)^{-1}(B^{(\tau)}(z)-B^{(\tau)}_\eff (z))\bigl( M^{(\tau)}_\soft(z)-B^{(\tau)}(z)\bigr)^{-1}(\gamma^\soft(\bar z))^*.
\end{multline}

On the other hand, using again the Kre\u\i n formula, where $\Gamma_0^\soft$ has been applied to both sides, we obtain
$$
\bigl(M^{(\tau)}_\soft(z)-B^{(\tau)}(z)\bigr)^{-1}(\gamma^\soft(\bar z))^*=  \Gamma_0^\soft (A_{ B^{(\tau)}(z)}^{(\tau)}-z)^{-1}.
$$
The resolvent $(A_{B^{(\tau)}(z)}^{(\tau)}-z)^{-1}$ on the right hand side is uniformly bounded in $z\in K_\sigma$ as an operator from $L^2({\mathbb G}_\soft)$ to the domain of the operator $A_{B^{(\tau)}(z)}^{(\tau)}$ equipped with its graph norm and therefore to $\dom A_{\max}^\soft$ considered as a Hilbert space; the triple property (see Definition \ref{Def_BoundTrip}) is then used to ascertain the boundedness of the operator $\Gamma_0^\soft (A_{B^{(\tau)}(z)}^{(\tau)}-z)^{-1}$.

The boundedness of the operator $\gamma^\soft(z)\bigl(M^{(\tau)}_\soft(z)-B^{(\tau)}_\eff(z)\bigr)^{-1}$ is shown in a similar way by passing to the adjoint and using the fact that $(B^{(\tau)}_\eff (z))^*=B^{(\tau)}_\eff (\bar z)$ and $(M^{(\tau)} (z))^*= M^{(\tau)} (\bar z)$. It follows that
$$
(A_{B^{(\tau)}(z)}-z)^{-1}-(A_{B^{(\tau)}_\eff(z)}-z)^{-1}=T_1(z) (B^{(\tau)}(z)-B^{(\tau)}_\eff (z))T_2(z)
$$
with $T_2(z)$ uniformly bounded from $L^2({\mathbb G}_\soft)$ to $\mathcal H$ and $T_1(z)$ uniformly bounded from
$\mathcal H$ to $L^2({\mathbb G}_\soft)$. The assumption that $B^{(\tau)}(z)-B^{(\tau)}_\eff (z)= O(\e^r)
$ is then used to complete the proof.
\end{proof}

As we will see in Section 7 below (in particular, {\it cf.} Remark \ref{rem:for_start}), the strategy suggested by the latter theorem is of a general nature, allowing to obtain a norm-resolvent asymptotics not only for $R_\e^{(\tau)}$, but also for the original resolvent $(A_\e^{(\tau)}-z)^{-1}$. It must be noted however that this theorem alone does not yield the most efficient form of the asymptotics sought, see Sections 7 and 8 for further details.

In the remaining part of this section, we apply Theorem \ref{thm:crucial} to the three examples introduced in Section
\ref{section:examples} to obtain the leading terms in asymptotic expansion of $R_\e^{(\tau)}$.

In doing so and facilitated by the fact that the problem is essentially on ODE one, we adopt the approach based on an direct calculation, as we see this rather instructive. At the same time, we refer the reader to Section 7, which allows for a simplification of relevant calculations owing to certain arguments of the general theory developed there.

\subsection*{(0) A medium with both components disconnected}
Here $\mathbb G_\soft$ is a graph of precisely one edge $e^{(2)}$. The edge $e^{(2)}$ is identified with the segment $(0,l^{(2)})$.

Define the $\e$-independent operator $R_\eff^{(\tau)}$ to be the generalised resolvent solving the following boundary problem:
\begin{equation}\label{eq:domainhom0}
\begin{aligned}
-\left(\frac d{dx}+i\tau\right)^2 u-zu&=f, \quad f\in H_\soft=L^2(\mathbb G_\soft),\ u\in W^{2,2}(\mathbb G_\soft),\\[0.3em]
u|_0- \bxit u|_{l^{(2)}}&=0, \\[0.6em]
\dtau u|_0- \bxit \dtau u|_{l^{(2)}}&= -z l^{(1)} u|_0,
\end{aligned}
\end{equation}
where
\[
\xi^{(\tau)}:=\exp(i l^{(1)} \tau), \qquad
\dtau u:=\biggl(\frac d {dx}+i\tau\biggr)u.
\]

\begin{lemma}\label{lemma:gen_res_0}
The generalized resolvent $R_\e^{(\tau)}(z):=P_\soft (A_\e^{(\tau)}-z)^{-1} P_\soft$ admits the following estimate in the uniform operator norm topology:
$$
R_\e^{(\tau)}(z)-R_\eff^{(\tau)}(z)=O(\e^2).
$$
This estimate is uniform in $\tau\in[-\pi,\pi)$ and $z\in K_\sigma$.
\end{lemma}

Note that in the case considered $R_\e^{(\tau)}$ has the uniform limit $R_\eff^{(\tau)}$ as $\e\to 0$.

The \emph{proof} of Lemma \ref{lemma:gen_res_0} is in Appendix A.

\subsection*{(1) A case of connected stiff component}
Here $\mathbb G_\soft$ is a graph of precisely one edge $e^{(2)}$. The edge $e^{(2)}$ is identified with the segment $(0,l^{(2)})$.

Define the $\e$-dependent operator $R_\eff^{(\tau)}$ to be the generalised resolvent solving the following boundary problem:
\begin{equation}\label{eq:domainhom1}
\begin{aligned}
-\left(\frac d{dx}+i\tau\right)^2 u-zu&=f, \quad f\in H_\soft=L^2(\mathbb G_\soft),\ u\in W^{2,2}(\mathbb G_\soft),\\[0.5em]
u|_{0}+\frac{\bxit}{|\xit|}u|_{l^{(2)}}&=0,\\[0.6em]
\dtau u|_{0}+\frac {\bxit}{|\xit|}\dtau u|_{l^{(2)}}&=\left(\biggl(\frac{l^{(1)}}{a_1^2}+\frac{l^{(3)}}{a_3^2}\biggr)^{-1}\biggl(\frac{\tau}{\e}\biggr)^2-(l^{(1)}+l^{(3)})z\right)u|_{0}.
\end{aligned}
\end{equation}
Here\footnote{Note that, in order to simplify notation, in this example we re-define the function $\xi^{(\tau)}$ compared to Example (0), in hope that this will not cause confusion.}
\begin{equation}
\xit=-\frac{a_1^2}{l^{(1)}}\exp(i\tau(l^{(1)}+l^{(3)}))-\frac{a_3^2}{l^{(3)}}\exp(-i\tau l^{(2)}),\qquad
\dtau u:=\biggl(\frac d {dx}+i\tau\biggr) u.
\label{star_form}
\end{equation}

\begin{lemma}\label{lemma:gen_res_1}
The generalized resolvent $R_\e^{(\tau)}(z):=P_\soft (A_\e^{(\tau)}-z)^{-1} P_\soft$ admits the following estimate in the uniform operator norm topology:
$$
R_\e^{(\tau)}(z)-R_\eff^{(\tau)}(z)=O(\e^2).
$$
This estimate is uniform in $\tau\in[-\pi,\pi)$ and $z\in K_\sigma$.
\end{lemma}

Note that in the case considered, unlike in Example (0), $R_\e^{(\tau)}$ only admits asymptotics in the uniform operator norm topology. The operator $R_\eff^{(\tau)}$ still depends on $\e$ in $(\tau/\e)^2$; this dependance cannot be dropped (cf. \cite{BirmanSuslina}, \cite{Friedlander}, \cite{Physics}). This is in fact to be expected. In particular, we point out the result of \cite{Friedlander} concerning the integrated density of states of a problem directly related to the model considered. In a nutshell, the integrated density of states of $\oplus\int_\tau A_\e^{(\tau)}$ has to vanish everywhere away from a discrete set of points, which can be described as the spectrum of the decoupled problem (in our setting, this is nothing but the point spectrum of $A_\e^{(-\pi)}$). In other words, the spectral bands of $\oplus\int_\tau A_\e^{(\tau)}$ are of a pseudogap nature,  
making the existence of a limit proper impossible.

The \emph{proof} of Lemma \ref{lemma:gen_res_1}  can be found in Appendix B.

\subsection*{(2) A case of connected soft component}
Here $\mathbb G_\soft$ is a graph comprising two edges $e^{(j)}$, $j=1,2$. Henceforth we write $u_j=u|_{e^{(j)}}$ and $f_j=f|_{e^{(j)}}$. The edges $e^{(j)}$ are identified with segments $[0,l^{(j)}]$, $j=1,2$.
Define the $\e$-independent operator $R_\eff^{(\tau)}$ to be the generalised resolvent solving the following boundary problem:
\begin{equation}\label{eq:domainhom2}
\begin{aligned}
&-a_j^2\left(\frac d{dx}+i\tau\right)^2 u_j-zu_j=f_j, \quad j=1,2,\ f\in H_\soft=L^2(\mathbb G_\soft),\ u\in W^{2,2}(\mathbb G_\soft),\\[0.5em]
&u_2|_0=\bxittwo  u_2|_{l^{(2)}}=\bxitone   u_1|_0=u_1|_{l^{(1)}},\\[0.9em]
&a_2^2 \dtau u_2|_0- a_2^2 \bxittwo  \dtau u_2|_{l^{(2)}}+ a_1^2 \bxitone   \dtau u_1|_0-a_1^2 \dtau u_1|_{l^{(1)}}=-z l^{(3)} u_2|_0
\end{aligned}
\end{equation}
where
\[
\xit_1=\exp(-i\tau (l^{(2)}+l^{(3)})),\quad\xit_2=\exp(-i\tau l^{(2)}),\quad
\dtau u_j:=\biggl(\frac d {dx}+i\tau\biggr) u_j,\ \ j=1,2.
\]

\begin{lemma}\label{lemma:gen_res_2}
The generalized resolvent $R_\e^{(\tau)}(z):=P_\soft (A_\e^{(\tau)}-z)^{-1} P_\soft$ admits the following estimate in the uniform operator norm topology:
$$
R_\e^{(\tau)}(z)-R_\eff^{(\tau)}(z)=O(\e^2).
$$
This estimate is uniform in $\tau\in[-\pi,\pi)$ and $z\in K_\sigma$.
\end{lemma}

Note that in the case considered, like in Example (0), $R_\e^{(\tau)}$ has the uniform limit $R_\eff^{(\tau)}$ as $\e\to 0$.

The \emph{proof} of Lemma \ref{lemma:gen_res_2} is obtained by repeating the calculation of Appendix A with $\xit$ replaced by $\xit_2$, $a_1$ replaced by $a_3$, and $l^{(1)}$ replaced by $l^{(3)}$ in the arguments of $\cot,$ $\csc.$  Indeed, this follows from the comparison of the expressions for $M_\e^{(\tau),\stiff}$ in \eqref{D2N_2} and \eqref{D2N_0}.

\begin{remark}
Note that in all examples considered the problem for $R_\eff^{(\tau)}$ can be viewed as the problem on a graph $\widetilde {\mathbb G}_\soft$ which is obtained by contracting the original graph $\mathbb G$ in the sense of \cite{Tutte} along all the edges comprising its stiff component $\mathbb G_\stiff$. The contraction mentioned above can be viewed as a process of removing every edge $e\in \mathbb G_\stiff$ by sending its length to zero. The left and right boundary vertices of $e$ are identified to form a new vertex. The $z$-dependent matching condition at this vertex admits the form of the so-called $\delta$-type (see, e.g., \cite{Yorzh3}) with a coupling constant depending on $z$ (cf. \cite{KuchmentZeng}).
\end{remark}

\begin{remark}
Following from the analysis contained in Appendices A, B, the main difference between Example (1) and Examples (0), (2) lies in the behaviour of the least eigenvalue of the matrix $B^{(\tau)}(0)$. The named eigenvalue can be interpreted as the least Steklov eigenvalue of $A_{\max}^\stiff$, i.e., the least (by absolute value) $\kappa$ such that the problem
$$
\begin{aligned}
A_{\max}^{\stiff}u&=0, \quad u\in W^{2,2}(\mathbb G_\stiff),\\[0.4em]
\Gamma_1^\stiff u&=\kappa \Gamma_0^\stiff u.
\end{aligned}
$$
admits a non-trivial solution.

It is easily seen that the least Steklov eigenvalue in the case of connected $\mathbb G_\stiff$  is identically zero if and only if the stiff component of ${\mathbb G}_{\text{\rm per}}$ is disconnected.
\end{remark}

\section{Asymptotic behaviour of the resolvent $(A_\e^{(\tau)}-z)^{-1}$}

\subsection{Derivation of the asymptotics}

In this section we prove that the original resolvent family $(A_\e^{(\tau)}-z)^{-1}$ under a mild additional assumption admits, along with the generalised resolvent $R_\e^{(\tau)}=P_\soft(A_\e^{(\tau)}-z)^{-1}P_\soft$, a uniform asymptotics. After this is established, we proceed with finding an explicit form of the latter.

\begin{theorem}\label{thm:NRA}
Let for $r\leq 2$ and a generalised resolvent $R_\eff^{(\tau)}(z)$ the estimate
$$
\bigl\| R_\e^{(\tau)}(z)-R_\eff^{(\tau)}(z)\bigr\|_{H_\soft\to H_\soft}= O(\e^r)
$$
hold uniformly in $\tau$ and $z\in K_\sigma$. Let further\footnote{We recall that by \cite{Strauss}, {\it cf.} Corollary \ref{cor:Strauss} and Theorem \ref{thm:crucial}, $R^{(\tau)}_\eff$ is  represented as $(A^{(\tau)}(z)-z)^{-1}$, where $A^{(\tau)}(z)$ is a ($z$-dependent) operator in $H_\soft$. } $R^{(\tau)}_\eff(z)\equiv (A^{(\tau)}(z)-z)^{-1}$ be such that $A^{(\tau)}(z)\subset A_{\max}^\soft$ for all $z\in \mathbb C_\pm$.

Then the resolvent $(A_\e^{(\tau)}-z)^{-1}$ admits the following asymptotics in the uniform operator-norm topology:
$$
(A_\e^{(\tau)}-z)^{-1}=\mathcal R^{(\tau)}_\eff + O(\e^r),
$$
where the operator $\mathcal R^{(\tau)}_\eff$, which is allowed to depend on $\e$, admits the following representation relative to the decomposition $H=H_\stiff\oplus H_\soft$:
\begin{equation}\label{eq:NRA}
\mathcal R^{(\tau)}_\eff=
\begin{pmatrix}
\mathcal R^{(\tau)}_\eff(z)&\ \ \Bigl(\mathfrak K_{\bar z}^{(\tau)}\bigl[R_\eff^{(\tau)}(\bar z)-(A_\infty^{(\tau),\soft}-\bar z)^{-1}\bigr]\Bigr)^*\Pi_\stiff^*\\[0.9em] \Pi_\stiff\mathfrak{K}^{(\tau)}_z \bigl[R_\eff^{(\tau)}(z)-(A_\infty^{(\tau),\soft}-z)^{-1}\bigr] & \ \ \Pi_\stiff\mathfrak K_{z}^{(\tau)}\Bigl(\mathfrak K_{\bar z}^{(\tau)}\bigl[R_\eff^{(\tau)}(\bar z)-(A_\infty^{(\tau),\soft}-\bar z)^{-1}\bigr]\Bigr)^*\Pi_\stiff^*
\end{pmatrix}.
\end{equation}
Here $\mathfrak{K}^{(\tau)}_z:=\Gamma_0^\soft|_{\mathfrak N_z}$, where $\mathfrak N_z:=\ker(A_{\max}^\soft-z)$, $z\in \mathbb C_\pm$, and $\Pi_\stiff:=\gamma^\stiff(0)$.
\end{theorem}

\begin{proof}
Consider the operator $\mathfrak{K}^{(\tau)}_z$.  By \cite{DM}, see also \cite{Ryzhov_later}, it is one-to-one and admits a bounded inverse $\gamma^\soft(z)$. Using the fact that the operator  $\mathfrak{K}^{(\tau)}_z$ is finite-dimensional, it follows that this operator is itself bounded
as an operator from $\mathfrak N_z$ equipped with the norm of $H_\soft$ to $\mathcal H$. A simple computation shows that its bound is uniform in $\tau$ and $z\in K_\sigma$.

Using the Kre\u\i n resolvent formula \eqref{eq:resolvent}, we regularise $R_\varepsilon^{(\tau)}$ and write
\begin{multline*}
P_\stiff (A_\e^{(\tau)} -z)^{-1}P_\soft = -\gamma^\stiff(z)\bigl(M^{(\tau)}_\soft(z)-B^{(\tau)}(z)\bigr)^{-1}\bigl(\gamma^\soft(\bar z)\bigr)^*\\[0.4em]
=-\gamma^\stiff (z)\Gamma_0^\soft \gamma^\soft(z)\bigl(M^{(\tau)}_\soft(z)-B^{(\tau)}(z)\bigr)^{-1} \bigl(\gamma^\soft(\bar z)\bigr)^*=
\gamma^\stiff (z) \Gamma_0^\soft\bigl[R_\e^{(\tau)}-(A_\infty^{(\tau),\soft}-z)^{-1}\bigr]\\[0.4em]
=\gamma^\stiff (z) \mathfrak{K}^{(\tau)}_z\bigl[R_\e^{(\tau)}-(A_\infty^{(\tau),\soft}-z)^{-1}\bigr].
\end{multline*}
Here in the second equality we take advantage of the fact that $\Gamma_0^\soft\gamma^\soft(z)=I.$
Furthermore, using the representation (see \cite{DM,Ryzhov_later})
\[
\gamma^\stiff(z)=\bigl(1-z (A_\infty^{(\tau),\stiff})^{-1}\bigr)^{-1}\Pi_\stiff
\]
together with the fact that the operator $(A_\infty^{(\tau),\stiff})^{-1}$ admits an obvious estimate
$
\bigl\|(A_\infty^{(\tau),\stiff})^{-1}\bigr\|=O(\e^2)
$
uniformly in $\tau$, we conclude that
$$
\gamma^\stiff(z)-\Pi_\stiff = O(\e^2).
$$
Since by assumption
$$
(A_{\max}^\soft -z)\bigl[R_\eff^{(\tau)}f-(A_\infty^{(\tau),\soft}-z)^{-1}f\bigr]=(A_{\max}^\soft -z)(A_z-z)^{-1}f-f=0,\quad f\in H_\soft,
$$
the following asymptotic formula holds:
\begin{equation}\label{eq:M12}
P_\stiff (A_\e^{(\tau)} -z)^{-1}P_\soft - \Pi_\stiff\mathfrak{K}^{(\tau)}_z\bigl[R_\e^{(\tau)}-(A_\infty^{(\tau),\soft}-z)^{-1}\bigr]=O(\e^r)
\end{equation}
uniformly in $\tau$ and $z\in K_\sigma$.

Passing over to the top right entry in \eqref{eq:NRA}, we write
\begin{align*}
P_\soft (A_\e^{(\tau)} -z)^{-1}P_\stiff&= -\gamma^\soft(z)\bigl(M^{(\tau)}_\soft(z)-B^{(\tau)}(z)\bigr)^{-1}(\gamma^\stiff)^*(\bar z)\\[0.4em]
&=\Bigl(\mathfrak K_{\bar z}^{(\tau)}\bigl[R_\e^{(\tau)}(\bar z)-(A_\infty^{(\tau),\soft}-\bar z)^{-1}\bigl]\Bigr)^*(\gamma^\stiff(\bar z))^*\\[0.4em]
&=\Bigl(\mathfrak K_{\bar z}^{(\tau)}\bigl[R_\e^{(\tau)}(\bar z)-(A_\infty^{(\tau),\soft}-\bar z)^{-1}\bigr]\Bigr)^*\Pi_\stiff^* (1-z (A_\infty^{(\tau),\stiff})^{-1})^{-1},
\end{align*}
from where the claim pertaining to the named entry follows by a virtually unchanged argument.
We remark that one could of course rewrite the right-hand side of the last expression as
$$
[R_\e^{(\tau)}(z)-(A_\infty^{(\tau),\soft}-z)^{-1}](\mathfrak K_{\bar z}^{(\tau)})^*(\gamma^\stiff(\bar z))^*
$$
so that its asymptotics is clearly the conjugate adjoint of the bottom left entry of the matrix in \eqref{eq:NRA}, but there is no value in doing so since $(\mathfrak K_{\bar z}^{(\tau)})^*$ admits no concise representation.

Finally, for the $(1,1)$ matrix element in \eqref{eq:NRA} we have
$$
P_\stiff (A_\e^{(\tau)} -z)^{-1}P_\stiff =
(A_\infty^{(\tau),\stiff}-z)^{-1}+\gamma^\stiff(z) \mathfrak K_{z}^{(\tau)}\Bigl(\mathfrak K_{\bar z}^{(\tau)}\bigl[R_\e^{(\tau)}(\bar z)-(A_\infty^{(\tau),\soft}-\bar z)^{-1}\bigr]\Bigr)^*\bigl(\gamma^\stiff(\bar z)\bigr)^*,
$$
which is used to complete the proof.
\end{proof}

\subsection{Assumptions and remarks}

We remark that the operator $\mathcal R_\eff^{(\tau)}$ cannot in general be claimed to be a resolvent of a self-adjoint operator acting in $H_\eff=\clos(\ran \Pi_\stiff) \oplus H_\soft$. If it were, one would be able to immediately invoke the argument of \cite{Naimark1940,Naimark1943}. As it stands however, one has to opt for the construction to follow.

Henceforth in the present Section, we adopt the following additional assumptions.
\subsection*{Assumptions}
\begin{itemize}
  \item[(i)] As in Theorem \ref{thm:NRA}, $\ran R_\eff^{(\tau)}(z)\subset \dom A_{\max}^\soft$ and moreover, $\Gamma_0^\soft \ran R_\eff^{(\tau)}(z)\subset \P \mathcal H$ for all $z\in \mathbb C_\pm$ and some orthogonal projection $\P$ in $\mathcal H$. The projection $\Port$ is introduced as the orthogonal projection complementary ot $\P.$
  \item[(ii)] At every value of $\tau$ and $z\in K_\sigma$, the operator $A_z^{(\tau)}$ such that $(A_z^{(\tau)}-z)^{-1}=R_\eff^{(\tau)}$ is an ($z$-dependent) almost solvable extension of the operator $\breve{A}_{\min}^{\soft}:=(\breve{A}_{\max}^{\soft})^*$, where $\breve{A}_{\max}^{\soft}$ is the restriction of the operator $A_{\max}^\soft$ to the domain
  $\dom \breve{A}_{\max}^{\soft}:=\dom A_{\max}^\soft \cap\ker(\Port\Gamma_0)$, relative to the triple $\{\breve{\mathcal{H}},\breve{\Gamma}_0^\soft,\breve{\Gamma}_1^\soft\}$. Here
      \begin{equation}\label{eq:truncated_triple}
      \breve{\mathcal{H}}:=\P\mathcal H, \quad \breve{\Gamma}_0^\soft:=\P \Gamma_0^\soft, \quad \breve{\Gamma}_1^\soft:=\P \Gamma_1^\soft.
      \end{equation}
        In what follows we equip all the objects pertaining to this triple with a breve on top.
  \item[(iii)] The parameterising matrix $B^{(\tau)}_\eff(z)$ of $A_z^{(\tau)}$ in terms of Section \ref{triple_section}, i.e., the matrix such that $A_z^{(\tau)}$ is a restriction of $\breve{A}_{\max}^{\soft}$ to the domain $\dom A_z^{(\tau)}=\{u\in\dom \breve{A}_{\max}^{\soft}: \breve{\Gamma}_1^\soft u = B^{(\tau)}_\eff(z)\breve{\Gamma}_0^\soft u\}$, is linear in $z$.
\end{itemize}

\begin{remark}\label{rem:assumptions}
1. The notation $\P,\Port$ refers in fact to the same projections as introduced in Appendices A and B, allowing us to keep the same symbols here.

2. The assumption (i) contains no further restrictions compared to those imposed by Theorem \ref{thm:NRA}. Indeed, the case $\P=I$ is not excluded.

3. The assumptions (ii) and (iii) are non-restrictive. Indeed, Theorem \ref{thm:crucial} allows to interpret the problem of establishing the asymptotic behaviour of $R_\e^{(\tau)}$ as that of finding the non-decreasing in $\e$ terms of the asymptotic expansion of $M_\e^{(\tau),\stiff}$. On the other hand, the latter is (up to multiplication by $1/\e^2$) a Dirichlet-to-Neumann map of a uniformly elliptic problem evaluated at the point $\e^2 z$. This allows us to utilise a standard expansion, which is obtained by a straightforward computation:
\begin{equation}\label{eq:standard_expansion}
M_\e^{(\tau),\stiff}(z)=\e^{-2}\bigl(M_0^{(\tau)}+\e^2 z M_1^{(\tau)}\bigr)+O(\e^2)
\end{equation}
for all $z\in K_\sigma$ and all $\tau$. Theorem \ref{thm:crucial} then permits us to drop the error term $O(\e^2)$ in the last expansion, leaving a linear operator-function in $z$, as required ({\it cf.} \cite{Physics,Friedlander,BirmanSuslina}). In the process, one ends up with an almost solvable extension of the operator $A_{\max}^\soft$. After this is done, one follows the argument of Appendix A, introducing the triple \eqref{eq:projector-triple}. A straightforward application of the Schur-Frobenius formula leads to an asymptotics of the type \eqref{eq:neededfor7}. Ultimately, one ``truncates'' the triple to the one defined in \eqref{eq:truncated_triple} to find oneself in the setup described by Assumptions (i)--(iii) above.
\end{remark}

\subsection{Strauss dilation}

Theorem \ref{thm:crucial} allows us to keep only the first two terms in the expansion (\ref{eq:standard_expansion}) for  $B^{(\tau)}(z)\equiv -M_\e^{(\tau),\stiff}(z),$ which are affine in $z.$ This, in turn, allows us to explicitly construct the Strauss dilation \cite{Strauss} of the generalised resolvent  defined by these leading-order terms of $B^{(\tau)}(z)$,  ({\it cf.} \cite{KuchmentZeng}). We then compute the resolvent of the latter in view of comparing it with the expression \eqref{eq:NRA}.

\begin{definition}\label{defn:gen_dilation}
Assume (i)--(iii) above. Consider the Hilbert space $\mathfrak H=H_\soft\oplus H^{(1)}$, where $H^{(1)}$ is an auxiliary Hilbert space. Let $\Pi$ be a bounded and boundedly invertible operator from $\breve{\mathcal H}$ to $H^{(1)}.$
Let
$$
\dom \mathcal A:=\Bigl\{(u,u^{(1)})^\top\in H_\soft\oplus H^{(1)}: u\in \dom \breve{A}_{\max}^\soft, u^{(1)}=\Pi\breve{\Gamma}_0^\soft u\Bigr\}.
$$
Clearly, $\dom \mathcal A$ thus defined is dense in $\mathfrak H$. We introduce a linear operator $\mathcal A$ on this domain by setting
\begin{equation}\label{eq:strauss_op}
  \mathcal A \binom{u}{u^{(1)}}:=\binom{\breve{A}_{\max}^{\soft} u}{-(\Pi^*)^{-1} \breve{\Gamma}_1^\soft u + \mathcal B u^{(1)}},
\end{equation}
where $\mathcal B$ is assumed to be a bounded operator in $H^{(1)}$.
\end{definition}

We have the following statement.
\begin{lemma}\label{lemma:self-adjointness}
 The operator $\mathcal A$ is symmetric on the domain $\dom \mathcal A$ if and only if $\mathcal B$ is self-adjoint in $H^{(1)}$.
\end{lemma}

\begin{proof}
On the one hand, for all $(u,u^{(1)})^\top\in \dom \mathcal A$ and for $(v,v^{(1)})^\top\in \dom \mathcal A$ one has
\begin{multline*}
 \left \langle \mathcal A \binom{u}{u^{(1)}},\binom{v}{v^{(1)}}\right\rangle =\langle \breve{A}_{\max}^\soft u,v\rangle -
  \langle (\Pi^*)^{-1} \breve{\Gamma}_1^\soft u,v^{(1)}\rangle + \langle \mathcal B \Pi \breve{\Gamma}_0^\soft u,v^{(1)} \rangle\\
  =\langle  u,\breve{A}_{\max}^\soft v\rangle +\langle \breve{\Gamma}_1^\soft u, \breve{\Gamma}_0^\soft v \rangle
  -\langle \breve{\Gamma}_0^\soft u, \breve{\Gamma}_1^\soft v \rangle-
  \langle (\Pi^*)^{-1} \breve{\Gamma}_1^\soft u,\Pi \breve{\Gamma}_0^\soft v\rangle + \langle \mathcal B \Pi \breve{\Gamma}_0^\soft u,\Pi \breve{\Gamma}_0^\soft v \rangle\\[0.4em]
  =\bigl\langle  u,\breve{A}_{\max}^\soft v\bigr\rangle -\bigl\langle \breve{\Gamma}_0^\soft u, \breve{\Gamma}_1^\soft v\bigr\rangle+ \bigl\langle \Pi \breve{\Gamma}_0^\soft u, {\mathcal B}^*\Pi \breve{\Gamma}_0^\soft v\bigr\rangle,
\end{multline*}
where the triple property for $\breve{A}_{\max}^\soft$ has been used.

On the other hand,
\begin{equation*}
  \left\langle \binom{u}{u^{(1)}},\mathcal A \binom{v}{v^{(1)}}\right\rangle
  =\langle  u,\breve{A}_{\max}^\soft v\rangle -
  \langle \breve{\Gamma}_0^\soft u, \breve{\Gamma}_1^\soft v\rangle+
 \bigl \langle  \Pi \breve{\Gamma}_0^\soft u,\mathcal B\Pi \breve{\Gamma}_0^\soft v\bigr\rangle,
\end{equation*}
and the claim follows by comparison.
\end{proof}

In fact, one can claim self-adjointness of $\mathcal A$ iff $\mathcal B$ is self-adjoint in $H^{(1)}$. This follows from the next Theorem via an explicit construction of the resolvent $(\mathcal A -z)^{-1}$.

\begin{theorem}\label{thm:dilation_resolvent}
Assume $\mathcal B=\mathcal B^*$. Then $\mathcal A$ is self-adjoint, and its resolvent $(\mathcal A -z)^{-1}$ is defined at all $z\in \mathbb C_\pm$ by the following expression, relative to the space decomposition $\mathfrak H=H_\soft \oplus H^{(1)}$ (cf. \eqref{eq:NRA}).
\begin{equation}\label{eq:NRA1}
(\mathcal A -z)^{-1}=
\begin{pmatrix}
  R(z) & \Bigl(\breve{\mathfrak{K}_{\overline {z}}}\bigl[R(\bar z)-(\breve{A}_\infty^\soft-\bar z)^{-1}\bigr]\Bigr)^* \Pi^* \\[0.6em]
   \Pi \breve{\mathfrak{K}}_z\bigl[R(z)-(\breve A_{\infty}^{\soft}-z)^{-1}\bigr] & \Pi\breve{\mathfrak{K}}_z\Bigl(\breve{\mathfrak{K}_{\overline{z}}}\bigl[R(\bar z)-(\breve{A}_\infty^\soft-\bar z)^{-1}\bigr]\Bigr)^* \Pi^*
\end{pmatrix}.
\end{equation}
Here $R(z)$ is a generalised resolvent in $H_\soft$ defined as $R(z)=(A_{B}-z)^{-1}$ with $B\equiv B(z):=\Pi^* (\mathcal B-z)\Pi$ relative to the triple $(\breve {\mathcal H},\breve{\Gamma}_0^\soft, \breve{\Gamma}_1^\soft)$; the bounded operator $\breve{\mathfrak K}_z$ is defined as $\breve{\mathfrak K}_z:=\breve{\Gamma}_0^\soft |_{\breve{\mathfrak N}_z}$, where $\breve{\mathfrak N}_z:=\ker (\breve{A}_{\max}^\soft-z)$. Finally, the operator $\breve A_{\infty}^{\soft}$ is the restriction of $\breve{A}_{\max}^\soft$ to the set $\dom \breve{A}_{\max}^\soft\cap\ker(\breve{\Gamma}_0^\soft).$
\end{theorem}

\begin{proof}
We start by considering the problem
$$
(\mathcal A - z) \binom{u}{u^{(1)}}=\binom{f}{0},
$$
which is rewritten as
\begin{align*}
\breve{A}_{\max}^\soft u-zu &=f,\\[0.4em]
-(\Pi^*)^{-1}\breve{\Gamma}_1^\soft u + \mathcal B u^{(1)}-z u^{(1)}&=0.
\end{align*}
The second equation admits the form
$$
\breve{\Gamma}_1^\soft u=\Pi^* (\mathcal B -z) \Pi \breve{\Gamma}_0^\soft u,
$$
and thus for $B(z):=\Pi^* (\mathcal B -z)\Pi$ one has $u=R(z)f$.

Since $(u,u^{(1)})^\top\in \dom \mathcal A,$ one also obtains from the same calculation that $u^{(1)}=\Pi \breve{\Gamma}_0^\soft u$ and therefore
$$
u^{(1)}=\Pi \breve{\Gamma}_0^\soft R(z) f=\Pi \breve{\mathfrak K}_z\bigl[R(z)-(\breve A_{\infty}^{\soft}-z)^{-1}\bigr],
$$
thus completing the proof for the first column of the matrix representation \eqref{eq:NRA1}.

We proceed with establishing the second column. Considering the problem
\begin{equation}\label{eq:solving_second_column}
(\mathcal A - z) \binom{u}{u^{(1)}}=\binom{0}{f^{(1)}},
\end{equation}
we rewrite it as
\begin{align*}
\breve{A}_{\max}^\soft u-zu &=0,\\[0.4em]
-(\Pi^*)^{-1}\breve{\Gamma}_1^\soft u + \mathcal B u^{(1)}-z u^{(1)}&=f^{(1)}.
\end{align*}
The second equation admits the form
$$
\breve{\Gamma}_1^\soft u=\Pi^* (\mathcal B -z) \Pi \breve{\Gamma}_0^\soft u-\Pi^*f^{(1)}.
$$
Pick a function $v_f\in\dom \breve A_{\max}^\soft\cap\ker(\breve{\Gamma}_0^\soft)=\dom \breve A_\infty^\soft$ that satisfies $\breve{\Gamma}_1^\soft v_f=\Pi^* f^{(1)}$. Such a choice is possible due to the surjectivity property of the triple. We look for a solution to \eqref{eq:solving_second_column} such that its first component $u$ admits the form $u=v-v_f$. For $v,$ using the fact that by construction $v_f\in \dom \breve A_{\infty}^\soft,$ we then obtain $v=R(z)(\breve A_{\infty}^\soft-z)v_f,$ and therefore
$$
u=R(z)(\breve A_{\infty}^\soft-z)v_f-v_f.
$$
Letting $u_f:=(\breve A_{\infty}^\soft-z)v_f,$  this amounts to
$$
u=\bigl[R(z)-(\breve A_{\infty}^\soft-z)^{-1}\bigr]u_f.
$$
Using the Kre\u\i n formula, we have
$$
u=\bigl[R(z)-(\breve A_{\infty}^\soft-z)^{-1}\bigr]u_f=-\breve{\gamma}^\soft(z)\bigl(\breve{M}_\soft(z)-B(z)\bigr)^{-1}\breve{\Gamma}_1^\soft v_f =
-\breve{\gamma}^\soft(z)\bigl(\breve{M}_\soft(z)-B(z)\bigr)^{-1}\Pi^* f^{(1)},
$$
where $\breve{\gamma}(z)$ and $\breve{M}_\soft(z)$ are the solution operator and the $M$-matrix of $\breve A_{\max}^\soft$, pertaining to the triple $(\breve{\mathcal H},\breve{\Gamma}_0^\soft,\breve{\Gamma}_1^\soft)$, respectively. Proceeding as in the proof of Theorem \ref{thm:NRA}, we rewrite the latter expression as:
$$
u=\Bigl(\breve{\mathfrak{K}_{\overline{z}}}\bigl[R(\bar z)-(\breve{A}_\infty^\soft-\bar z)^{-1}\bigr]\Bigr)^* \Pi^* f^{(1)},
$$
and thus complete the proof of representation \eqref{eq:NRA1}. The fact that $\mathcal A$ is self-adjoint in $\mathfrak H$ now follows from Lemma \ref{lemma:self-adjointness}.
\end{proof}

\subsection{``Self-adjointness" of the asymptotics}

Introduce the truncated triple $(\breve{\mathcal H},\breve{\Gamma}_0^\stiff, \breve{\Gamma}_1^\stiff)$ for $\breve{A}_{\max}^\stiff$ by the formulae (cf. \eqref{eq:truncated_triple})
      \begin{equation}\label{eq:truncated_triple_stiff}
      \breve{\mathcal{H}}:=\P\mathcal H, \quad \breve{\Gamma}_0^\stiff:=\P \Gamma_0^\stiff, \quad \breve{\Gamma}_1^\stiff:=\P \Gamma_1^\stiff,
      \end{equation}
where $\breve{A}_{\max}^{\stiff}$ is the restriction of the operator $A_{\max}^\stiff$ to the domain $\dom \breve{A}_{\max}^{\stiff}:=\dom A_{\max}^\stiff \cap\ker(\Port\Gamma_0^\stiff).$

We have by the definition of solution operators: $\breve{\gamma}^\stiff (z)=\gamma^\stiff(z) \P$. Indeed, this follows from $\Gamma_0^\stiff \gamma^\stiff(z)\P=\P$. Then $\breve{\Pi}_\stiff=\Pi_\stiff\P,$ where, in line with the preceding notation, we have denoted $\breve{\Pi}_\stiff:=\breve{\gamma}^\stiff(0).$

\begin{theorem}\label{cor:NRA} Under the assumptions (i)--(iii), let $\bigl(B^{(\tau)}_\eff\bigr)'_z(0)=-(\breve{\Pi}_\stiff)^*\breve{\Pi}_\stiff$.
Then the asymptotics $\mathcal R^{(\tau)}_\eff$ of $(A_\e^{(\tau)}-z)^{-1}$ provided by Theorem \ref{thm:NRA} is the resolvent $(\mathcal{A}^{(\tau)}_\eff-z)^{-1}$ of a self-adjoint operator $\mathcal{A}^{(\tau)}_\eff$, introduced by Definition \ref{defn:gen_dilation} with $\Pi=\Pi^{(\tau)}$ chosen as $\Pi^{(\tau)}=\breve{\Pi}_\stiff$ and $\mathcal B=\mathcal{B}^{(\tau)}$ chosen so that
$$
B^{(\tau)}_\eff(0)=(\breve{\Pi}_\stiff)^*\mathcal{B}^{(\tau)}\breve{\Pi}_\stiff.
$$
\end{theorem}

\begin{proof}

Under the assumptions one has
$$
(A_\infty^{(\tau),\soft}-z)^{-1}=(\breve{A}_\infty^{(\tau),\soft}-z)^{-1}; \quad \breve{\mathfrak{K}}_z^{(\tau)}=\breve{\Gamma}_0^\soft|_{\breve{\mathfrak{N}}_z}=\P\Gamma_0^\soft|_{\breve{\mathfrak{N}}_z}=\P\mathfrak{K}_z^{(\tau)}|_{\breve{\mathfrak{N}}_z}.
$$
Since at the same time $[R_\eff^{(\tau)}(\bar z)-(A_\infty^{(\tau),\soft}-\bar z)^{-1}]f\in \breve{\mathfrak{N}}_z$, one has
$$
\Pi_\stiff\mathfrak{K}^{(\tau)}_z [R_\eff^{(\tau)}(z)-(A_\infty^{(\tau),\soft}-z)^{-1}]=
\breve{\Pi}_\stiff\breve{\mathfrak{K}}^{(\tau)}_z [R_\eff^{(\tau)}(z)-(A_\infty^{(\tau),\soft}-z)^{-1}].
$$

Therefore, the representation \eqref{eq:NRA} provided by Theorem \ref{thm:NRA} admits the form
\begin{equation}\label{eq:NRA_breve}
\mathcal R^{(\tau)}_\eff=
\begin{pmatrix}
 \mathcal R^{(\tau)}_\eff(z) &\ \  \Bigl(\breve{\mathfrak {K}}_{\overline {z}}^{(\tau)}\bigl[R_\eff^{(\tau)}(\bar z)-(\breve{A}_\infty^{(\tau),\soft}-\bar z)^{-1}\bigr]\Bigr)^*\breve{\Pi}_\stiff^*
 \\[0.9em]
\breve{\Pi}_\stiff\breve{\mathfrak{K}}^{(\tau)}_z\bigl[R_\eff^{(\tau)}(z)-(\breve{A}_\infty^{(\tau),\soft}-z)^{-1}\bigr] &\ \   \breve{\Pi}_\stiff\breve{\mathfrak {K}}_{z}^{(\tau)}\Bigl(\breve{\mathfrak {K}}_{\overline {z}}^{(\tau)}\bigl[R_\eff^{(\tau)}(\bar z)-(\breve{A}_\infty^{(\tau),\soft}-\bar z)^{-1}\bigr]\Bigr)^*\breve{\Pi}_\stiff^*
\end{pmatrix}.
\end{equation}
Comparing this with the statement of Theorem \ref{thm:dilation_resolvent}, we note that $\mathcal R^{(\tau)}_\eff$ is the resolvent of the stated self-adjoint operator provided that $B_\eff^{(\tau)}=\Pi^*(\mathcal B^{(\tau)} -z)\Pi$.
\end{proof}

\begin{remark}\label{rem:for_start}
The above theorem opens up a multitude of ways to offer a norm-resolvent asymptotics for $(A_\e^{(\tau)}-z)^{-1}$. Seemingly the simplest would follow if one chose $\P=I$. In this case, Theorem \ref{thm:crucial} provides the following recipe for the asymptotics sought. One starts with $B_\e^{(\tau)}(z):=-M_\e^{(\tau),\stiff}(z)$. Expanding this Dirichlet-to-Neumann map in powers of $\e$, one gets (see \eqref{eq:standard_expansion})
$$
B_\e^{(\tau)}(z)=-\e^{-2} M_0^{(\tau)} - z M_1^{(\tau)}+O(\e^2),
$$
yielding
$$
B_\eff^{(\tau)}:=-\e^{-2} M_0^{(\tau)} - z M_1^{(\tau)},
$$
where $M_1^{(\tau)}$ can be obtained as $M_1^{(\tau)}=\bigl(M_\e^{(\tau),\stiff}\bigr)'_z(0)$, whereas $M_0^{(\tau)}=\e^2\Gamma_1^\stiff \Pi_\stiff$ does not depend on $\e$ due to the choice of $\Gamma_1^\stiff$, see \eqref{weightedK1} and Section \ref{triple_section}. On the other hand, \cite{Ryzhov_later} provides us with the representation $\bigl(M_\e^{(\tau),\stiff}\bigr)'_z(0)=\Pi_\stiff^*\Pi_\stiff$, and hence
$$
B_\eff^{(\tau)}=-\e^{-2} M_0^{(\tau)}  - z\Pi_\stiff^*\Pi_\stiff,
$$
as required for the applicability of Theorem \ref{cor:NRA} after making the suitable choice of $\mathcal B^{(\tau)}:$
\[
\mathcal B^{(\tau)}=-\e^2(\Pi_\stiff^*)^{-1} M_0^{(\tau)}(\Pi_\stiff)^{-1}.
\]

However, as seen from Lemmata \ref{lemma:gen_res_0}, \ref{lemma:gen_res_2}, this strategy is not the most effective from the point of view of the form of the final result, as it requires the asymptotics $\mathcal R^{(\tau)}_\eff$ to depend on $\e$ even when it proves possible to obtain a uniform limit for the generalised resolvents $R_\e^{(\tau)}.$
\end{remark}

Motivated by the just mentioned results, and also by the classical elliptic argument of \cite{BirmanSuslina}, \cite{Friedlander} ({\it cf.} Remark \ref{rem:assumptions}, (3)), one arrives at $R^{(\tau)}_\eff$ as described by Assumptions (i)--(iii) with a non-trivial projection $\P$. In order to better understand this case, consider an intermediate operator family $\mathfrak{A}_\e^{(\tau)}$ defined by the same differential expression as $A_\e^{(\tau)}$ on the domain $\dom \mathfrak{A}_\e^{(\tau)}$:
$$
\dom \mathfrak{A}_\e^{(\tau)}:=\bigl\{u\in \dom {A}_{\max}: \Port \Gamma_0 u =0, \P \Gamma_1 u = 0\bigr\}
$$
({\it cf.} \cite{CherKis}, where this intermediate family was postulated -- in fact, as a matter of an ``educated guesswork''). It is easily seen that
for the corresponding $M$-matrix $\mathfrak{M}^{(\tau)}_\e$ one has $\mathfrak{M}^{(\tau)}_\e(z)=\P M^{(\tau)}_\e(z)\P$, where $M^{(\tau)}_\e$ is as above the $M-$matrix of $A_{\max}$ relative to the triple $(\mathcal H, \Gamma_0, \Gamma_1)$ and $\mathfrak{M}^{(\tau)}_\e$ is the $M$-matrix of the operator $\mathfrak{A}_{\max}$ relative to the triple $(\P\mathcal H, \P{\Gamma}_0, \P{\Gamma}_1)$. Here $\mathfrak{A}_{\max}$ is defined as the restriction of $A_{\max}$ to the domain $\dom \mathfrak{A}_{\max}:=\dom {A}_{\max}\cap\ker(\Port\Gamma_0).$

The analysis of Appendices A, B (see \cite{Physics} for the general case) is then invoked to show that
\begin{equation}\label{eq:two_terms}
B_\eff^{(\tau)}(z)=-\P M_\e^{(\tau), \stiff}(0)\P- z\breve{\Pi}_\stiff^*\breve{\Pi}_\stiff=-\mathfrak{M}^{(\tau),\stiff}_\e(0)-z \bigl(\mathfrak{M}^{(\tau),\stiff}_\e\bigr)'(0),
\end{equation}
where $\mathfrak{M}^{(\tau),\stiff}_\e$ is the $M-$matrix of the stiff component of the media, introduced for the operator $\mathfrak{A}^{(\tau)}_\e$ as in Section \ref{triple_section}.

By an application of Theorem \ref{thm:crucial} one then has
that the operator family $(\mathfrak{A}_\e^{(\tau)}-z)^{-1}$ admits the same asymptotics provided by Theorem \ref{cor:NRA} as the family $({A}_\e^{(\tau)}-z)^{-1}$. This leads to the possibility to treat critical-contrast periodic media as a particular example of the so-called \emph{folded media}, see {\it e.g.} \cite{Leonhardt, Milton_et_al}. This subject is beyond the scope of the present paper though and will be treated elsewhere.

\begin{remark}
  We remark that the decomposition \eqref{eq:two_terms} yields a useful equivalent definition for the operator $\mathcal A$ (see Definition \ref{defn:gen_dilation}). Indeed, under the assumption that (\ref{eq:two_terms}) holds and using the general fact that $M_\e^{(\tau), \stiff}(0)=\Gamma_1^\stiff \Pi_\stiff,$ one computes $\mathcal B^{(\tau)}=-(\Pi_\stiff^*)^{-1}\Gamma_1^\stiff$, thus arriving at the following expression for the action of $\mathcal A$ (cf. \eqref{eq:strauss_op}):
  \begin{equation}\label{eq:strauss_op_mod}
  \mathcal A \binom{u}{u^{(1)}}:=\binom{\breve{A}_{\max}^{\soft} u}{-(\Pi^*)^{-1} \breve{\Gamma}_1^\soft u - (\Pi^*)^{-1} \breve{\Gamma}_1^\stiff u^{(1)}}.
\end{equation}
Here, as per Theorem \ref{cor:NRA}, $\Pi=\breve{\Pi}_\stiff$ and therefore $u^{(1)}=\breve{\Pi}_\stiff \breve{\Gamma}_0^\soft u\in \ker \breve{A}_{\max}^\stiff$, which ascertains the correctness of \eqref{eq:strauss_op_mod}.
\end{remark}

\section{Examples: homogenised operators}

\label{examples_section}

In this section, we apply results of Section 7 to the three examples introduced in Section
\ref{section:examples} to obtain the leading terms in asymptotic expansion of $(A_\e^{(\tau)}-z)^{-1}$.

\subsection*{(0) A medium with both components disconnected}
Let $H_\hom=H_\soft\oplus \mathbb C^1$. For all values $\tau\in[-\pi, \pi),$ consider a self-adjoint operator $\mathcal A_{\rm hom}^{(\tau)}$ on the space $H_\hom,$ defined as follows. Let the domain $\dom \mathcal A_{\hom}^{(\tau)}$ be defined as
\begin{equation*}\label{eq:ex0_domain}
\dom \mathcal A_{\hom}^{(\tau)}=\Bigl\{(u,\beta)^\top\in H_\hom:\ u\in W^{2,2}(0,l^{(2)}), u(0)=\bxit u(l^{(2)})=\beta/\sqrt{l^{(1)}}\Bigr\}.
\end{equation*}

On $\dom \mathcal A_{\hom}^{(\tau)}$ the action of the operator is set by
$$
\mathcal A_{\hom}^{(\tau)}\binom{u}{\beta}=
\left(\begin{array}{c}\biggl(\dfrac{1}{\rm i}\dfrac{d}{dx}+\tau\biggr)^2\\[0.7em]
-\dfrac{1}{\sqrt{l^{(1)}}}
\bigl(\partial^{(\tau)} u\bigr|_0 - \bxit\partial^{(\tau)} u\bigr|_{l^{(2)}}\bigr)
\end{array}\right).
$$
We recall that
\[
\xi^{(\tau)}:=\exp(i l^{(1)} \tau),\qquad
\dtau u:=\biggl(\frac d {dx}+i\tau\biggr) u.
\]

\begin{theorem}\label{thm:ex0}
The resolvent $(A_\e^{(\tau)}-z)^{-1}$ admits the following estimate in the uniform operator norm topology:
$$
(A_\e^{(\tau)}-z)^{-1}-\Psi^* (\mathcal A_{\hom}^{(\tau)}-z)^{-1}\Psi=O(\e^2),
$$
where $\Psi$ is a partial isometry from $H$ to $H_\hom$.
This estimate is uniform in $\tau\in[-\pi,\pi)$ and $z\in K_\sigma$.
\end{theorem}

The \emph{proof} is carried out by invoking Lemma \ref{lemma:gen_res_0} and Theorem \ref{cor:NRA}.

\subsection*{(1) A case of connected stiff component}
Let $H_\hom=H_\soft\oplus \mathbb C^1$. For all values $\tau\in[-\pi, \pi),$ consider a self-adjoint operator $\mathcal A_{\rm hom}^{(\tau)}$ on the space $H_\hom,$ defined as follows. Let the domain $\dom \mathcal A_{\hom}^{(\tau)}$ be defined as
\begin{equation*}
\dom \mathcal A_{\hom}^{(\tau)}=\biggl\{(u,\beta)^\top\in H_\hom:\ u\in W^{2,2}(0,l^{(2)}), u|_{0}=-\frac {\bxit}{|\xit|}u|_{l^{(2)}}=\frac{\beta}{\sqrt{l^{(1)}+l^{(3)}}}\biggr\}.
\end{equation*}

On $\dom \mathcal A_{\hom}^{(\tau)}$ the action of the operator is set by
$$
\mathcal A_{\hom}^{(\tau)}\binom{u}{\beta}=
\left(\begin{array}{c}\biggl(\dfrac{1}{\rm i}\dfrac{d}{dx}+\tau\biggr)^2\\[0.8em]
-\dfrac{1}{\sqrt{l^{(1)}+l^{(3)}}}
\biggl(\partial^{(\tau)} u\bigr|_0 + \dfrac {\bxit}{|\xit|}u\bigr|_{l^{(2)}}\biggr)+\bigl(l^{(1)}+l^{(3)}\bigr)^{-1}\biggl(\dfrac {l^{(1)}}{a_1^2}+\dfrac {l^{(3)}}{a_3^2}\biggr)^{-1}\biggl(\dfrac{\tau}{\e}\biggr)^2 \beta
\end{array}\right).
$$
We recall that
\[
\xit=-\frac{a_1^2}{l^{(1)}}\exp\bigl(i\tau(l^{(1)}+l^{(3)})\bigr)-\frac{a_3^2}{l^{(3)}}\exp(-i\tau l^{(2)}),\qquad
\dtau u:=\biggl(\frac d {dx}+i\tau\biggr) u.
\]

\begin{theorem}\label{thm:ex1}
The resolvent $(A_\e^{(\tau)}-z)^{-1}$ admits the following estimate in the uniform operator norm topology:
$$
(A_\e^{(\tau)}-z)^{-1}-\Psi^* (\mathcal A_{\hom}^{(\tau)}-z)^{-1}\Psi=O(\e^2),
$$
where $\Psi$ is a partial isometry  from $H$ to $H_\hom$.
This estimate is uniform in $\tau\in[-\pi,\pi)$ and $z\in K_\sigma$.
\end{theorem}

The \emph{proof} is carried out by invoking Lemma \ref{lemma:gen_res_1} and Theorems \ref{thm:NRA}, \ref{cor:NRA}, see also Appendix A.

\begin{remark}
It can be shown that the term
$$
\mathfrak G=\bigl(l^{(1)}+l^{(3)}\bigr)^{-1}\biggl(\frac {l^{(1)}}{a_1^2}+\frac {l^{(3)}}{a_3^2}\biggr)^{-1}
$$
in the definition of $\mathcal{A}_\hom^{(\tau)}$ is precisely the spectral germ of \cite{BirmanSuslina}, introduced for the operator $A_\e^\stiff$. This self-adjoint operator is defined on $\mathbb{G}_\stiff$ by the same differential expression as $A_{\max}^\stiff$ on the domain
$$
\dom{A_\e^\stiff}=\{u\in \dom A_{\max}^\stiff : \Gamma_1^\stiff u =0\}
$$
and describes the original media with the soft component dropped. The approach developed in this paper therefore leads to a natural generalisation of the operator-theoretical approach of Birman and Suslina, which is inapplicable in non-strongly elliptic setting.

\end{remark}

\subsection*{(2) A case of connected soft component}
Let $H_\hom=H_\soft\oplus \mathbb C^1$. For all values $\tau\in[-\pi, \pi),$ consider a self-adjoint operator $\mathcal A_{\hom}^{(\tau)}$ on the space $H_\hom,$ defined as follows. Let the domain $\dom \mathcal A_{\hom}^{(\tau)}$ be defined as
\begin{multline*}
\dom \mathcal A_{\hom}^{(\tau)}=\biggl\{(u_1,u_2,\beta)^\top\in L^2[0,l^{(1)}]\oplus L^2(0,l^{(2)})\oplus \mathbb C^1:\\
u_j\in W^{2,2}(0,l^{(j)}), j=1,2;\quad u_2|_0=\bxittwo  u_2|_{l^{(2)}}=\bxitone   u_1|_0=u_1|_{l^{(1)}}=\frac{\beta}{\sqrt{l^{(3)}}}\biggr\}.
\end{multline*}

On $\dom \mathcal A_{\hom}^{(\tau)}$ the action of the operator is set by
$$
\mathcal A_{\hom}^{(\tau)}\begin{pmatrix}{u_1}\\ u_2\\ {\beta}\end{pmatrix}=
\left(\begin{array}{c}a_1^2\biggl(\dfrac{1}{\rm i}\dfrac{d}{dx}+\tau\biggr)^2\\[0.8em]
a_2^2\biggl(\dfrac{1}{\rm i}\dfrac{d}{dx}+\tau\biggr)^2\\[1.1em]
-\dfrac{1}{\sqrt{l^{(3)}}}
\Bigl(a_2^2 \dtau u_2|_0- a_2^2 \bxittwo  \dtau u_2|_{l^{(2)}}+ a_1^2 \bxitone   \dtau u_1|_0-a_1^2 \dtau u_1|_{l^{(1)}}\Bigr)
\end{array}\right).
$$
We recall that
\[
\xit_2=\exp(-i l^{(2)}\tau),\quad \xit_1=\exp\bigl(-i(l^{(2)}+l^{(3)})\tau\bigr),\qquad
\dtau u_j:=\biggl(\frac d {dx}+i\tau\biggr) u_j,\quad j=1,2.
\]

\begin{theorem}\label{thm:ex2}
The resolvent $(A_\e^{(\tau)}-z)^{-1}$ admits the following estimate in the uniform operator norm topology:
$$
(A_\e^{(\tau)}-z)^{-1}-\Psi^* (\mathcal A_{\hom}^{(\tau)}-z)^{-1}\Psi=O(\e^2),
$$
where $\Psi$ is a partial isometry from $H$ to $H_\hom$.
This estimate is uniform in $\tau\in[-\pi,\pi)$ and $z\in K_\sigma$.
\end{theorem}

The \emph{proof} is carried out by invoking Lemma \ref{lemma:gen_res_0} and Theorem \ref{cor:NRA}.

\section{Schur-Frobenius complement of the sandwiched resolvent on the soft component}

In this section we continue the study of the three examples, for which in Section \ref{examples_section} we constructed the resolvent asymptotics, in view to obtain equivalent time-dispersive formulations on the real line. In order to achieve this, we first introduce the orthogonal  projection  ${\mathfrak P}$  of $H_{\rm hom}$ onto $H_{\rm hom}\ominus H_{\rm soft},$ the latter space being ${\mathbb C}^1$ in all three cases. Following this, we determine the corresponding Schur-Frobenius complement ${\mathfrak P}(\mathcal A_{\hom}^{(\tau)}-z)^{-1}{\mathfrak P},$ see \cite[p.\,416]{Fuerer}.

\subsection{Examples (0) and (1)}
\label{Sch_F_sub}
Due to the fact that the soft component in each of these examples consists of only one edge, we shall consider Examples (0) and (1) of Section \ref{examples_section} simultaneously. To this end, we set
\begin{equation}
\Gamma_\tau\left(\begin{matrix}u\\[0.2em] \beta\end{matrix}\right)=-\dtau u\big\vert_0+\overline{w_\tau}\dtau u\big\vert_{l^{(2)}}+\biggl(\frac{\sigma\tau}{\varepsilon}\biggr)^2\frac{\beta}{\rho},
\label{Gamma}
\end{equation}
where $w_\tau,$ $\sigma$ and $\rho$ depend on the particular case, {\it cf.} Theorems \ref{thm:ex0}, \ref{thm:ex1}.
The problem of calculating ${\mathfrak P}(\mathcal A_{\hom}^{(\tau)}-z)^{-1}{\mathfrak P}$
consists in determining $\beta$ that solves
\begin{align}
-&\biggl(\frac{d}{dx}+{\rm i}\tau\biggr)^2u-zu=0,\label{diff_part}\\
&\left(\begin{matrix}u\\[0.2em] \beta\end{matrix}\right)\in\dom \mathcal A_{\hom}^{(\tau)},\quad \frac{1}{\rho}\,\Gamma_\tau\left(\begin{matrix}u\\[0.2em] \beta\end{matrix}\right)-z\beta=\delta\label{boundary_part}.
\end{align}
In order to exclude $u$ from (\ref{diff_part})--(\ref{boundary_part}), we represent it as a sum of two functions: one of them is a solution to the related inhomogeneous Dirichlet problem, while the other takes care of the boundary condition. More precisely, consider the solution $v$ to the problem
\begin{equation*}
-\biggl(\frac{d}{dx}+{\rm i}\tau\biggr)^2v=0,\qquad\qquad
v(0)=1,\ \ \ \ \ v(l^{(2)})=w_\tau,
\end{equation*}
{\it i.e.}
\begin{equation}
v(x)=\Bigl\{1+(l^{(2)})^{-1}\Bigl(w_\tau\exp({\rm i}\tau l^{(2)})-1\Bigr)x\Bigr\}\exp(-{\rm i}\tau x),\quad\quad x\in(0, l^{(2)}).
\label{function_v}
\end{equation}
The function
\[
\widetilde{u}:=u-\frac{\beta}{\rho}v
\]
satisfies
\begin{equation*}
-\biggl(\frac{d}{dx}+{\rm i}\tau\biggr)^2\widetilde{u}-z\widetilde{u}=\frac{z\beta}{\rho}v,\qquad
 \widetilde{u}(0)=\widetilde{u}(l^{(2)})=0.
\end{equation*}
Equivalently, one has
\begin{equation*}
\widetilde{u}=\frac{z\beta}{\rho}(A_{\rm D}-zI)^{-1}v,
\end{equation*}
where $A_{\rm D}$ is the Dirichlet operator in $L^2(0, l^{(2)})$ associated with the differential expression
\[
-\biggl(\frac{d}{dx}+{\rm i}\tau\biggr)^2.
\]
We now write the ``boundary''  part of the system (\ref{diff_part})--(\ref{boundary_part}) as
\begin{equation}
K(\tau, z)\beta-z\beta=\delta,
\label{K_eq}
\end{equation}
where
\begin{equation}
K(\tau, z):=\dfrac{1}{\rho^2}\left\{z\Gamma_\tau\left(\begin{matrix}
(A_{\rm D}-zI)^{-1}v\\[0.3em] 0\end{matrix}\right)+
\Gamma_\tau\left(\begin{matrix}v\\[0.2em] \rho\end{matrix}\right)\right\}.
\label{K_expr}
\end{equation}
Thus ${\mathfrak P}(\mathcal A_{\hom}^{(\tau)}-z)^{-1}{\mathfrak P}$ is the operator of multiplication in ${\mathbb C }^1$ by $(K(\tau, z)-z)^{-1}.$



The formula (\ref{K_expr}) shown, in particular, that the dispersion function $K$ is singular only at eigenvalues of the Dirichlet Laplacian on the soft component. It allows to compute $K$ in terms of the spectral decomposition of $A_{\rm D},$ {\it cf.} \cite{Zhikov2000}. In order to see this, we represent the action of the resolvent $(A_{\rm D}-zI)^{-1}$ as a series in terms of the normalised eigenfunctions
\begin{equation}
\varphi_j(x)=\sqrt{\frac{2}{l^{(2)}}}\exp(-{\rm i}\tau x)\sin\frac{\pi jx}{l^{(2)}},\qquad x\in(0,l^{(2)}),\qquad\qquad j=1,2,3,\dots,
\label{function_phi}
\end{equation}
of the operator $A_{\rm D},$ which yields
\begin{equation}
K(\tau, z):=\dfrac{1}{\rho^2}\left\{z\sum_{j=1}^\infty\dfrac{\langle v, \varphi_j\rangle}{\mu_j-z}\Gamma_\tau\left(\begin{matrix}
\varphi_j\\[0.3em] 0\end{matrix}\right)+
\Gamma_\tau\left(\begin{matrix}v\\[0.2em] \rho\end{matrix}\right)\right\}.
\label{K_general1}
\end{equation}
where $\mu_j=(\pi j/l^{(2)})^2,$ $j=1,2,3,\dots,$ are the corresponding eigenvalues and $v$ is defined in (\ref{function_v}). In each case we consider the problem (\ref{diff_part})--(\ref{boundary_part}), where operator $\Gamma_\tau$ depends on the specific example at hand.

\subsection{Example (2)}
Here we define
\begin{align*}
&\Gamma_\tau\left(\begin{matrix}u\\[0.2em] \beta\end{matrix}\right)=-a_1^2\Bigl(-\dtau u_1\big\vert_{l^{(1)}}+\overline{\xi_1^{(\tau)}}\dtau u_1\big\vert_0\Bigr)+a_2^2\Bigl(-\dtau u_2\big\vert_0+\overline{\xi_2^{(\tau)}}\dtau u_2\big\vert_{l^{(2)}}\Bigr),\\[0.4em]
&\xi_1^{(\tau)}:=\exp\bigl(-{\rm i}(l^{(2)}+l^{(3)})\tau\bigr),\quad \xi_2^{(\tau)}:=\exp(-{\rm i}l^{(2)}\tau),
\end{align*}
where $u_1$ and $u_2$  are the restrictions of the function $u$ to the edges $(0, l^{(1)})$ and $(0, l^{(2)}),$ respectively, and the resolvent problem for $A_{\rm hom}^{(\tau)}$ is given by ({\it cf.} (\ref{diff_part})--(\ref{boundary_part}))
\begin{align}
-&a_1^2\biggl(\frac{d}{dx}+{\rm i}\tau\biggr)^2u_1-zu_1=0,\label{new_1}\\[0.3em]
-&a_2^2\biggl(\frac{d}{dx}+{\rm i}\tau\biggr)^2u_2-zu_2=0,\label{new_2}\\[0.3em]
&\left(\begin{matrix}u\\[0.2em] \beta\end{matrix}\right)\in\dom \mathcal A_{\hom}^{(\tau)},\quad\frac{1}{\rho}\,\Gamma_\tau\left(\begin{matrix}u\\[0.2em] \beta\end{matrix}\right)-z\beta=\delta,\label{new_3}
\end{align}
where $\rho=\sqrt{l^{(3)}}.$ Following the strategy of Section \ref{Sch_F_sub}, we consider the functions $v_j,$ $j=1,2$ that satisfy appropriate Dirichlet problems:
\begin{equation*}
-\biggl(\frac{d}{dx}+{\rm i}\tau\biggr)^2v_1=0,\qquad\qquad
v_1(0)=\xi_1^{(\tau)},\ \ \ \ \ v_1(l^{(1)})=1,
\end{equation*}

\begin{equation*}
-\biggl(\frac{d}{dx}+{\rm i}\tau\biggr)^2v_2=0,\qquad\qquad
v_1(0)=1,\ \ \ \ \ v_2(l^{(2)})=\xi_2^{(\tau)},
\end{equation*}
{\it i.e.}
\begin{equation*}
v_1(x)=
\xi_1^{(\tau)}\Bigl\{1+(l^{(1)})^{-1}\bigl(\exp(i\tau)-1\bigr)x\Bigr\}\exp(-{\rm i}\tau x),\quad x\in(0, l^{(1)}),\qquad v_2(x)=\exp(-{\rm i}\tau x),
\quad x\in(0, l^{(2)}).\\
\end{equation*}

As in Section \ref{Sch_F_sub}, we infer that
\begin{equation*}
\widetilde{u}=\frac{z\beta}{\rho}\sum_{n=1}^2\chi^{(n)}(A_{\rm D}^{(n)}-zI)^{-1}v_n,
\end{equation*}
where $A_{\rm D}^{(n)},$ $n=1,2,$ are the Dirichlet operators in $L^2(0, l^{(n)}),$ $n=1,2$ associated with the differential expression
\[
-a_n^2\biggl(\frac{d}{dx}+{\rm i}\tau\biggr)^2,\quad n=1,2,
\]
and $\chi^{(n)},$ $n=1,2,$ are the characteristic functions of the edges $(0, l^{(n)}),$ $n=1,2.$
Therefore we can write the ``boundary''  part of the resolvent equation (\ref{new_1})--(\ref{new_3}) as
\begin{equation*}
K(\tau, z)\beta-z\beta=\delta,
\end{equation*}
where
\begin{align}
K(\tau, z)&:=\dfrac{1}{\rho^2}\sum\limits_{n=1}^2\left\{z\Gamma_\tau\left(\begin{matrix}
\chi^{(n)}(A_{\rm D}^{(j)}-zI)^{-1}v_n\\[0.3em] 0\end{matrix}\right)+
\Gamma_\tau\left(\begin{matrix}\chi^{(n)}v_n\\[0.2em] \rho\end{matrix}\right)\right\}\nonumber\\[0.4em]
&=\dfrac{1}{\rho^2}\sum\limits_{n=1}^2\left\{z\sum_{j=1}^\infty\dfrac{\langle v_n, \varphi_j^{(n)}\rangle}{\mu_j^{(n)}-z}\Gamma_\tau\left(\begin{matrix}
\chi^{(n)}\varphi_j^{(n)}\\[0.3em] 0\end{matrix}\right)+
\Gamma_\tau\left(\begin{matrix}\chi^{(n)}v_n\\[0.2em] \rho\end{matrix}\right)\right\}.\label{K_ex2}
\end{align}
Here ({\it cf.} (\ref{function_phi}))
\begin{equation*}
\varphi_j^{(n)}(x)=\sqrt{\frac{2}{l^{(n)}}}\exp(-{\rm i}\tau x)\sin\frac{\pi jx}{l^{(n)}},\qquad x\in(0,l^{(n)}),\qquad \mu_j^{(n)}=\biggl(\frac{\pi j}{l^{(n)}}\biggr),\qquad j=1,2,3,\dots,\qquad n=1,2.
\end{equation*}

\subsection{Expressions for the dispersion functions}

\begin{lemma} In each of the three examples introduced in Section \ref{section:examples}, the action of the Schur-Frobenius complement ${\mathfrak P}(\mathcal A_{\hom}^{(\tau)}-z)^{-1}{\mathfrak P}$
is represented as the operator of multiplication by $(K(\tau, z)-z)^{-1},$ where the dispersion function $K$ is given by the following formulae:
\begin{align}
&{\rm Example\ (0):}\qquad K(\tau, z)=\frac{2\sqrt{z}\bigl(\cos(l^{(2)}\sqrt{z})-\cos\tau\bigr)}{l^{(1)}\sin(l^{(2)}\sqrt{z})},\label{K_example0th}\\[0.6em]
&{\rm Example\ (1):}\qquad K(\tau, z)=\frac{1}{l^{(1)}+l^{(3)}}\Biggl\{2\sqrt{z}\biggl(
\cot(l^{(2)}\sqrt{z})
-\frac{\Re\theta(\tau)}{\sin(l^{(2)}\sqrt{z})}\biggr)+\biggl(\frac{\sigma\tau}{\varepsilon}\biggr)^2\Biggr\},\label{K_example1th}\\[0.7em]
&{\rm Example\ (2):}\qquad K(\tau, z)=\frac{2\sqrt{z}}{l^{(3)}}\Biggl\{a_1^2\frac{\cos(l^{(1)}\sqrt{z})-\cos\tau}{\sin(l^{(1)}\sqrt{z})}-a_2^2\tan\biggl(\frac{l^{(2)}\sqrt{z}}{2}\biggr)\Biggr\},\label{K_example2th}
 \end{align}
 where $\tau\in[-\pi,\pi),$ and
 \[
 \theta(\tau):=\biggl|\dfrac{a_1^2}{l^{(1)}}{\rm e}^{-i\tau}+\dfrac{a_3^2}{l^{(3)}}\biggr|^{-1}\left(\dfrac{a_1^2}{l^{(1)}}{\rm e}^{-i\tau}+\dfrac{a_3^2}{l^{(3)}}\right),\qquad \sigma^2:=\biggl(\dfrac{l^{(1)}}{a_1^2}+\dfrac{l^{(3)}}{a_3^2}\biggr)^{-1}.
 \]
\end{lemma}

The derivation of the above expressions is given in Appendix C.

\section{Effective macroscopic problems on the real line}


Here we shall interpret the Schur-Frobenius complements constructed in the previous section as a result of applying the Gelfand transform (see Section \ref{Gelfand_section}) to a one-dimensional homogeneous medium.
To this end, we unitarily immerse the $L^2$ space of functions of $t$ into the $L^2$ space of functions of $t$ and $x,$ corresponding to the stiff component of the original medium, by the formula
\[
\beta(t)\mapsto\beta(t)\frac{1}{\sqrt{\varepsilon L}}{\mathbbm 1}(x),
\]
where $L$ is the length of the stiff component, {\it i.e.} $L=l^{(1)}$ in Examples (0), $L=l^{(1)}+l^{(3)}$ in Example (1), $L=l^{(3)}$ in Example (2), and write the effective problem (\ref{K_eq})
in the form
\begin{equation}
K(\varepsilon t, z)\beta(t)\frac{1}{\sqrt{\varepsilon L}}{\mathbbm 1}(x)-z\beta(t)\frac{1}{\sqrt{\varepsilon L}}{\mathbbm 1}(x)=\delta(t)\frac{1}{\sqrt{\varepsilon L}}{\mathbbm 1}(x),\qquad t\in[-\pi/\varepsilon,\pi/\varepsilon), \quad  x\in(0,\varepsilon L),
\label{proj_eq}
\end{equation}
The solution operator for (\ref{proj_eq}), namely
\[
\delta(t)\frac{1}{\sqrt{\varepsilon L}}{\mathbbm 1}(x)\mapsto \beta(t)\frac{1}{\sqrt{\varepsilon L}}{\mathbbm 1}(x)\ \ {\rm such\ that\ (\ref{proj_eq})\  holds,}
\]
is the composition of a projection operator in $L^2\bigl((-\pi/\varepsilon,\pi/\varepsilon)\times(0,\varepsilon L)\bigr)$ onto constants in $x$ and multiplication by the function $\bigl(K(\varepsilon t, z)-z\bigr)^{-1},$ as follows:
\begin{equation}
\bigl(K(\varepsilon t, z)-z\bigr)^{-1}\biggl\langle\cdot, \frac{1}{\sqrt{\varepsilon L}}{\mathbbm 1}(x)\biggr\rangle \frac{1}{\sqrt{\varepsilon L}}{\mathbbm 1}(x),
\label{sol_proj}
\end{equation}
for all $z$ such that $K(\varepsilon t, z)-z$ is invertible, in particular, for $z\in K_\sigma.$
The sought representation on ${\mathbb R}$ is the Schur-Frobenius complement obtained by sandwiching the operator (\ref{sol_proj}) with the Gelfand transform
\[
{\mathcal G}F(x,t)=\sqrt{\frac{\varepsilon }{2\pi}}\sum_{n\in{\mathbb Z}}F(x+n\varepsilon)\exp\bigl(-{\rm i}(x+n\varepsilon)t\bigr),\qquad x\in(0,\varepsilon),\quad t\in[-\pi/\varepsilon, \pi/\varepsilon),\qquad F\in L^2({\mathbb R}),
\]
and its inverse
\[
{\mathcal G}^*u(x)=\sqrt{\frac{\varepsilon}{2\pi}}\int_{-\pi/\varepsilon}^{\pi/\varepsilon}u(x,t)\exp({\rm i}xt)dt,\quad x\in{\mathbb R}, \qquad u\in L^2\bigl((-\pi/\varepsilon,\pi/\varepsilon)\times(0,\varepsilon)\bigr),
\]
so that the overall operator is given by
\[
{\mathcal G}^*\Biggl\{\bigl(K(\varepsilon t, z)-z\bigr)^{-1}\biggl\langle{\mathcal G}\,\cdot, \frac{1}{\sqrt{\varepsilon L}}{\mathbbm 1}(x)\biggr\rangle \frac{1}{\sqrt{\varepsilon L}}{\mathbbm 1}(x)\Biggr\}.
\]
In constructing the above operator we assume that the operator given by (\ref{sol_proj}) has been extended by zero to the soft component of the medium.

This results in the mapping
\begin{align}
F\mapsto\Psi_K^\varepsilon F:=&\sqrt{\frac{\varepsilon}{2\pi}}\int_{-\pi/\varepsilon}^{\pi/\varepsilon}\bigl(K(\varepsilon t, z)-z\bigr)^{-1}\biggl\langle{\mathcal G}F, \frac{1}{\sqrt{\varepsilon L}}{\mathbbm 1}\biggr\rangle(t)\frac{1}{\sqrt{\varepsilon L}}{\mathbbm 1}(x)\exp({\rm i}tx)dt\nonumber\\[0.6em]
=&\frac{1}{L\sqrt{2\pi}}\int_{-\pi/\varepsilon}^{\pi/\varepsilon}\bigl(K(\varepsilon t, z)-z\bigr)^{-1}\widehat{F}(t)\exp({\rm i}tx)dt\label{Uform},
\end{align}
whose inverse
yields the required effective problem on ${\mathbb R}.$ Here
\[
\widehat{F}(t)=\frac{1}{\sqrt{2\pi}}\int_{-\infty}^\infty F(x)\exp(-{\rm i}x t)dx,\qquad t\in{\mathbb R},
\]
is the Fourier transform of the function $F.$

By applying Theorems \ref{thm:ex0}, \ref{thm:ex1}, \ref{thm:ex2}, we arrive at the following statement.
\begin{theorem}
The direct integral of Schur-Frobenius complements
\begin{equation}
\oplus\int_{\pi/\varepsilon}^{\pi/\varepsilon}P_{\rm stiff}(A_\varepsilon^{(t)}-z)^{-1}P_{\rm stiff}dt
\label{dir_int}
\end{equation}
is $O(\varepsilon^r)$-close, in the uniform operator-norm topology, to an operator unitary equivalent to the pseudo-differential operator defined by (\ref{Uform}).
Here $r=1$ in Example (1) and $r=2$ in Examples (0) and (2).

The direct integral (\ref{dir_int}) is the composition of the original resolvent family $(A_\varepsilon-z)^{-1}$ applied  to functions supported by the stiff component of ${\mathbb G}_{\rm per}$ and the orthogonal projection onto the same stiff component

\end{theorem}
On the basis of the above theorem, we will now explicitly characterise the effective time-dispersive medium in each of the examples.

\subsection{Example (0)}
Notice that, by (\ref{Uform}), for $U:=\Psi_{K}^\varepsilon F$ one has
\begin{equation}
\frac{1}{2}\Bigl(U(x+\varepsilon)+U(x-\varepsilon)\Bigr)=\frac{1}{l^{(1)}\sqrt{2\pi}}\int_{-\pi/\varepsilon}^{\pi/\varepsilon}\frac{\cos(\varepsilon t)}{K(\varepsilon t, z)-z}\widehat{F}(t)\exp({\rm i}t x)dt,
\label{Ex0_U}
\end{equation}
and since in Example (0) we have, see (\ref{K_example0th}),
\[
K(\tau, z)=\frac{2\sqrt{z}}{l^{(1)}}\cot(l^{(2)}\sqrt{z})-\frac{2\sqrt{z}}{l^{(1)}\sin(l^{(2)}\sqrt{z})}\cos\tau,
\]
we obtain
\begin{multline*}
\frac{2\sqrt{z}}{l^{(1)}}\cot(l^{(2)}\sqrt{z})U(x)-\frac{1}{2}\Bigl(U(x+\varepsilon)+U(x-\varepsilon)\Bigr)\frac{2\sqrt{z}}{l^{(1)}\sin(l^{(2)}\sqrt{z})}-zU(x)\\[0.5em]
=\frac{1}{l^{(1)}\sqrt{2\pi}}\int_{-\pi/\varepsilon}^{\pi/\varepsilon}\widehat{F}(t)\exp({\rm i}t x)dt
\sim\frac{1}{l^{(1)}}F(x),\quad \varepsilon\to0.
\end{multline*}
It follows that the asymptotic form of the equation on the function $U$ is
\begin{equation*}
-\frac{\sqrt{z}}{\sin(l^{(2)}\sqrt{z})}\Delta_\varepsilon U-\biggl\{l^{(1)}z+2\sqrt{z}\tan\biggl(\frac{l^{(2)}\sqrt{z}}{2}\biggr)\biggr\}U=F,
\end{equation*}
where
\begin{equation}
\Delta_\varepsilon U:=U(\cdot+\varepsilon)+U(\cdot-\varepsilon)-2U,\qquad \varepsilon>0.
\label{LaplaceU}
\end{equation}
is the difference Laplace operator. Clearly, by a unitary rescaling of the independent variable we obtain an $\varepsilon$-independent limit problem.



\subsection{Example (1)}
\begin{lemma}
\label{10.2}
One has the estimate
\[
\Vert\Psi_K^\varepsilon-\Psi_K^0\Vert_{L^2({\mathbb R})\to L^2({\mathbb R})}=O(\varepsilon^2),\quad\varepsilon\to0,
\]
where ({\it cf.} \ref{Uform})
\[
\Psi_K^0:=\frac{1}{(l^{(1)}+l^{(3)})\sqrt{2\pi}}\int_{-\infty}^{\infty}\bigl(K(\varepsilon t, z)-z\bigr)^{-1}\widehat{F}(t)\exp({\rm i}tx)dt,
\]
with $K(\tau, z)$ defined by the formula (\ref{K_example1th}) for all values of $\tau.$
\end{lemma}
\begin{proof}
The proof is standard, see {\it e.g.} \cite{Hoermander}.
\end{proof}

 It follows from $\Re\theta(\tau)=\Re\theta(0)+O(\tau^2)=1+O(\tau^2)$ that
\[
K(\varepsilon t, z)=\widetilde{K}(t, z)+O(\varepsilon^2 t^2),\qquad  \widetilde{K}(t, z):=\frac{1}{l^{(1)}+l^{(3)}}\biggl\{
(\sigma t)^2-2\sqrt{z}\tan\biggl(\frac{l^{(2)}\sqrt{z}}{2}\biggr)
\biggr\},\qquad t\in[-\pi/\varepsilon, \pi/\varepsilon),
\]
from which we infer
\[
\bigl(K(\varepsilon t, z)-z\bigr)^{-1}=\bigl(\widetilde{K}(t, z)-z\bigr)^{-1}+O(\varepsilon^{2}),
\]
and hence we obtain the following statement.
\begin{lemma} The following estimate holds:
\[
\Vert\Psi_K^0-\Psi_{\widetilde{K}}^0\Vert_{L^2({\mathbb R})\to L^2({\mathbb R})}=O(\varepsilon^{2}),\quad \varepsilon\to0.
\]
\end{lemma}
Therefore, for $U:=\Psi_{\widetilde{K}}^0F,$
we obtain
\begin{equation}
-\sigma^2\,U''(x)
-\biggl\{\bigl(l^{(1)}+l^{(3)}\bigr)z+2\sqrt{z}\tan\biggl(\frac{l^{(2)}\sqrt{z}}{2}\biggr)\biggr\}U(x)=F(x),\qquad x\in{\mathbb R}.
\label{lim1}
\end{equation}


We summarise the results of this section in the following theorem.
\begin{theorem}
\label{10.4}
In the case of Example (1) the effective time-dispersive formulation on the real line is provided by
the formula (\ref{lim1}) with an error bound of order $O(\varepsilon^{2})$.
\end{theorem}

\begin{remark}
1. Note that in  Examples (0) and (2) the effective time-dispersive formulation is given by a difference equation, whereas in Example (1) --- by a differential one. The reason for this is the global connectedness of the stiff component ({\it cf.} \cite{Zhikov2000}) in Example (1), which leads, see (\ref{K_example2th}) to a nonuniform in $\varepsilon$ dependence of the kernel $K(\tau, z)$ on $\tau.$

2. We point out that even in the case of globally connected stiff component one could end up in a situation where Theorem \ref{10.4} yields a rate of convergence lower than $O(\e^2)$, see \cite{Physics} for further details. In this case, a ``corrected'' result can still be obtained, with the rate of convergence $O(\e^2)$, by replacing \eqref{lim1} by a differential equation with non-local perturbation.
\end{remark}

\subsection{Example (2)} By analogy with Example (0), we use the formula (\ref{Ex0_U}) and note that, in view of (\ref{K_example2th}), we have
\[
K(\tau, z)=\frac{2\sqrt{z}}{l^{(3)}}\Biggl\{a_1^2\tan(l^{(1)}\sqrt{z})-a_2^2\tan\biggl(\frac{l^{(2)}\sqrt{z}}{2}\biggr)\Biggr\}-\frac{2a_1^2\sqrt{z}}{l^{(3)}\sin(l^{(1)}\sqrt{z})}\cos\tau.
\]
It follows that the time-dispersive effective formulation on $U:=\Psi_{K}^\varepsilon F$ has the form
\[
-\frac{a_1^2\sqrt{z}}{\sin(l^{(1)}\sqrt{z})}\Delta_\varepsilon U-\Biggl\{l^{(3)}z+2\sqrt{z}\Biggl(a_1^2\tan\biggl(\frac{l^{(1)}\sqrt{z}}{2}\biggr)+a_2^2\tan\biggl(\frac{l^{(2)}\sqrt{z}}{2}\biggr)\Biggr)\Biggr\}U=F,
\]
where the difference Laplacian $\Delta_\varepsilon$ is defined by (\ref{LaplaceU}).

\begin{remark}
1. A version of Theorem \ref{10.4} seems to be impossible to obtain in Examples (0) and (2), due to the fact that  there is no suitable counterpart of Lemma \ref{10.2} available. In these cases the convergence to the described effective medium holds, albeit without explicit control of the order of the remainder term.

2. The effective formulation (\ref{lim1}) is precisely the one yielded by the approach of \cite{Zhikov2000}. We note that in the cited paper the stated result involves only two-scale convergence with no estimate on the error term. In contrast, our approach provides norm-resolvent convergence, with an order-explicit error estimate. 
\end{remark}

\section*{Appendix A. The proof of Lemma \ref{lemma:gen_res_0}}

We start by expanding the matrix $B^{(\tau)}(z)$ into power series with respect to small parameter $\e$. By Lemma \ref{lemma:M_0}, this matrix admits the form
$$
B^{(\tau)}(z)\equiv -M^{(\tau)}_\stiff(z)= \dfrac 1{\e}
\begin{pmatrix}
a_1 k \cot \dfrac{k l^{(1)} \e}{a_1}& -\dfrac{a_1 k}{\sin \dfrac{k \e l^{(1)}}{a_1} }e^{-i l^{(1)} \tau}\\[0.8em]
-\dfrac{a_1 k}{\sin \dfrac{k \e l^{(1)}}{a_1}} e^{i l^{(1)} \tau}& a_1 k \cot \dfrac{k l^{(1)} \e}{a_1}
\end{pmatrix},
$$
where $k=\sqrt{z}$ with the usual choice of branch $\Im k\geq 0$. Then,
$$
B^{(\tau)}(z)=B^{(\tau)}(0)+O(1)=\dfrac{a_1^2}{l^{(1)}\e^2} \begin{pmatrix} 1&-\bxit \\[0.7em]
                                                     -\xit & 1
                                                     \end{pmatrix}+
O(1)=:B_0+O(1)
$$
We remark that, since $B^{(\tau)}$ is proportional to the $M$-matrix of the operator $A_{\max}^\stiff$, it is meromorphic in $\mathbb C$ as a function of $z$, all its poles located at a distance of order $1/\e^2$ from the origin. It is therefore entire in $K_\sigma$. Moreover, it is Hermitian for real values of $z$ away from its poles.


The matrix $\e^2 B_0$ is a Hermitian matrix possessing two distinct eigenvalues, one of which is equal to zero at all values of $\tau$. The associated eigenvector $\psi^{(\tau)}$ is given by $\psi^{(\tau)}=(1/\sqrt{2})(1, \xit)^\top$. Introduce the orthogonal projection $\P$ in $\mathcal H:$
$$
\P=\bigl\langle \cdot,\psit\bigr\rangle \psit,
$$
and its orthogonal complement $\Port$ defined by the vector $\psit_\bot=(1/\sqrt{2})(1,-\xit)^\top$. We then pass over to an auxiliary triple $(\mathcal H, \widehat\Gamma_0,\widehat\Gamma_1)$ (where, in order to simplify the notation, we  drop the superscript ``soft'', in hope this does not lead to any misunderstanding)
that diagonalises the matrix $B_0$ and thus asymptotically diagonalises the matrix $B^{(\tau)}(z)$ as $\e\to 0$. To this end, consider the matrix $X$ defined by
$$
X=\frac 1{\sqrt{2}}\begin{pmatrix}
1&1\\[0.5em]
\xit&-\xit
\end{pmatrix}
$$
and define $\widehat \Gamma_0:=X^* \Gamma_0^\soft$, $\widehat \Gamma_1:=X^* \Gamma_1^\soft$. Since $X$ is unitary, $(\mathcal H, \widehat \Gamma_0, \widehat\Gamma_1)$ is also a triple for $A_{\max}^{\soft}$:
\begin{equation*}
\langle \widehat\Gamma_1 u, \widehat\Gamma_0 v \rangle-
\langle \widehat\Gamma_0 u, \widehat\Gamma_1 v \rangle=
\langle X^*\Gamma_1^\soft u, X^*\Gamma_0^\soft v \rangle-
\langle X^*\Gamma_0^\soft u, X^*\Gamma_1^\soft v \rangle=
\langle \Gamma_1^\soft u, \Gamma_0^\soft v \rangle-
\langle \Gamma_0^\soft u, \Gamma_1^\soft v \rangle
\end{equation*}
for any $u,v\in \dom A_{\max}^{\soft}$.

The $M$-matrix and the matrix $B^{(\tau)}$ thus converted to
$$
\widehat M^{(\tau)}_\soft(z)=X^* M^{(\tau)}_\soft(z)X, \quad \widehat B^{(\tau)}(z)=X^*  B^{(\tau)}(z)X,
$$
where
$$
\widehat B_0:= \widehat B^{(\tau)}(0)=\diag\left\{0,\dfrac{2a_1^2}{\e^2 l^{(1)}}\right\}.
$$
Thus, the behaviour of the matrix $\widehat B^{(\tau)}$ as $\e\to 0$ is drastically different in the two orthogonal components of $\mathcal H=\P\mathcal H\oplus \Port \mathcal H$: on $\P\mathcal H$
it is determined by the second, non-singular in $\e$, term of its asymptotic expansion, whereas on $\Port\mathcal H$ it is determined by the $O(1/\e^2)$ contribution in $\widehat B_0$. In order that Theorem \ref{thm:crucial} be applicable, one needs to balance these together. With this goal in mind, we consider a further change in the boundary triple. 

Namely, we pass over to the triple $(\mathcal H, \widetilde \Gamma_0, \widetilde \Gamma_1)$, where the boundary operators $\widetilde \Gamma_0$ and $\widetilde \Gamma_1$ are defined as follows:
\begin{equation}\label{eq:projector-triple}
\widetilde \Gamma_0:=\P \widehat \Gamma_0+\Port \widehat \Gamma_1;\quad \widetilde \Gamma_1:=\P \widehat \Gamma_1-\Port \widehat \Gamma_0.
\end{equation}
One checks that this is indeed a boundary triple:
\begin{multline*}
\bigl\langle \widetilde \Gamma_1 u, \widetilde \Gamma_0 v\bigr\rangle-\bigl\langle \widetilde \Gamma_0 u, \widetilde \Gamma_1 v\bigr\rangle\\[0.4em]
=\bigl\langle \P \widehat \Gamma_1 u, \P\widehat \Gamma_0 v\bigr\rangle- \bigl\langle \Port\widehat \Gamma_0 u, \Port \widehat \Gamma_1 v \bigr\rangle-
\bigl\langle \P\widehat \Gamma_0 u, \P\widehat \Gamma_1 v\bigr\rangle+\bigl\langle \Port\widehat \Gamma_1 u, \Port\widehat \Gamma_0 v\bigr\rangle\\[0.4em]
=\bigl\langle \P\widehat \Gamma_1 u+\Port \widehat \Gamma_1 u, \widehat \Gamma_0 v\bigr\rangle- \bigl\langle \Port \widehat \Gamma_0 u +\P \widehat \Gamma_0 u, \widehat \Gamma_1 v\bigr\rangle=
\bigl\langle \widehat \Gamma_1 u, \widehat \Gamma_0 v\bigr\rangle-\bigl\langle \widehat \Gamma_0 u, \widehat \Gamma_1 v\bigr\rangle
\end{multline*}
for all $u,v\in \dom A_{\max}^{\soft}$.

The calculation of $\widetilde B^{(\tau)}$ pertaining to the boundary triple $(\mathcal H, \widetilde \Gamma_0, \widetilde \Gamma_1)$ is based on the identities $\widehat \Gamma_1 u=\widehat B^{(\tau)}\widehat \Gamma_0 u$ and $\widetilde \Gamma_1 u=\widetilde B^{(\tau)}\widetilde \Gamma_0 u$ ($u\in \dom A_{\max}^{\soft}$) which must yield the same linear set $\dom A_{ B^{(\tau)}(z)}^{(\tau)}$:
\begin{multline*}
\widetilde \Gamma_1 u=\widetilde B^{(\tau)}\widetilde \Gamma_0 u \Leftrightarrow
\P \widehat \Gamma_1 u - \Port \widehat \Gamma_0 u = \widetilde B^{(\tau)} \P \widehat \Gamma_0 u + \widetilde B^{(\tau)} \Port \widehat \Gamma_1 u
\\[0.3em]
\Leftrightarrow
\bigl(\P-\widetilde B^{(\tau)}\Port\bigr)\widehat \Gamma_1 u =\bigl(\Port +\widetilde B^{(\tau)} \P\bigr)\widehat \Gamma_0 u
\\[0.3em]
\Leftrightarrow
\widehat \Gamma_1 u =\bigl(\P - \widetilde B^{(\tau)}\Port\bigr)^{-1}\bigl(\Port+\widetilde B^{(\tau)}\P\bigr) \widehat \Gamma_0u
\end{multline*}
This implies
\begin{equation}\label{eq:Btilde}
\widetilde B^{(\tau)}=\bigl(\P \widehat B^{(\tau)}-\Port\bigr)\bigl(\Port \widehat B^{(\tau)}+\P\bigr)^{-1}.
\end{equation}
An explicit calculation based on \eqref{eq:Btilde} immediately yields
\begin{equation}\label{eq:Beff02}
\widetilde B^{(\tau)}(z)=
\diag\left\{
\frac{a_1 k}{\e}\Biggl[\cot \frac {k \e l^{(1)}}{a_1}-\biggl(\sin \frac {k\e l^{(1)}}{a_1}\biggr)^{-1}\Biggr],-\frac{\e}{a_1 k}\Biggl[\cot \frac {k \e l^{(1)}}{a_1}+\biggl(\sin \frac {k\e l^{(1)}}{a_1}\biggr)^{-1}\Biggr]
\right\}\end{equation}
and hence
\begin{equation}\label{eq:neededfor7}
\widetilde B^{(\tau)}(z)= \begin{pmatrix}
-\dfrac {l^{(1)}}2z&0\\[0.7em]
0&0
\end{pmatrix} + O(\e^2)=:\widetilde B^{(\tau)}_\eff +O(\e^2)
\end{equation}
uniformly in $\tau$ (since the matrix $\widetilde B^{(\tau)}$ does not depend on $\tau$!) and $z\in K_\sigma$. Theorem
\ref{thm:crucial} is now applicable, yielding, by an explicit calculation,
$$
\begin{pmatrix}
\dtau u|_0-\bxit \dtau u|_{l^{(2)}}\\[0.5em]
-(u|_0-\bxit u|_{l^{(2)}})
\end{pmatrix}=
\widetilde B^{(\tau)}_\eff
\begin{pmatrix}
u|_0+ \bxit u|_{l^{(2)}}\\[0.5em]
\dtau u|_0+ \bxit \dtau u|_{l^{(2)}}
\end{pmatrix},
$$
and the claim follows.

\section*{Appendix B. The proof of Lemma \ref{lemma:gen_res_1}}

Expanding the matrix $B^{(\tau)}(z)$ (see \eqref{D2N_1}) into power series with respect to $\e$  yields:
$$
B^{(\tau)}(z)= B^{(\tau)}(0)+O(1)=\frac{1}{\e^2 } \begin{pmatrix} D&\bxit\\[0.5em]
                                                     \xit& D
                                                     \end{pmatrix}+
O(1)=:B_0+O(1),
$$
where we set
\[
\xit=-\frac{a_1^2}{l^{(1)}}\exp\bigl(i\tau(l^{(1)}+l^{(3)})\bigr)-\frac{a_3^2}{l^{(3)}}\exp\bigl(-i\tau l^{(2)}\bigr),\qquad
D=\frac{a_1^2}{l^{(1)}}+\frac{a_3^2}{l^{(3)}}.
\]

As in Appendix A, the matrix function $B^{(\tau)}$ is entire in $K_\sigma$ and also Hermitian for real values of $z$. The matrix $\e^2 B_0$ is a Hermitian matrix possessing two distinct eigenvalues, $\mu(\tau)=D-|\xit|$ and $\mu_\bot(\tau)=D+|\xit|$. The eigenvalue branch that is instrumental for our analysis, $\mu(\tau)$, is singled out by the condition $\mu(0)=0$. We diagonalise the matrix $B_0$ by considering the normalised eigenvectors $\psit=(1/\sqrt{2})(1,-\xit/|\xit|)^\top$ and $\psit_\bot=(1/\sqrt{2})(1,\xit/|\xit|)^\top$ corresponding to the eigenvalues $\mu(\tau)$ and $\mu_\bot(\tau)$, respectively, and the unitary affinity
$$
X=\dfrac 1{\sqrt{2}} \begin{pmatrix}
1&1\\[0.5em]
-\dfrac{\xit}{|\xit|}&\dfrac{\xit}{|\xit|}
\end{pmatrix}.
$$
We also set
$$
\P=\bigl\langle \cdot, \psit\bigr\rangle \psit; \quad
\Port=\bigl\langle \cdot, \psit_\bot\bigr\rangle \psit_\bot.
$$

Passing over to the triple $(\mathcal H, \widehat \Gamma_0,\widehat \Gamma_1)$, where as in Appendix A $\widehat \Gamma_0:=X^* \Gamma_0^\soft$, $\widehat \Gamma_1:=X^* \Gamma_1^\soft$,
the matrix $B^{(\tau)}$ is mapped to
$
\widehat B^{(\tau)}(z)=X^* B^{(\tau)}(z)X.
$
An explicit calculation yields
$$
\widehat B^{(t)}(z)=\dfrac 1{\e}
\begin{pmatrix}
\alpha - \Re \dfrac{\bxit \beta}{|\xit|}&-i \Im \dfrac{\bxit \beta}{|\xit|}\\[1.7em]
i\Im \dfrac{\bxit \beta}{|\xit|}&\alpha - \Re \dfrac{\bxit \beta}{|\xit|}
\end{pmatrix},
$$
where
\begin{align*}
\alpha&=a_1k \cot \dfrac{k\e l^{(1)}}{a_1}+a_3 k \cot \dfrac{k\e l^{(3)}}{a_3},\\[0.7em]
\beta&=-a_1k\exp\bigl(i\tau(l^{(1)}+l^{(3)})\bigr)\biggl(\sin \dfrac{k \e l^{(1)}}{a_1}\biggr)^{-1}-a_3k\exp\bigl(-i\tau l^{(2)}\bigr)\biggl(\sin \dfrac{k \e l^{(3)}}{a_3}\biggr)^{-1}.
\end{align*}

Precisely as in Appendix A, we pass over to the triple $(\mathcal H, \widetilde \Gamma_0, \widetilde \Gamma_1)$, where the boundary operators $\widetilde \Gamma_0$ and $\widetilde \Gamma_1$ are defined as follows:
$$
\widetilde \Gamma_0:=\P \widehat \Gamma_0+\Port \widehat \Gamma_1,\qquad\widetilde \Gamma_1:=\P \widehat \Gamma_1-\Port \widehat \Gamma_0.
$$
The argument leading to the representation  \eqref{eq:Btilde} remains unchanged, which allows to compute the matrix $\widetilde B^{(\tau)}(z)$ as:
$$
\widetilde B^{(\tau)}(z)=
\dfrac{\e}{\alpha + \Re\dfrac{\bxit \beta}{|\xit|}}
\begin{pmatrix}
\dfrac{\alpha^2-|\beta|^2}{\e^2}&-i \e^{-1}\Im \dfrac{\bxit \beta}{|\xit|}\\[1.5em]
i \e^{-1}\Im \dfrac{\bxit \beta}{|\xit|}&-1
\end{pmatrix}
$$
Expanding trigonometric functions into power series with respect to small variable $\e$, we obtain:
\begin{equation}\label{eq:denominator1}
\Re \dfrac{\bxit \beta}{|\xit|}= \dfrac 1{\e} |\xit|+\frac{1}{6}k^2\e|\xit|^{-1}\biggl(a_1^2+a_3^2+\biggl(\dfrac{a_1^2 l^{(3)}}{l^{(1)}}+\dfrac{a_3^2 l^{(1)}}{l^{(3)}}\biggr)\cos\tau\biggr)+O(\e^3),
\end{equation}
where $|\xit|$ is determined from
$$
|\xit|^2= \left( \dfrac{a_1^2}{l^{(1)}} \right)^2+\left( \dfrac{a_3^2}{l^{(3)}} \right)^2+2 \dfrac{a_1^2}{l^{(1)}}\dfrac{a_3^2}{l^{(3)}} \cos \tau.
$$
Although formally $\xit$ could be equal to zero at $\tau=\pm \pi$ (that is, at both endpoints of the domain of quasimomentum) iff $a_1^2/l^{(1)}=a_3^2/l^{(3)}$, one need not worry about this special case. This follows from the fact that the matrix $B^{(\tau)}(z)$ is analytic in $\tau$ and so are its eigenvalues and eigenvectors. Clearly, this also applies to its transformations $\widehat B^{(\tau)}$ and $\widetilde B^{(\tau)}$.

Since the power series expansion of $\alpha$ yields
\begin{equation}\label{eq:denominator2}
\alpha=\dfrac{a_1^2}{\e l^{(1)}}+\dfrac{a_3^2}{\e l^{(3)}}-\e z \dfrac{l^{(1)}+l^{(3)}}{3}+O(\e^3),
\end{equation}
the denominator $\alpha + \Re\dfrac{\bxit \beta}{|\xit|}$ admits an estimate of order $O(1/\e).$

Furthermore,
$$
\Im\dfrac{\bxit \beta}{|\xit|}=\dfrac{\sin\tau}{|\xit|}\Biggl( -\dfrac{a_1^2}{l^{(1)}}\dfrac{a_3 k}{\sin \dfrac{k\e l^{(3)}}{a_3}} +\dfrac{a_3^2}{l^{(3)}}\dfrac{a_1 k}{\sin \dfrac{k\e l^{(1)}}{a_1}}\Biggr)
= \dfrac{\sin\tau}{|\xit|}\left( \dfrac{z \e}{6} \left[\dfrac {a_3^2 l^{(1)}}{l^{(3)}}-\dfrac{a_1^2 l^{(3)}}{l^{(1)}}\right]\right)+ O(\e^3).
$$
It follows that all elements of $\widetilde B^{(\tau)}(z),$ except the element $(1,1),$ admit uniform (with respect to $\tau$ and $z\in K_\sigma$) estimates as $O(\e^2)$, which by an application of Theorem \ref{thm:crucial} allows us to drop them at an expense of $O(\e^2)$.

In order to simplify notation, we will therefore keep using the symbol $\widetilde B^{(\tau)}(z)$ for the matrix
\begin{equation}\label{eq:app_B_complex}
\biggl(\alpha + \Re\dfrac{\bxit \beta}{|\xit|}\biggr)^{-1}
\diag\left\{
\dfrac{\alpha^2-|\beta|^2}{\e^2},0\right\}.
\end{equation}
One has
\begin{multline}\label{eq:nominator}
\alpha^2-|\beta|^2=-(a_1^2+a_3^2)z+2 a_1 a_3 z \left( \cot\dfrac{k\e l^{(1)}}{a_1}\cot\dfrac{k\e l^{(3)}}{a_3}-
\cos \tau\biggl(\sin \dfrac{k\e l^{(1)}}{a_1}\biggr)^{-1}\biggl(\sin \dfrac{k\e l^{(3)}}{a_3}\biggr)^{-1} \right)\\[0.9em]
=\dfrac{2a_1^2 a_3^2}{l^{(1)}l^{(3)}\e^2}(1-\cos \tau)-(a_1^2+a_3^2)z-\dfrac z3 \left(\dfrac{a_1^2l^{(3)}}{l^{(1)}}+\dfrac{a_3^2l^{(1)}}{l^{(3)}}\right)(2+\cos \tau) +O(\e^2).
\end{multline}
Combining \eqref{eq:denominator1} and \eqref{eq:denominator2}, one gets
\begin{align}
\e\biggl(\alpha + \Re\dfrac{\bxit \beta}{|\xit|}\biggr)&=
\dfrac{a_1^2}{ l^{(1)}}+\dfrac{a_3^2}{ l^{(3)}}+\bigl|\xit\bigr|- \dfrac{l^{(1)}+l^{(3)}}{3}\e^2z\nonumber\\[0.6em]
&+\frac{1}{6}z\e^2\bigl|\xit\bigr|^{-1}(a_1^2+a_3^2)+\biggl(\dfrac{a_1^2 l^{(3)}}{l^{(1)}}+\dfrac{a_3^2 l^{(1)}}{l^{(3)}}\biggr)\cos\tau+O(\e^4).\label{eq:denominator}
\end{align}

The first summand on the right hand side of \eqref{eq:nominator} does not permit to apply Theorem \ref{thm:crucial} directly in the triple $(\mathcal H, \widetilde \Gamma_0,\widetilde \Gamma_1)$, as it is non-uniform in $\tau$. The reason for this is clear: in the ($\e$-dependent) regions in $\tau$ where it blows up as $\e\to 0$, the operator considered converges to the decoupling, the domain of which is defined by the condition $\widetilde\Gamma_0 u =0$. In order to obtain a uniform estimate we transform the triple so that the named decoupling corresponds to the condition $\widetilde \Gamma'_1 u=0$.

We redefine the triple as $(\mathcal H, \widetilde \Gamma'_0, \widetilde \Gamma'_1)$, where the boundary operators are
$$
\widetilde\Gamma'_0:=\mathcal{P}^\bot \widetilde \Gamma_0 + \mathcal{P}\widetilde \Gamma_1,
\qquad
\widetilde\Gamma'_1:=\mathcal{P}^\bot \widetilde \Gamma_1 - \mathcal{P}\widetilde\Gamma_0.
$$
Here  $\mathcal{P}^\bot$ and $\mathcal{P}$ are orthogonal projectors
$$
\mathcal{P}=\begin{pmatrix} 1&0\\[0.4em] 0&0\end{pmatrix},
\qquad
\mathcal{P}^\bot=\begin{pmatrix}0&0\\[0.4em]0&1\end{pmatrix}.
$$
Since this choice of a triple is exactly as in \eqref{eq:projector-triple}, where $\mathcal{P}^\bot$ takes place of $\P$ and $\mathcal{P}$ takes place of $\Port$, one has the following representation for $\widetilde B^{(\tau)}$ in the triple $(\mathcal H, \widetilde \Gamma'_0, \widetilde \Gamma'_1)$:
$$
B^{(\tau)\prime}=\bigl(\mathcal{P}^\bot \widetilde B^{(t)}-\mathcal P\bigr)\bigl(\mathcal P \widetilde B^{(t)}+\mathcal{P}^\bot\bigr)^{-1}=
-\diag\bigl\{\delta(\tau,\e),0\bigr\},\qquad
 \delta(t,\e):=\dfrac{\e\biggl(\alpha + \Re\dfrac{\bxit \beta}{|\xit|}\biggr)}{\alpha^2-|\beta|^2}.$$ One has:
$$
\e\biggl(\alpha + \Re\dfrac{\bxit \beta}{|\xit|}\biggr)=m_2+\e^2 m_3+O(\e^4); \quad
\alpha^2-|\beta|^2=m_0\dfrac 1{\e^2}+ m_1+ O(\e^2),
$$
where
\begin{align*}
m_0&=2\dfrac {a_1^2 a_3^2}{l^{(1)} l^{(3)}} (1-\cos \tau),\\[0.7em]
m_1&=-(a_1^2+a_3^2)z-\dfrac z 3\biggl(\frac{a_1^2l^{(3)}}{l^{(1)}}+\frac{a_3^2 l^{(1)}}{l^{(3)}}\biggr)(2+\cos\tau),\\[0.8em]
m_2&=\dfrac {a_1^2}{l^{(1)}}+\dfrac {a_3^2}{l^{(3)}} + |\xit|,\\[0.6em]
m_3&=-\dfrac{l^{(1)}+l^{(3)}}3z + \dfrac z 6\bigl|\xit\bigr|^{-1}\biggl(a_1^2 +a_3^2 +\biggl(\frac{a_1^2l^{(3)}}{l^{(1)}}+\frac{a_3^2 l^{(1)}}{l^{(3)}}\biggr)\cos \tau\biggr),
\end{align*}
and
$$
\bigl|\xit\bigr|=\sqrt{\biggl(\frac{a_1^2}{l^{(1)}}\biggr)^2+\biggl(\frac{a_3^2}{l^{(3)}}\biggl)^2+2\frac{a_1^2a_3^2}{l^{(1)}l^{(3)}}\cos\tau}.
$$
Therefore, $\delta(\tau,\e)$ admits the form
$$
\delta(\tau,\e)=\dfrac{m_2+\e^2 m_3+O(\e^4)}{m_0/\e^2+m_1+O(\e^2)},
$$
where the leading coefficient $m_0$ of the denominator is non-uniform in $\tau$ as a function of $\e$.

By a straightforward computation which we skip for the sake of brevity, we arrive at the uniform in $\tau$ and $z\in K_\sigma$ estimate
$$
\delta(\tau,\e)-\dfrac{2\biggl(\dfrac{a_1^2}{l^{(1)}}+\dfrac{a_3^2}{l^{(3)}}\biggr)}{\dfrac{a_1^2 a_3^2}{l^{(1)} l^{(3)}}\biggl(\dfrac{\tau}{\e}\biggr)^2-\bigl(l^{(1)}+l^{(3)}\bigr)\biggl(\dfrac{a_1^2}{l^{(1)}}+\dfrac{a_3^2}{l^{(3)}}\biggr)z}=O(\e^2).
$$

Finally, returning to the triple $(\mathcal H, \widetilde \Gamma_0, \widetilde \Gamma_1)$, we have thus obtained the effective matrix $\widetilde B^{(\tau)}_\eff$. This is:
$$
\widetilde B^{(\tau)}_\eff=\diag\left\{
\dfrac{\dfrac{a_1^2 a_3^2}{l^{(1)} l^{(3)}}\biggl(\dfrac{\tau}{\e}\biggr)^2-\bigl(l^{(1)}+l^{(3)}\bigr)\biggl(\dfrac{a_1^2}{l^{(1)}}+\dfrac{a_3^2}{l^{(3)}}\biggr)z}{2\biggl(\dfrac{a_1^2}{l^{(1)}}+\dfrac{a_3^2}{l^{(3)}}\biggr)}, 0\right\},
$$
leading by an application of Theorem \ref{thm:crucial} to the norm-resolvent estimate with an error bound of the order $O(\e^2)$ and thus completing the proof of Lemma \ref{lemma:gen_res_1}.

\section*{Appendix C. Calculation of the effective dispersion function $K$ in specific examples}

\subsection*{Example (0)}

Here we have $w_\tau={\xi}^{(\tau)}=\exp({\rm i}l^{(1)}\tau),$ $\sigma=0,$ $\rho=\sqrt{l^{(1)}}$  in (\ref{Gamma}), and therefore
we obtain
\begin{equation}
\langle v, \varphi_j\rangle=\frac{\sqrt{2l^{(2)}}}{\pi j}\bigl((-1)^{j+1}\exp({\rm i}\tau)+1\bigr),\qquad
\Gamma_\tau\left(\begin{matrix}\varphi_j\\[0.2em] 0\end{matrix}\right)=\sqrt{\frac{2}{l^{(2)}}}\frac{\pi j}{l^{(2)}}\bigl((-1)^j\exp(-{\rm i}\tau)-1\bigr),\qquad j=1,2,3,\dots,
\label{aux1}
\end{equation}
\begin{equation}
\Gamma_\tau\left(\begin{matrix}v\\[0.2em] \rho\end{matrix}\right)=\frac{2}{l^{(2)}}(1-\cos\tau).
\label{aux2}
\end{equation}
Substituting the above expressions into the general formula (\ref{K_general1}), we obtain
\[
K(\tau, z)=\frac{2}{l^{(1)}l^{(2)}}\biggl(-2z\sum_{j=1}^\infty\dfrac{1-(-1)^j\cos\tau}{(\pi j/l^{(2)})^2-z}-\cos\tau+1\biggr).
\]
Finally, using the formulae, see {\it e.g.} \cite[p.\,48]{Gradshteyn_Ryzhik},
\begin{equation}
\sum_{j=1}^\infty\frac{1}{(\pi j)^2-x^2}=\frac{1}{2}\biggl(\frac{1}{x^2}-\frac{\cos x}{x\sin x}\biggr),\qquad\qquad
\sum_{j=1}^\infty\frac{(-1)^j}{(\pi j)^2-x^2}=\frac{1}{2}\biggl(\frac{1}{x^2}-\frac{1}{x\sin x}\biggr),\qquad x\neq \pi{\mathbb Z},
\label{RG2}
\end{equation}
yields
\begin{equation}
K(\tau, z)=\frac{2\sqrt{z}\bigl(\cos(l^{(2)}\sqrt{z})-\cos\tau\bigr)}{l^{(1)}\sin(l^{(2)}\sqrt{z})}.
\label{K_example0}
\end{equation}


\subsection*{Example (1)}



In this case
\begin{align*}
w_\tau&=-\frac{\xi^{(\tau)}}{|\xi^{(\tau)}|},\quad \xi^{(\tau)}=-\frac{a_1^2}{l^{(1)}}\exp\bigl(i\tau(l^{(1)}+l^{(3)})\bigr)-\frac{a_3^2}{l^{(3)}}\exp(-i\tau l^{(2)}),\\[0.7em]
\sigma^2&=\biggl(\dfrac{l^{(1)}}{a_1^2}+\dfrac{l^{(3)}}{a_3^2}\biggr)^{-1},\qquad\rho=\sqrt{l^{(1)}+l^{(3)}},
\end{align*}
so that
\[
\Gamma_\tau\left(\begin{matrix}
\varphi_j\\[0.3em] 0\end{matrix}\right)=-\sqrt{\frac{2}{l^{(2)}}}\frac{\pi j}{l^{(2)}}\bigl((-1)^{j+1}\overline{\theta(\tau)}+1\bigr),\quad\quad
\langle v, \phi_j\rangle=\frac{\sqrt{2l^{(2)}}}{\pi j}\bigl((-1)^{j+1}\theta(\tau)+1\bigr),\quad\quad j=1,2,3,\dots,
\]

\[
\Gamma_\tau\left(\begin{matrix}v\\[0.2em] \rho\end{matrix}\right)
=\frac{2}{l^{(2)}}\bigl(1-\Re\theta(\tau)\bigr)+\biggl(\frac{\sigma\tau}{\varepsilon}\biggr)^2,\qquad \theta(\tau):=\frac{\dfrac{a_1^2}{l^{(1)}}{\rm e}^{-i\tau}+\dfrac{a_3^2}{l^{(3)}}}{\biggl|\dfrac{a_1^2}{l^{(1)}}{\rm e}^{-i\tau}+\dfrac{a_3^2}{l^{(3)}}\biggr|}.
\]

Substituting the above expressions into (\ref{K_general1}) and making use of the formulae (\ref{RG2}) again,
we obtain
\begin{equation}
K(\tau, z)=\frac{1}{l^{(1)}+l^{(3)}}\Biggl\{2\sqrt{z}\biggl(
\cot(l^{(2)}\sqrt{z})
-\frac{\Re\theta(\tau)}{\sin(l^{(2)}\sqrt{z})}\biggr)+\biggl(\frac{\sigma\tau}{\varepsilon}\biggr)^2\Biggr\}.
\label{K_example1}
\end{equation}


\subsection*{Example (2)} The analysis is similar to the case of Example (0) and is based on the formulae
\[
\langle v_1, \varphi_j^{(1)}\rangle=\xi_1^{(\tau)}\frac{\sqrt{2l^{(1)}}}{\pi j}\bigl((-1)^{j+1}\exp({\rm i}\tau)+1\bigr),
\]
\[
\Gamma_\tau\left(\begin{matrix}\chi^{(1)}\varphi_j^{(1)}\\[0.2em] 0\end{matrix}\right)=a_1^2\sqrt{\frac{2}{l^{(1)}}}\frac{\pi j}{l^{(1)}}\bigl((-1)^j-\exp(i\tau)\bigr)\exp(-{\rm i}\tau l^{(1)}),\qquad j=1,2,3,\dots,
\]
\[
\Gamma_\tau\left(\begin{matrix}\chi^{(1)}v_1\\[0.2em] \rho\end{matrix}\right)=\frac{a_1^2\xi_1^{(\tau)}}{l^{(1)}}\bigl(\exp(-i\tau)-1\bigr)\bigl(\exp(i\tau l^{(1)})-1\bigr),
\]
\[
\langle v_2, \varphi_j^{(2)}\rangle=\frac{\sqrt{2l^{(2)}}}{\pi j}\bigl((-1)^{j+1}+1\bigr),\qquad
\Gamma_\tau\left(\begin{matrix}\chi^{(2)}\varphi_j^{(2)}\\[0.2em] 0\end{matrix}\right)=a_2^2\sqrt{\frac{2}{l^{(2)}}}\frac{\pi j}{l^{(2)}}\bigl((-1)^j-1\bigr),\qquad j=1,2,3,\dots,
\]
\[
\Gamma_\tau\left(\begin{matrix}\chi^{(2)}v_2\\[0.2em] \rho\end{matrix}\right)=0.
\]
Substituting these into (\ref{K_ex2}) and using (\ref{RG2}) yields
\begin{align*}
K(\tau, z)&=\frac{1}{l^{(3)}}\biggl\{-4za_1^2l^{(1)}\sum_{j=1}^\infty\frac{1+(-1)^{j+1}\cos\tau}{(\pi j)^2-(l^{(1)}\sqrt{z})^2}+\frac{2a_1^2}{l^{(1)}}(1-\cos\tau)
-4za_2^2l^{(2)}\sum_{j=1}^\infty\frac{1+(-1)^{j+1}}{(\pi j)^2-(l^{(2)}\sqrt{z})^2}\biggr\}\\[0.6em]
&=\frac{1}{l^{(3)}}\Biggl\{-2za_1^2l^{(1)}\biggl(\frac{1}{(l^{(1)}\sqrt{z})^2}-\frac{\cos(l^{(1)}\sqrt{z})}{l^{(1)}\sqrt{z}\sin(l^{(1)}\sqrt{z})}\biggr)\\[0.6em]
&\ \ \ \ \ \ \ \ \ \ +2za_1^2l^{(1)}\cos\tau\biggl(\frac{1}{(l^{(1)}\sqrt{z})^2}-\frac{1}{l^{(1)}\sqrt{z}\sin(l^{(1)}\sqrt{z})}\biggr)
+\frac{2a_1^2}{l^{(1)}}(1-\cos\tau)\\[0.6em]
&\ \ \ \ \ \ \ \ \ \ -2za_2^2l^{(2)}\biggl(\frac{1}{(l^{(2)}\sqrt{z})^2}-\frac{\cos(l^{(2)}\sqrt{z})}{l^{(2)}\sqrt{z}\sin(l^{(2)}\sqrt{z})}\biggr)
+2za_2^2l^{(2)}\biggl(\frac{1}{(l^{(2)}\sqrt{z})^2}-\frac{1}{l^{(2)}\sqrt{z}\sin(l^{(2)}\sqrt{z})}\biggr)\Biggr\}\\[0.6em]
&=\frac{2\sqrt{z}}{l^{(3)}}\Biggl\{a_1^2\frac{\cos(l^{(1)}\sqrt{z})-\cos\tau}{\sin(l^{(1)}\sqrt{z})}-a_2^2\tan\biggl(\frac{l^{(2)}\sqrt{z}}{2}\biggr)\Biggr\}.
\end{align*}

\section*{Acknowledgements}
KDC and YYE is grateful for the financial support of
the Engineering and Physical Sciences Research Council: Grant EP/L018802/2 ``Mathematical foundations of metamaterials: homogenisation, dissipation and operator theory''. SNN and YYE have been partially supported by the  RFBR grant 19-01-00657-a. SNN also gratefully acknowledges funding provided by Knut and Alice Wallenberg Foundation (program for mathematics 2018). AVK has been partially supported by the Russian Federation Government megagrant 14.Y26.31.0013.


\begin{thebibliography}{99}

\bibitem{AdamyanPavlov}
Adamyan, V. M.; Pavlov, B. S., 1986. Zero-radius potentials and M. G. Kre\u\i n's formula for generalized resolvents.
 {\em J. Soviet Math.} {\bf 42}(2), 1537--1550.

\bibitem{Lions}
Bensoussan, A., Lions, J.-L., Papanicolaou, G., 1978. {\it Asymptotic Analysis for Periodic Structures,} North Holland.

\bibitem{Kuchment2} Berkolaiko, G., Kuchment, P., 2012. \emph{Introduction to Quantum
Graphs}, Mathematical Surveys and Monographs \textbf{186}, American Mathematical Society.


\bibitem{Birman} Birman, M. Sh., 1956. On the self-adjoint extensions of positive definite operators. {\it Math. Sb.}
{\bf 38}, 431--450.

\bibitem{BirmanSuslina} Birman, M. Sh., Suslina, T. A.,  2004. Second order periodic differential operators. Threshold properties and homogenisation. {\it St. Petersburg. Math. J.}
{\bf 15}(5), 639--714.

\bibitem{BirmanSuslina_corr} Birman, M. Sh., Suslina, T. A., 2006. Homogenization with corrector term for periodic elliptic differential operators.
{\it St. Petersburg Math. J.} {\bf 17}(6), 897--973.


\bibitem{BouchitteFelbacq}
Bouchitt\'{e}, G., Felbacq, D., 2004. Homogenisation near resonances and artificial magnetism from dielectrics. {\it C. R. Math. Acad. Sci.
Paris} {\bf 339}(5), 377--382.

\bibitem{Phys_book} Capolino, F., 2009. {\it Theory and Phenomena of Metamaterials.} Taylor \& Francis.

\bibitem{Physics} Cherednichenko K. D., Ershova Yu., and Kiselev A. V., 2019. Time-Dispersive Behaviour as a Feature of Critical Contrast Media, \emph{SIAM J. Appl. Math.} \textbf{79}(2), 690--715.

\bibitem{PDE_paper} Cherednichenko K. D., Ershova Yu., and Kiselev A. V., 2019. Effective behaviour of critical-contrast PDEs: micro-resonances, frequency conversion, and time dispersive properties. I. \href{https://arxiv.org/abs/1808.03961}{arXiv: 1808.03961}.


\bibitem{CherKis}
Cherednichenko, K. D., Kiselev, A. V., 2017. Norm-resolvent convergence of one-dimensional high-contrast periodic problems to a Kronig-Penney dipole-type model. {\it Comm. Math. Phys.} {\bf 349}(2), 441--480.


\bibitem{CherKisSilva} Cherednichenko, K. D., Kiselev, A. V., Silva, L. O., 2017. Functional model for extensions of symmetric operators and applications to scattering theory. {\it Networks and Heterogeneous Media} {\bf 13}(2), 191--215.


\bibitem{Datta} Datta, S., 1995. \emph{Electronic transport in mesoscopic systems}. Cambridge University Press.


\bibitem{DM} Derkach, V. A., Malamud M. M., 1991. Generalised resolvents and
the boundary value problems for Hermitian operators with gaps, {\it J. Funct. Anal. } \textbf{95}, 1--95.

\bibitem{Yorzh3} Ershova, Yu., Karpenko, I. I., Kiselev, A. V., 2016. Isospectrality for graph Laplacians under the change of
coupling at graph vertices, {\it J. Spectral Th.}  \textbf{6}(1), 43--66.

\bibitem{Exner}
Exner, P., Post, O., 2005. Convergence of spectra of graph-like thin manifolds. {\it J. Geom. Phys.} \textbf{54}(1), 77--115.




\bibitem{Figotin_Schenker_2005}
Figotin, A., Schenker, J. H., 2005. Spectral analysis of time dispersive and dissipative systems, {\it Journal of Statistical Physics,} {\bf 118}(1--2), 199--263.

\bibitem{Figotin_Schenker_2007b}
Figotin, A., Schenker, J. H., 2007. Hamiltonian structure for dispersive and dissipative dynamical systems.
{\it Journal of Statistical Physics} {\bf 128}(4), 969--1056


\bibitem{Friedlander}  Friedlander, L., 2002. On the density of states of periodic media in the large coupling limit. {\it Comm. Partial Differential Equations,} {\bf 27}(1-2), 355--380.




\bibitem{Gor}Gorbachuk, V. I., Gorbachuk, M. L., 1991. {\it Boundary value
problems for operator differential equations.}
Mathematics and its
Applications (Soviet Series) \textbf{48},  Kluwer Academic
Publishers Group, Dordrecht.

\bibitem{Gradshteyn_Ryzhik}
Gradshteyn, I. S., Ryzhik, I. M., 2007. {\it Table of Integrals, Series, and Products,} Academic Press.

\bibitem{HempelLienau_2000} Hempel, R., Lienau, K., 2000. Spectral Properties of the Periodic Media in
Large Coupling Limit. \emph{Comm. Partial Differential Equations}  \textbf{25} (7--8), 1445--1470.

\bibitem{Hoermander} H\"{o}rmander, L., 2003. {\it The Analysis of Linear Partial Differential Operators III. Pseudo-Differential Operators,} Springer.

\bibitem{Jikov_book} Jikov, V. V., Kozlov, S. M., Oleinik, O. A., 1994. {\it Homogenisation of Differential Operators and Integral Functionals,} Springer.

\bibitem{Ko1}
Ko\v cube\u\i\ A. N., 1975. On extension of symmetric operators
and symmetric binary  relations, {\it Math. Notes} \textbf{17}, 41--48.

\bibitem{Koch}
Ko\v cube\u\i\ A. N., 1980. Characteristic functions of symmetric operators
and their extensions (Russian), {\it Izv. Akad. Nauk Arm. SSR Ser.
Mat.} \textbf{15}(3), 219--232.

\bibitem{Krein} Kre\u\i n, M. G., 1947. Theory of self-adjoint extensions of semibounded Hermitian operators and
applications II. {\it Mat. Sb.} {\bf 21}(3), 365--404.

\bibitem{KuchmentZeng}
 Kuchment, P., Zeng, H., 2001. Convergence of spectra of mesoscopic systems collapsing onto a graph. {\it J. Math. Anal. Appl.}
 {\bf 258}(2), 671--700.

\bibitem{KuchmentZeng2004}
 Kuchment, P.,  Zeng, H., 2004. Asymptotics of spectra of Neumann Laplacians in thin domains. {\it Contemporary Mathematics} {\bf 327},
Amer. Math. Soc., Providence, Rhode Island, 199--213.

\bibitem{LP}
Lax, P. D., Phillips, R. S., 2010. {\it Scattering Theory}. Pure and Applied Mathematics {\bf 26}. Academic Press, New York-London.

\bibitem{Leonhardt}
Leonhardt, U., Philbin, T., 2010. {\it Geometry and Light. The Science of Invisibility.} Dover Publications.



\bibitem{Milton_et_al}
Milton, G. W., Nicorovici, N.-A. P., McPhedran, R. C., Cherednichenko, K. D., Jacob, Z., 2008. Solutions in folded geometries and associated cloaking due to anomalous resonance. New Journal of Physics {\bf 10}(11), 115021.

\bibitem{Naimark1940}
Neumark, M., 1940. Spectral functions of a symmetric operator. (Russian) {\it  Bull. Acad. Sci. URSS. Ser. Math. [Izvestia Akad. Nauk SSSR]} {\bf 4}, 277--318.

\bibitem{Naimark1943}
Neumark, M., 1943 Positive definite operator functions on a commutative group. (Russian) {\it Bull. Acad. Sci. URSS Ser. Math. [Izvestia Akad. Nauk SSSR]} {\bf 7}, 237--244.

\bibitem{Ryzhov}
Ryzhov, V., 2007. Functional model of a class of nonselfadjoint extensions of symmetric operators.
{\it Oper. Theory Adv. Appl.} {\bf 174}, Birkh\"{a}user, Basel, 117--158.

\bibitem{Ryzhov_later} Ryzhov, V., 2009. Weyl-Titchmarsh function of an abstract boundary value problem, operator colligations, and linear systems with boundary control. {\it Complex Anal. Oper. Theory} {\bf 3}(1), 289--322.

\bibitem{Fuerer} Schur, I., 1905. {\it Neue Begr\"undung der Theorie der Gruppencharaktere Sitzungsberichte der Preussischen Akademie der Wissenschaften,} Physikalisch-Mathematische Klasse, 406--436.

\bibitem{Simon}
Simon, B. 1979. {\it Functional Integration and Quantum Physics,} American Mathematical Society.

\bibitem{Shkalikov} Shkalikov, A. A., 2004. 
Spectral portraits of the Orr-Sommerfeld operator at large Reynolds   numbers.
\emph{J. Math. Sci. (N.Y.)}
\textbf{124} (6), 5417--5441.

\bibitem{Strauss} Strauss A. V., 1954. Generalised resolvents of symmetric operators. (Russian) \emph{Izv. Akad. Nauk SSSR, Ser. Mat.,} \textbf{18}, 51--86.

\bibitem{Tip_1998}
Tip, A., 1998. Linear absorptive dielectrics. {\it Phys. Rev. A} {\bf 57}: 4818--4841.

\bibitem{Tip_2006}
Tip, A., 2006. Some mathematical properties of Maxwell's equations for macroscopic dielectrics. {\it J. Math. Phys.} {\bf 47}, 012902.


\bibitem{Tutte} Tutte, W. T., 1984. \emph{Graph theory}.
 Encyclopedia of Mathematics and its Applications {\bf 21}. Addison-Wesley Publishing Company,
 Reading, MA.


\bibitem{Veselago} Veselago, V. G., 1968. The electrodynamics of substances with simultaneously negative values of $\e$ and $\mu$. {\it Soviet Physics Uspekhi.} {\bf 10}(4): 509--514.






\bibitem{Vishik} Vi\v sik, M.I., 1952. On general boundary problems for elliptic differential equations (Russian). {\it Trudy
Moskov. Mat. Ob\v sc.} {\bf 1}, 187--246.

\bibitem{Yafaev}
Yafaev D.~R., 1992.
\newblock {\em Mathematical scattering theory}, Translations
  of Mathematical Monographs, {\bf 105}.
\newblock American Mathematical Society, Providence, RI.

\bibitem{Zhikov_1989}
Zhikov, V. V., 1989. Spectral approach to asymptotic diffusion problems (Russian). {\it Differentsial'nye uravneniya}
{\bf 25}(1), 44--50.


\bibitem{Zhikov2000}
Zhikov, V. V., 2000. On an extension of the method of two-scale convergence and its applications, {\it Sbornik: Mathematics} {\bf 191}(7), 973--1014.

\bibitem{Zhikov_singular_structures}
Zhikov, V. V., 2002. Averaging of problems in the theory of elasticity on singular structures.
{\it Izv. Math.} {\bf 66}(2), 299--365.

\bibitem{Zhikov2005}
Zhikov, V. V., 2005. On gaps in the spectrum of some divergence elliptic operators with periodic coefficients. {\it St. Petersburg Math. J.} {\bf 16}(5), 773--719.


\end{thebibliography}
\end{document}